\pdfoutput=1 
\PassOptionsToPackage{unicode}{hyperref}
\PassOptionsToPackage{hyphens}{url}
\PassOptionsToPackage{dvipsnames,svgnames,x11names}{xcolor}
\documentclass[
  12pt]{article}

\usepackage{amsmath,amssymb}
\usepackage{iftex}
\ifPDFTeX
  \usepackage[T1]{fontenc}
  \usepackage[utf8]{inputenc}
  \usepackage{textcomp} 
\else 
  \usepackage{unicode-math}
  \defaultfontfeatures{Scale=MatchLowercase}
  \defaultfontfeatures[\rmfamily]{Ligatures=TeX,Scale=1}
\fi
\usepackage{lmodern}
\ifPDFTeX\else
\fi
\IfFileExists{upquote.sty}{\usepackage{upquote}}{}
\IfFileExists{microtype.sty}{
  \usepackage[]{microtype}
  \UseMicrotypeSet[protrusion]{basicmath} 
}{}
\makeatletter
\@ifundefined{KOMAClassName}{
  \IfFileExists{parskip.sty}{%
    \usepackage{parskip}
  }{
    \setlength{\parindent}{0pt}
    \setlength{\parskip}{6pt plus 2pt minus 1pt}}
}{
  \KOMAoptions{parskip=half}}
\makeatother
\usepackage{xcolor}
\makeatletter
\ifx\paragraph\undefined\else
  \let\oldparagraph\paragraph
  \renewcommand{\paragraph}{
    \@ifstar
      \xxxParagraphStar
      \xxxParagraphNoStar
  }
  \newcommand{\xxxParagraphStar}[1]{\oldparagraph*{#1}\mbox{}}
  \newcommand{\xxxParagraphNoStar}[1]{\oldparagraph{#1}\mbox{}}
\fi
\ifx\subparagraph\undefined\else
  \let\oldsubparagraph\subparagraph
  \renewcommand{\subparagraph}{
    \@ifstar
      \xxxSubParagraphStar
      \xxxSubParagraphNoStar
  }
  \newcommand{\xxxSubParagraphStar}[1]{\oldsubparagraph*{#1}\mbox{}}
  \newcommand{\xxxSubParagraphNoStar}[1]{\oldsubparagraph{#1}\mbox{}}
\fi
\makeatother

\usepackage{longtable,booktabs,array}
\usepackage{calc} 
\usepackage{etoolbox}
\makeatletter
\patchcmd\longtable{\par}{\if@noskipsec\mbox{}\fi\par}{}{}
\makeatother
\IfFileExists{footnotehyper.sty}{\usepackage{footnotehyper}}{\usepackage{footnote}}
\makesavenoteenv{longtable}
\usepackage{graphicx}
\makeatletter
\def\maxwidth{\ifdim\Gin@nat@width>\linewidth\linewidth\else\Gin@nat@width\fi}
\def\maxheight{\ifdim\Gin@nat@height>\textheight\textheight\else\Gin@nat@height\fi}
\makeatother
\setkeys{Gin}{width=\maxwidth,height=\maxheight,keepaspectratio}
\makeatletter
\def\fps@figure{htbp}
\makeatother

\makeatletter
\@ifpackageloaded{caption}{}{\usepackage{caption}}
\AtBeginDocument{%
\ifdefined\contentsname
  \renewcommand*\contentsname{Table of contents}
\else
  \newcommand\contentsname{Table of contents}
\fi
\ifdefined\listfigurename
  \renewcommand*\listfigurename{List of Figures}
\else
  \newcommand\listfigurename{List of Figures}
\fi
\ifdefined\listtablename
  \renewcommand*\listtablename{List of Tables}
\else
  \newcommand\listtablename{List of Tables}
\fi
\ifdefined\figurename
  \renewcommand*\figurename{Figure}
\else
  \newcommand\figurename{Figure}
\fi
\ifdefined\tablename
  \renewcommand*\tablename{Table}
\else
  \newcommand\tablename{Table}
\fi
}
\@ifpackageloaded{float}{}{\usepackage{float}}
\floatstyle{ruled}
\@ifundefined{c@chapter}{\newfloat{codelisting}{h}{lop}}{\newfloat{codelisting}{h}{lop}[chapter]}
\floatname{codelisting}{Listing}

\makeatother
\makeatletter
\@ifpackageloaded{caption}{}{\usepackage{caption}}
\@ifpackageloaded{subcaption}{}{\usepackage{subcaption}}
\makeatother

\ifLuaTeX
  \usepackage{selnolig}  
\fi
\usepackage[]{natbib}
\usepackage{bookmark}

\IfFileExists{xurl.sty}{\usepackage{xurl}}{} 
\hypersetup{
  pdftitle={Title},
  pdfauthor={Author 1; Author 2},
  pdfkeywords={3 to 6 keywords, that do not appear in the title},
  colorlinks=true,
  linkcolor={blue},
  filecolor={Maroon},
  citecolor={Blue},
  urlcolor={Blue},
  pdfcreator={LaTeX via pandoc}}

\usepackage{setspace}
\usepackage{enumitem}
\usepackage{url}
\usepackage{float}
\usepackage{booktabs}
\usepackage{tikz}
\usetikzlibrary{shapes.geometric,arrows.meta,positioning}
\tikzset{
  boxW/.store in=\boxW, boxW = 17.5cm,
  boxH/.store in=\boxH, boxH=7.5cm,
  diaW/.store in=\diaW, diaW=3.5cm,
  diaH/.store in=\diaH, diaH = 1.5cm,
  block/.style   = {rectangle, rounded corners = 1.2mm, draw=black, very thick,
                    align = center, inner sep=5pt, fill=white,
                    text width=\boxW, minimum height=\boxH},
  decision/.style= {diamond, draw=black, very thick, align = center, fill=white,
                    inner ysep=-5pt, text width=\diaW,
                    minimum width=\diaW, minimum height=\diaH},
  result/.style  = {rectangle, rounded corners = 1.2mm, draw=black, very thick,
                    align=left, inner sep=5pt, fill=white,
                    text width=\boxW, minimum height=\boxH},
  arrow/.style   = {-{Latex[length=2mm]}, very thick},
  tinycap/.style = {font=\footnotesize}
}

\PassOptionsToPackage{table}{xcolor}
\usepackage{tabularray}
\usepackage{colortbl}
\usepackage{tcolorbox}
\usepackage{amsfonts,amsthm}
\usepackage{mathtools}
\usepackage{bbm}
\usepackage{bm}
\usepackage{csquotes}
\usepackage[noabbrev,capitalise]{cleveref}
\usepackage{xr-hyper}
\usepackage{listings}

\newtheoremstyle{newstyle}
{10pt} 
{10pt} 
{\itshape\onehalfspacing} 
{} 
{\bfseries} 
{.} 
{ } 
{} 

\theoremstyle{newstyle}

\newtheorem{prop}{Proposition}

\newtheorem{lem}{Lemma}
\newtheorem{cor}{Corollary}

\newtheorem{assm}{Assumption}
\newtheorem{rmk}{Remark}

\DeclareMathOperator{\E}{\rm{E}}
\DeclareMathOperator{\R}{\mathbbm{R}}

\DeclareMathOperator{\Var}{\rm{Var}}

\DeclareMathOperator*{\argmin}{arg\,min}

\DeclarePairedDelimiter\abs{\lvert}{\rvert}
\newcommand{\given}[1][]{\;#1\mid\;}

\allowdisplaybreaks

\usepackage[compact,bottomtitles]{titlesec}
\titlespacing*{\section}{0pt}{.8ex plus 0.1ex minus 1ex}{0.5ex plus 0.1ex minus 1ex}
\titlespacing*{\subsection}{0pt}{.6ex plus 0.1ex minus 1ex}{0.4ex plus 0.1ex minus 1ex}
\titlespacing*{\subsubsection}{0pt}{.4ex plus 0.1ex minus 1ex}{0.3ex plus 0.1ex minus 1ex}
\newcommand{\anon}{1}

\begin{document}

\def\spacingset#1{\renewcommand{\baselinestretch}%
{#1}\small\normalsize} \spacingset{1}

\providecommand{\AlphaLevel}{0.05}
\providecommand{\PctSectorMissing}{34}
\providecommand{\PctBeatMissing}{80}
\providecommand{\BaselineTauPP}{2.30}
\providecommand{\BaselineSE}{0.0026}
\providecommand{\GammaRhoZeroInsig}{1.06}
\providecommand{\RhoMinOffsetAtGammaStar}{0.05}
\providecommand{\PlateauGammaMin}{1.06}
\providecommand{\PlateauGammaMax}{1.19}
\providecommand{\PlateauRho}{0.05}
\providecommand{\GammaAllInsigMin}{1.33}
\providecommand{\GammaHighestRhoSigMax}{1.23}
\providecommand{\RhoSigMaxHighest}{0.80}
\providecommand{\RhoNonsigMinHighest}{0.85}
\providecommand{\GammaLowestRhoSigMax}{1.32}
\providecommand{\RhoSigMaxLowest}{0.35}
\providecommand{\RhoNonsigMinLowest}{0.40}
\providecommand{\ChangePointLower}{1.32}
\providecommand{\ChangePointUpper}{1.33}
\providecommand{\GeoChangePointXiZeroRhoLower}{1.37}
\providecommand{\GeoCdfAtChangePointXiZeroRhoLower}{23.2}
\providecommand{\GeoCdfComplementXiZeroRhoLower}{77}
\providecommand{\GeoChangePointXiZeroRhoUpper}{1.36}
\providecommand{\GeoCdfAtChangePointXiZeroRhoUpper}{23.1}
\providecommand{\GeoCdfComplementXiZeroRhoUpper}{77}
\providecommand{\GeoChangePointXiTwentyFiveRhoLower}{1.37}
\providecommand{\GeoCdfAtChangePointXiTwentyFiveRhoLower}{28.2}
\providecommand{\GeoCdfComplementXiTwentyFiveRhoLower}{72}
\providecommand{\GeoChangePointXiTwentyFiveRhoUpper}{1.37}
\providecommand{\GeoCdfAtChangePointXiTwentyFiveRhoUpper}{28.0}
\providecommand{\GeoCdfComplementXiTwentyFiveRhoUpper}{72}


\if1\anon
{
  \title{\bf Sequential Sensitivity Analysis for Multiple Assumptions: A Framework for Understanding Racial Disparity in Police Use of Force}
  \author{Thomas Leavitt\hspace{.2cm}\\
    Marxe School of Public and International Affairs \\ Baruch College, City University of New York (CUNY)\\
    and \\
    Jake Bowers \\
    Political Science and Statistics\\University of Illinois at Urbana-Champaign\\
    and \\
    Luke Miratrix \\
    Graduate School of Education and Statistics\\Harvard University}
  \maketitle
} \fi

\if0\anon
{
  \bigskip
  \bigskip
  \bigskip
  \begin{center}
    {\LARGE\bf Sequential Sensitivity Analysis for Multiple Assumptions: A Framework for Understanding Racial Disparity in Police Use of Force}
\end{center}
  \medskip
} \fi

\begin{abstract}
\noindent Inferring racial discrimination in police use of force --- the average causal effect of civilian race on use of force --- requires two assumptions about policing prior to potential use of force: that officers do not discriminate in whom they would stop (no discrimination in stops) and that, conditional on patrol context, the probability that an encounter is with a minority rather than a white civilian does not vary across encounters (no bias in encounters). As \citet{knoxetal2020} show, violations of the first can mask racial disparity in force. Whether it reflects discrimination in force also depends on the second. Existing sensitivity analyses address one assumption at a time. We develop a framework that varies both sequentially and apply it to NYPD Stop, Question, and Frisk data (2003--2013). Under plausible levels of discrimination in stops, we find substantial racial \textit{disparity} in force. However, the conclusion that this disparity reflects \textit{discrimination} is fragile to modest departures from no bias in encounters that census-based calibration suggests are demographically feasible. By jointly addressing both confounding channels, the framework reveals how they interact in ways that separate analyses cannot, contributing to understanding what generates racial disparities and how they might be addressed.
\end{abstract}


\newpage
\spacingset{1.8} 
\setlength{\abovedisplayskip}{6pt plus 2pt minus 2pt}
\setlength{\belowdisplayskip}{6pt plus 2pt minus 2pt}
\setlength{\abovedisplayshortskip}{4pt plus 2pt minus 2pt}
\setlength{\belowdisplayshortskip}{4pt plus 2pt minus 2pt}

\section{Introduction}

What explains racial disparity in police use of force? Two mechanisms compete. The first is racial discrimination: an officer would use force against a minority civilian but not against a white civilian in an otherwise identical encounter. The second is variation across officers in their chances of encountering a minority civilian, variation that may be associated with how readily officers use force. Officers move through space, and the probability of encountering a minority civilian depends on the patrol context and the officer's choices within that context.

\citet{fryer2018,fryer2019} finds little racial disparity in use of force after conditioning on patrol-context features, concluding that racial discrimination in force is negligible. This inference rests on two assumptions. The first --- which we call No-Bias-in-Encounters --- holds that within the same context, each encounter has equal probability of being with a minority civilian rather than a white civilian, that \enquote{civilian race is `as good as randomly assigned'} \citep[][p.~229]{fryer2018}. The second is that officers do not discriminate in whom they stop.

Police administrative data include only encounters that escalate to a stop, so even under No-Bias-in-Encounters, racial discrimination in whom officers would stop can hide racial disparity in use of force. \citet{knoxetal2020} develop a sensitivity analysis to reason about the consequences of the second assumption: holding No-Bias-in-Encounters fixed, their framework varies the assumed level of discrimination in stops and traces how conclusions about discrimination in force respond. At plausible levels of discrimination in stops, the analysis reveals a substantial racial disparity in force, supporting the conclusion that under No-Bias-in-Encounters officers would use force at a higher rate against minority civilians than against white civilians in the same encounters.

But No-Bias-in-Encounters is itself a strong assumption that police data cannot verify and that is unlikely to hold exactly. Departures from it can produce the same disparity in use of force without discrimination by officers. The two confounding channels --- sample selection and encounter assignment --- have been studied in separate methodological traditions \citep{manski1999,rosenbaum1999a,rosenbaum1999b}, but police data present them together; causal inference about racial discrimination in use of force requires the two assumptions jointly.

Building on \citet{knoxetal2020}, we develop such a joint sensitivity framework. The framework proceeds in two steps. The first posits a level of racial discrimination in stops and uses it to handle the sample selection problem; the second assesses sensitivity to bias in encounters on the resulting data. These two steps cannot be separated: conducting the encounter-bias analysis first requires specifying \textit{some} level of racial discrimination in stops, and the sensitivity to encounter bias depends on \textit{how much} discrimination one specifies.

We apply the framework to the New York Police Department (NYPD) Stop, Question, and Frisk (SQF) data, examining how conclusions about racial discrimination in use of force depend on departures from No-Bias-in-Encounters across plausible levels of discrimination in stops. The conclusion that officers discriminate in force against minority civilians proves sensitive to small departures from No-Bias-in-Encounters --- a finding about what would follow from such hypothetical departures, not whether they occur.

To aid interpretation, we calibrate the sensitivity analysis against census data on the racial composition of NYPD patrol sub-areas. Within each patrol context, local demographics place a ceiling on how minority encounter probabilities can vary across officers: an officer can direct their patrol toward areas of different racial composition but cannot change the composition of any one area. Restricting attention to demographically feasible departures from No-Bias-in-Encounters strengthens the conclusion only modestly --- most patrol areas have enough demographic variation to permit even conclusion-altering departures.

\section{Data and Substantive Motivation} \label{sec: data}

The NYPD public administrative data on SQF stops begin in 2003 and continue through the present. Each record corresponds to an encounter that escalated to a stop. For each stop, the data include the officer's perception of the civilian's race using the NYPD's administrative categories. We restrict attention to stops with civilians categorized as white, Black, or Hispanic. The data also record whether the officer used force, the patrol context (precinct, and when available, sector and beat), temporal information, and other descriptors of the encounter and civilian.

Because these data include only encounters that escalate to a stop, they capture only a selected subset. Most police--civilian encounters would not result in a stop regardless of the civilian's race: passing on the sidewalk, a voluntary ``request for information,'' or a common-law inquiry short of reasonable suspicion. Nested within all encounters are those with \textit{potentially stoppable civilians}: encounters that could generate a stop, possibly depending on the civilian's race. If officers discriminate in whom they stop, only a subset of these potentially stoppable encounters appear in the SQF data; among those stops, only a further subset escalate to use of force. \Cref{fig: encounter hierarchy} illustrates this hierarchy.
\begin{figure}[H]
\centering
\includegraphics[width=0.58\linewidth]{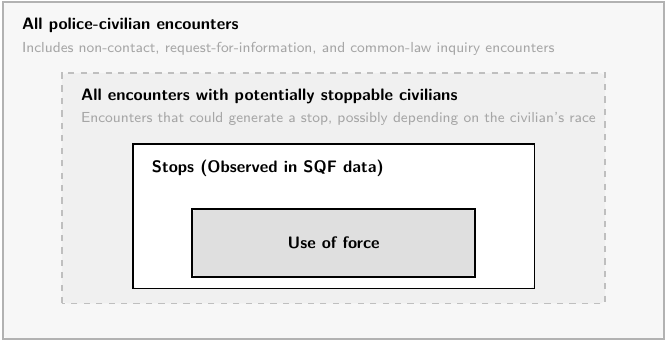}%
\caption{Nested populations of police--civilian interactions. The dashed region shows encounters with potentially stoppable civilians: encounters that could generate a stop, possibly depending on civilian race. Our analysis focuses on this region.}
\label{fig: encounter hierarchy}
\end{figure}

We ask how civilian race affects officer use of force among potentially stoppable encounters. These encounters --- the dashed region of \Cref{fig: encounter hierarchy} --- are where officers may intrude on civilian liberty, by initiating a stop or, more severely, by using force. Encounters outside this region cannot escalate to a stop, so including them would add structural zeros and mechanically attenuate the effect.

The observed data arise from a two-stage process. The first stage, the \textit{assignment mechanism}, governs how each officer-civilian encounter comes to be with either a minority civilian or a white civilian. The second stage, \textit{sample selection}, determines which of those encounters escalate to a stop and therefore appear in police administrative data.

\subsection{Encounter Assignment Process} \label{sec: encounter assignment}

We draw on opportunity models of crime from criminology and related fields \citep[for a review, see][]{wilcoxcullen2018}. On each day, an officer is assigned to a patrol context --- bureau (Patrol Services, Housing, or Transit), precinct, sector, beat, and tour --- and may have multiple encounters. We conceive those encounters as arising from a stochastic path-intersection process: officers and civilians move through the same environment over time, and an encounter occurs when their paths intersect. Each civilian enters an encounter already bearing a racial category, produced by the civilian's racialization within the social and historical context that gives rise to the NYPD's categories. Following the distinction between exposure and perception in \citet[][pp.~259--260]{hukohler-hausmann2025}, we take the race of the civilian in each encounter to be this prior racial status, not the officer's perception during the encounter \citep{greinerrubin2011}. In a ``veil of darkness'' setting \citep{groggerridgeway2006,piersonetal2020}, for example, we still classify an encounter as one with a minority civilian whenever the officer's path intersects that of a minority individual.

In the SQF administrative data, however, civilian race is the officer's post-stop classification on the UF-250 form (reproduced in Supplement Section~S.7), reflecting the officer's perception. We therefore treat this classification as a proxy for the civilian's racial status, not as the race itself. If the stopping decision may depend on race, race cannot be defined by a classification that exists only for civilians who are stopped. Racial discrimination in stops instead requires a notion of civilian race defined before, and independently of, the post-stop marking. The officer's classification and perception need not coincide with this prior categorization; we treat such discrepancies as limitations of the data rather than of the path-intersection framework, and set them aside here. Officer perception may still shape where officers patrol, condition responses to civilians of different races, and mediate the effect of civilian race on use of force, but those processes concern how policing decisions respond to race, not the definition of race in the encounter.

With civilian race so defined, we formalize \textit{encounter}. An \textit{encounter} is a position in the officer's patrol sequence --- the first, the second, and so on --- not the civilian who occupies it. Depending on the paths officers and civilians take, the same position could be occupied by a civilian of a different race. This position is the unit at which we can envision the ``counterfactual substitution of an individual with a different racial identity into the encounter, while holding the encounter's objective context ... fixed'' \citep[][p.~621]{knoxetal2020}.

The encounter's objective context includes the patrol context, fixed across realizations of the path-intersection process: the officer, the tour, the beat (within sector and precinct), and the number of encounters. The civilian profile occupying that encounter may still vary in features beyond race. Two Black male civilians, one at a public-housing entrance and one in a private-home driveway, share racial status but present distinct profiles because location is a nonracial attribute, analogous to an attribute in a conjoint profile \citep{hainmuelleretal2014}. Because the path-intersection process realizes race and nonracial attributes jointly, civilian race may be associated with location within the beat, dress, behavior, and other characteristics.

Holding the officer and patrol context fixed, each civilian's nonracial profile elicits one of four counterfactual stopping behaviors from the officer: stop regardless of race, stop only if minority, stop only if white, or do not stop. Following \citet{knoxetal2020}, we classify nonracial profiles into these four principal strata; the profiles that elicit each behavior may differ across officers and contexts. For example, an officer might stop an assault regardless of race (\textit{Always-Stop}), stop loitering near a housing entrance only if the civilian is a minority (\textit{Only-Minority-Stop}), stop another profile only if the civilian were white (\textit{Only-White-Stop}), or not stop at all (\textit{Never-Stop}). After this classification (formalized in \Cref{sec: strata}), civilians of different races who could occupy a given encounter may differ in nonracial profiles, but only when those differences do not alter the officer's counterfactual stopping behavior. The profile-level classification warrants treating principal strata as attributes of the encounter rather than of the civilian, as in \citet{knoxetal2020}: the civilians who could occupy a given encounter face the same officer in the same patrol context, and their profiles belong to the same principal stratum.

To define the causal effect of race, we ask what would happen in the same encounter with a civilian of a different race. The framework above characterizes this counterfactual: substituting a different-race civilian into the encounter --- with the patrol context and the principal stratum fixed --- defines the effect in the ``all else equal except for race'' convention \citep{heckman1998}. We acknowledge that this convention is contested on both conceptual and normative grounds \citep[see, e.g.,][]{hu2025}. Rather than defend a particular causal contrast, we argue in Supplement Section~S.1 that defining the contrast and inferring its effect are separable tasks: the confounding channels our sensitivity analysis addresses arise regardless of which contrast one adopts. We therefore emphasize the value of a framework that jointly addresses these channels while bracketing debates over which causal contrast to target.

The probability that an encounter is with a minority rather than a white civilian depends on both patrol context and officer behavior. Potentially stoppable civilian populations vary across contexts, and officer-specific characteristics --- patrol style, tolerance for risk, familiarity with an area, and implicit or explicit racial attitudes --- may further shape where within a context officers patrol. An officer who associates minority status with criminality may concentrate patrol at buildings or blocks with larger minority populations, raising the probability of a minority-civilian encounter relative to another officer in the same context.

\subsection{Stop Selection Process} \label{sec: stop selection}

If officers discriminate in whom they stop, some potentially stoppable encounters --- the dashed region of \Cref{fig: encounter hierarchy} --- go unrecorded, creating a sample selection problem. \citet{knoxetal2020} address this with two assumptions about the sample-selection process, formalized in \Cref{sec: setup}: No-Force-Without-Stop (if no stop occurs, force cannot occur) and No-Only-White-Stops (no encounter would result in a stop only if the civilian were white). Under these assumptions, the data miss only nonforce encounters with white civilians who would have been stopped had they been minority.

\section{Formal Setup}\label{sec: setup}

We target the average causal effect of civilian race on use of force among encounters that could result in a stop. Defining this estimand requires assumptions about how encounters relate to one another, how race affects the stopping decision, and how officers encounter civilians of different races. We state these assumptions in turn, connecting each to the substantive setting in \Cref{sec: data}.

\subsection{Encounters and Potential Outcomes}\label{sec: setup:po}

Condition on the realized total of $N$ encounters and index them $i \in \{1, \ldots, N\}$. The civilian profile in encounter $i$ is $\bm{t}_i \coloneqq (z_i, \bm{v}_i) \in \mathcal{T}$, where $\mathcal{T} \coloneqq \{0, 1\} \times \mathcal{V}$, $z_i \in \{0, 1\}$ is the civilian's racial status ($z_i = 1$ for Black or Hispanic, $z_i = 0$ for white), and $\bm{v}_i \in \mathcal{V}$ collects nonracial attributes. Following \Cref{sec: encounter assignment}, $z_i$ records the civilian's racial status in the exposure sense \citep{hukohler-hausmann2025} --- the race of the civilian whose path intersects with the officer's --- not the officer's perception.

To state our no-interference assumption, we initially allow each encounter's outcomes to depend on the full profile vector $\bm{t} \coloneqq (\bm{t}_1, \ldots, \bm{t}_N) \in \mathcal{T}^N$ via potential stopping and use-of-force outcome functions $s_i, y_i: \mathcal{T}^N \to \{0, 1\}$.

\begin{assm}[No interference]\label{assm: no-interference}
For all $i \in \{1, \ldots, N\}$ and all $\bm{t}, \bm{t}^{\prime} \in \mathcal{T}^{N}$ satisfying $\bm{t}_{i} = \bm{t}^{\prime}_{i}$, $s_{i}(\bm{t}) = s_{i}(\bm{t}^{\prime}) \quad \text{and} \quad y_{i}(\bm{t}) = y_{i}(\bm{t}^{\prime}).$
\end{assm}
\noindent Under \Cref{assm: no-interference}, the potential outcomes for encounter $i$ depend on only the profile $\bm{t}_{i}$ occupying that encounter. We therefore write $s_{i}(z, \bm{v})$ and $y_{i}(z, \bm{v})$ hereafter.

We also assume that the potential outcomes for use of force and stops follow the nested hierarchy in \Cref{fig: encounter hierarchy} --- analogous to \citet{knoxetal2020}'s ``Mandatory Reporting'' assumption.
\begin{assm}[No-Force-Without-Stop] \label{assm: no-force-without-stop}
For all encounters $i \in \{1, \ldots, N\}$ and $z \in \{0, 1\}$, $y_{i}(z, \bm{v}) \leq s_{i}(z, \bm{v})$.
\end{assm} \vspace{-.5em}
\noindent This is a \textit{structural zero} assumption \citep[in the sense of][]{zhangrubin2003} in which encounters not resulting in a stop are constrained to have no use of force.

\subsection{Principal Strata and the Potentially Stoppable Population} \label{sec: strata}

Under \Cref{assm: no-interference}, the stopping potential outcome $s_i(z, \bm{v})$ is a well-defined function of the profile of the civilian occupying encounter $i$. Because both $z$ and $s_i$ are binary, every nonracial profile $\bm{v} \in \mathcal{V}$ falls into exactly one of four categories defined by the pair $(s_i(1, \bm{v}), s_i(0, \bm{v}))$. Following \citet{knoxetal2020}, we label the four resulting principal strata \textit{Always-Stop}, \textit{Only-Minority-Stop}, \textit{Only-White-Stop}, and \textit{Never-Stop}, with $\mathcal{V}^{\mathrm{AS}}_i, \mathcal{V}^{\mathrm{OMS}}_i, \mathcal{V}^{\mathrm{OWS}}_i, \mathcal{V}^{\mathrm{NS}}_i$ defined by $(s_i(1, \bm{v}), s_i(0, \bm{v}))$ equal to $(1,1), (1,0), (0,1), (0,0)$ respectively. Because the stopping potential outcome reflects the judgment of a particular officer in a particular patrol context, the partition is encounter-specific, and the same nonracial profile may belong to different principal strata in different encounters.

Because $s_i(z, \bm{v})$ is constant on $\mathcal{V}_i^r$, we introduce the principal stratum label function $r_i: \mathcal{V} \to \{\mathrm{AS}, \mathrm{OMS}, \mathrm{OWS}, \mathrm{NS}\}$, where $r_i(\bm{v})$ is the unique label $r$ such that $\bm{v} \in \mathcal{V}_i^r$, and write $s_i(z, r)$ for the common value of $s_i(z, \bm{v})$ on $\mathcal{V}_i^r$.

We now impose a parallel restriction on the use of force potential outcomes.
\begin{assm}[Use-of-force depends only on race within principal strata] \label{assm: stratum-sufficiency}
For every $i \in \{1, \ldots, N\}$ and any $z \in \{0, 1\}$, if $r_i(\bm{v}) = r_i(\bm{v}^{\prime})$, then $y_i(z, \bm{v}) = y_i(z, \bm{v}^{\prime})$.
\end{assm}
\noindent Civilians of the same race who could occupy a given encounter may differ in their nonracial profiles, but \Cref{assm: stratum-sufficiency} implies that such differences are irrelevant for the use of force potential outcomes when those civilians belong to the same principal stratum. The use of force potential outcome therefore depends on $\bm{v}$ only through the stratum label $r_i(\bm{v})$, and we write $y_i(z, r)$ for its common value across all profiles with $r_i(\bm{v}) = r$.

We condition on the realized principal stratum label $r_i(\bm{v}_i)$, which takes one of the values in $\{\mathrm{AS}, \mathrm{OMS}, \mathrm{OWS}, \mathrm{NS}\}$. Because $s_i(z, r)$ and $y_i(z, r)$ depend on $\bm{v}$ only through $r$, conditioning on $r_i(\bm{v}_i)$ eliminates the remaining non-race channel, so civilians who could fill the encounter differ only in their race. To lighten the notation in what follows, we write $r_i$ for $r_i(\bm{v}_i)$, and correspondingly $s_i(z)$ for $s_i(z, r_i)$ and $y_i(z)$ for $y_i(z, r_i)$, leaving the conditioning implicit. From here on, both the potential outcomes and the assignment mechanism are specified in terms of civilian race $z$ alone.

Our substantive interest in encounters that could result in a stop or use of force (\Cref{sec: data}), together with the structural constraint imposed by \Cref{assm: no-force-without-stop}, motivates restricting attention to \textit{potentially stoppable} encounters --- those whose principal stratum is AS, OMS, or OWS. Because we condition on $r_i$, this label is held fixed across possible assignments rather than being a property of the civilian who occupies the encounter, and the set of potentially stoppable encounters is therefore well-defined independently of the realized assignment. We therefore let $n \leq N$ denote the number of potentially stoppable encounters and re-index them as $i \in \{1, \ldots, n\}$.

Among the potentially stoppable encounters, we further restrict which principal strata are present, imposing \citep[following][]{knoxetal2020} that no encounter would result in a stop only if the civilian were white.
\begin{assm}[No-Only-White-Stops] \label{assm: no-OWS}
For all encounters $i \in \{1, \ldots, n\}$, $s_{i}(1) \geq s_{i}(0)$.
\end{assm}
\noindent \Cref{assm: no-OWS} rules out the Only-White-Stop principal stratum, implying that $r_i \in \{\mathrm{AS},\mathrm{OMS}\}$ for all $i$. Like \citet{knoxetal2020}, we regard the assumption of few Only-White-Stops as plausible, although exactly \textit{no} Only-White-Stops is perhaps less so.

\subsection{Stratification and the Assignment Model} \label{sec: setup:assignment}

Each encounter $i \in \{1, \ldots, n\}$ inherits a baseline covariate vector $\bm{x}_i$ summarizing the patrol context (precinct, beat, time of day, Impact Zone status, and additional features detailed in \Cref{sec: application}) prior to the civilian occupying that encounter. We group encounters into strata based on $\bm{x}_i$, with $\mathcal{G}$ denoting the resulting strata, $g_i \in \mathcal{G}$ the stratum of encounter $i$, and $n_g$ the number of encounters in stratum $g$ (indexed $i = 1, \ldots, n_g$).

Let $\bm{z}_{g} \coloneqq (z_{g,1}, \ldots, z_{g,n_{g}}) \in \{0, 1\}^{n_{g}}$ denote the vector of civilian-race indicators in stratum $g$. Define $n_{g,1} \coloneqq \sum_{i=1}^{n_g} z_{g,i}$ and $n_{g,0} \coloneqq n_g - n_{g,1}$ as the numbers of minority-civilian and white-civilian encounters, respectively. Our inferential framework requires at least one minority-civilian and one white-civilian encounter within each stratum, and we therefore restrict attention to $\mathcal{G}^{\ast} \coloneqq \{g \in \mathcal{G} : n_{g,1} \geq 1 \text{ and } n_{g,0} \geq 1\}$. Encounters in strata outside $\mathcal{G}^{\ast}$ are excluded from the analysis, and we let $n^{\ast} \coloneqq \sum_{g \in \mathcal{G}^{\ast}} n_{g} \leq n$.

Conditioning on the realized count $n_{g,1}$ within each $g \in \mathcal{G}^*$, the assignment space is $\Omega_g \coloneqq \{\bm{z}_g \in \{0,1\}^{n_g}: \sum_{i=1}^{n_g} z_{g,i} = n_{g,1}\}$, and the full assignment space is $\Omega \coloneqq \prod_{g \in \mathcal{G}^*} \Omega_g$. Across $\bm{z}_g \in \Omega_g$, only the assignment of civilian races varies; the officer, patrol context, and stopping principal stratum are held fixed.

We adapt a standard restriction on the assignment mechanism \citep[][Chapter~4]{rosenbaum2002a}. The civilian-race indicators $\{Z_{g,i} : i = 1, \ldots, n_g,\; g \in \mathcal{G}\}$ are mutually independent Bernoulli random variables with probabilities $\pi_{g,i}$ and, for any two encounters in the same stratum, the odds that the encounter is with a minority rather than a white civilian may differ by at most a factor $\Gamma \geq 1$:
\begin{align}\label{eq: assign-model}
\dfrac{1}{\Gamma} \leq \dfrac{\pi_{g,i} / (1 - \pi_{g,i})}{\pi_{g,j} / (1 - \pi_{g,j})} \leq \Gamma \quad \text{for all } i, j \in \{1, \ldots, n_{g}\} \text{ and all } g \in \mathcal{G},
\end{align}
where $\pi_{g,i} \coloneqq \Pr(Z_{g,i} = 1)$. When $\Gamma = 1$, all encounters within a stratum have the same probability of being with a minority civilian. Larger values of $\Gamma$ permit progressively greater heterogeneity in these probabilities.

We conduct all inference conditional on the event $\bm{Z}_g \in \Omega_g$, which fixes the number of minority-civilian encounters within each stratum. We therefore write $\varphi_{g,i} \coloneqq \Pr(Z_{g,i} = 1 \given \bm{Z}_g \in \Omega_g)$ for the conditional probability that encounter $i$ in stratum $g$ is with a minority civilian. For a full assignment $\bm{z}_g \in \Omega_g$, we write $p(\bm{z}_g) \coloneqq \Pr(\bm{Z}_g = \bm{z}_g \given \bm{Z}_g \in \Omega_g)$.

We say that \textit{No-Bias-in-Encounters} holds if, within every stratum, all encounters have the same conditional probability of being with a minority civilian. Formally,
\begin{equation}\label{eq: no-bias}
\textbf{No-Bias-in-Encounters:} \quad \varphi_{g,i} = \varphi_{g,j} \quad \text{for all } i, j \in \{1, \ldots, n_{g}\} \text{ and all } g \in \mathcal{G}.
\end{equation}
Under \eqref{eq: assign-model}, $\Gamma = 1$ yields No-Bias-in-Encounters, while $\Gamma > 1$ permits departures from this condition. Referring back to \Cref{sec: data}, $\Gamma = 1$ --- and hence No-Bias-in-Encounters --- could hold if the racial composition of the potentially stoppable civilian population were similar across encounters within a stratum and if officers did not differ in baseline characteristics that systematically shaped their paths through space and time.

\subsection{Causal Parameters}\label{sec: causal parameters}

We define two causal parameters for the subpopulation of potentially stoppable encounters in informative strata $\mathcal{G}^{\ast}$. The first is the average causal effect of civilian race on the stop decision. Under \Cref{assm: no-interference}, this effect within stratum $g \in \mathcal{G}^{\ast}$ is
\begin{align}\label{eq: ate-stop g}
\rho_g \coloneqq \dfrac{1}{n_g}\sum_{i = 1}^{n_g}\left[s_{g,i}(1) - s_{g,i}(0)\right].
\end{align}
The overall average causal effect across informative strata is the weighted average
\begin{align}\label{eq: ate-stop}
\rho \coloneqq \sum_{g \in \mathcal{G}^{\ast}} \left(n_g / n^{\ast}\right) \rho_g.
\end{align}
Under the additional assumption of No-Only-White-Stops (\Cref{assm: no-OWS}), this parameter equals the proportion of potentially stoppable encounters that belong to the Only-Minority-Stop principal stratum. Encounters in this principal stratum are precisely those for which sample selection arises, since a white civilian occupying such an encounter would not trigger a stop and therefore would not appear in police administrative data.

The second causal parameter is the primary quantity of interest: the average causal effect of civilian race on use of force. Under Assumptions \ref{assm: no-interference} and \ref{assm: stratum-sufficiency}, and conditional on principal stratum membership $\{r_{g,i}\}$ established in \Cref{sec: strata}, this effect within stratum $g \in \mathcal{G}^{\ast}$ is
\begin{align}\label{eq: ate-force g}
\tau_g \coloneqq \dfrac{1}{n_g}\sum_{i = 1}^{n_g}\left[y_{g,i}(1) - y_{g,i}(0)\right].
\end{align}
The overall average causal effect across informative strata is the weighted average
\begin{align}\label{eq: ate-force}
\tau \coloneqq \sum_{g \in \mathcal{G}^{\ast}} \left(n_g / n^{\ast}\right) \tau_g.
\end{align}
The parameter $\tau$ is the target of inference for our joint sensitivity analysis.

\section{Sensitivity Analysis for Sample Selection Alone} \label{sec: existing sens analysis}

The average causal effect $\tau$ is not identified from police data because encounters that were not stopped do not appear in the data. Under Assumptions~\ref{assm: no-interference}--\ref{assm: no-OWS}, the stratum-specific effect $\tau_g$ admits a tractable decomposition (proved in Supplement Section~S.3.2):
\begin{align} \label{eq: bias-in-force decomp}
\tau_g & = \bar{y}_g(1) - (1 - \rho_g)\,\bar{y}_g^{\mathrm{AS}}(0),
\end{align}
where $\bar{y}_g(z) \coloneqq n_g^{-1}\sum_{i=1}^{n_g} y_{g,i}(z)$ is the average potential outcome under civilian race $z$ in stratum $g$, $\bar{y}_g^{\mathrm{AS}}(z)$ restricts the average to Always-Stop encounters, and $\rho_g = n_{g,\mathrm{OMS}}/n_g$ is the proportion of Only-Minority-Stop encounters.

The $(1 - \rho_g)$ term accounts for Only-Minority-Stop encounters, which contribute zero to the white-civilian average because a white civilian in such an encounter would not be stopped and hence (under Assumption~2) could not experience force. When $\rho_g = 0$ every encounter is Always-Stop and $\tau_g$ reduces to $\bar{y}_g(1) - \bar{y}_g^{\mathrm{AS}}(0)$; as $\rho_g \to 1$ the subtracted term vanishes and $\tau_g$ rises toward $\bar{y}_g(1)$. Both $\bar{y}_g(1)$ and $\bar{y}_g^{\mathrm{AS}}(0)$ are estimable from stopped encounters, so identification reduces to a sensitivity analysis over $\rho_g$.

Researchers can restrict $\bm{\rho} \coloneqq (\rho_g)_{g \in \mathcal{G}^{\ast}}$ using domain knowledge. \citet[][p.~631]{knoxetal2020} restrict $\rho_g = \rho \in [0.32, 0.34]$ for all $g$ using
excess minority stop rates from \citet{gelmanetal2007} and hit-rate differentials from \citet{goeletal2016}.

A plug-in estimator of \eqref{eq: bias-in-force decomp} adjusts the Difference-in-Means among stopped encounters by $\rho_g$:
\begin{align} \label{eq: est tau_g}
\hat{\tau}_g \coloneqq \hat{\bar{y}}_g(1) - (1-\rho_g)\hat{\bar{y}}_g(0),
\end{align}
where $\hat{\bar{y}}_g(z)$ is the mean use of force among stopped encounters of civilian race $z$ in stratum $g$ (Supplement Section~S.4 gives the explicit ratios). Under Assumptions~\ref{assm: no-interference}--\ref{assm: no-OWS} and No-Bias-in-Encounters \eqref{eq: no-bias}, $\hat{\tau}_g$ is consistent for $\tau_g$ as stratum sizes grow with the number of strata held fixed (Supplement Section~S.4). This consistency result justifies estimating the aggregate $\tau$ by weighting each $\hat{\tau}_g$ by its stratum share. Standard asymptotic theory then permits inference for $\tau$ at any fixed $\bm{\rho}$ \citep[][Theorem2.1]{bickelvanzwet1978}.

Departures from No-Bias-in-Encounters expose two issues with this $\rho_g$-only sensitivity analysis, the first for estimation and the second for inference. First, $\rho_g$ is a proportion, not a count, so $\rho_g$ does not determine $n_g$ or the assignment space $\Omega_g$. Under uniform assignment, this is harmless --- $\hat{\tau}_g$ is an IPW estimator \citep{horvitzthompson1952} whose $\Omega_g$-dependence cancels --- but under nonuniform assignment the cancellation fails and the estimator depends on an unidentified $n_g$ through $\abs{\Omega_g}$. Second, inference requires $\Omega_g$ regardless of the assignment distribution, since $p$-values and confidence sets compare the observed test statistic to its distribution across $\Omega_g$. The remedy for both is a sensitivity parameter that fixes the \textit{number} of missing encounters, and therefore $n_g$ and $\Omega_g$.

\section{Sequential Inference and Sensitivity Analysis} \label{sec: seq sens}

We address the two confounding channels --- racial discrimination in stops and racial bias in encounters --- sequentially. The stop-selection channel must come first: until we specify the number of missing white-civilian encounters and thereby determine $n_g$ and $\Omega_g$, the encounter-assignment sensitivity analysis has no assignment space on which to operate. We therefore introduce a sensitivity parameter for discrimination in stops that fixes this count, then conduct inference on the augmented data under the restriction in \eqref{eq: assign-model}, evaluating sensitivity to violations of No-Bias-in-Encounters.

\subsection{\texorpdfstring{Step 1: Augment the Data for Posited $\bm{\underline{\rho}}$}{Step 1: Augment the Data for Posited rho}}

In standard sample selection problems, the number of units is known but some outcomes are missing. Our setting inverts this structure. The use of force outcomes of the unobserved encounters are in fact known --- every missing encounter was not stopped and therefore, by \Cref{assm: no-force-without-stop} (No-Force-Without-Stop), has a use of force outcome of $0$. What is unknown is how many such encounters exist. Under \Cref{assm: no-OWS} (No-Only-White-Stops), the only missing encounters are white-civilian encounters that were not stopped, so each stratum's minority-civilian count $n_{g,1}$ is observed while the white-civilian count $n_{g,0}$, and hence the total $n_g$, is not. Specifying the number of missing white-civilian encounters that were not stopped therefore fixes the total number of encounters and suffices to reconstruct the dataset that would have been observed absent sample selection \citep[see][p.~630]{knoxetal2020}.

Under \Cref{assm: no-OWS} (No-Only-White-Stops), every white-civilian encounter in stratum $g$ is either Always-Stop and observed or Only-Minority-Stop and missing. The conditioning in \Cref{sec: setup:assignment} fixes the race counts $n_{g,1}$ and $n_{g,0}$, but not their Always-Stop and Only-Minority-Stop compositions. The number of missing white-civilian encounters, $n_{g,0,\mathrm{OMS}}(\bm{Z}_g) \coloneqq \sum_{i: r_{g,i} = \mathrm{OMS}} (1 - Z_{g,i})$, is therefore a random variable under $\bm{Z}_g \in \Omega_g$, as is the analogous minority-civilian count $n_{g,1,\mathrm{OMS}}(\bm{Z}_g) \coloneqq \sum_{i: r_{g,i} = \mathrm{OMS}} Z_{g,i}$.

Recall that the discrimination in stops parameter $\rho_g$, defined in \Cref{sec: causal parameters}, is the proportion of encounters in stratum $g$ that are Only-Minority-Stop, given by $\rho_g = n_{g,\mathrm{OMS}} / n_g$. The numerator decomposes as $n_{g,\mathrm{OMS}} = n_{g,0,\mathrm{OMS}}(\bm{Z}_g) + n_{g,1,\mathrm{OMS}}(\bm{Z}_g)$, the sum of the two random-variable counts introduced above. This total is fixed across assignments even though its two summands vary with $\bm{Z}_g$.

Positing a realized value for $n_{g,0,\mathrm{OMS}}(\bm{Z}_g)$ under the observed data specifies the realized value of the first summand but leaves the realized value of $n_{g,1,\mathrm{OMS}}(\bm{Z}_g)$ unspecified. Because the unspecified summand is nonnegative, the posited value places a lower bound $\underline{\rho}_g$ on $\rho_g$. This bound is tight when $n_{g,1,\mathrm{OMS}}(\bm{Z}_g)$ is zero and loosens as that value grows.

The lower bound $\underline{\rho}_g$ can be expressed in terms of a number $w$ of missing Only-Minority-Stop white-civilian encounters, interpreted as a realized value for $n_{g,0,\mathrm{OMS}}(\bm{Z}_g)$ under the observed data. Each value of $w$ determines a total encounter count $n_{g,1} + n_{g,0,\mathrm{AS}} + w$ and a corresponding lower bound $\underline{\rho}_g = w / (n_{g,1} + n_{g,0,\mathrm{AS}} + w)$ on the proportion of Only-Minority-Stop encounters in the stratum. Conversely, every feasible value of $\underline{\rho}_g$ maps to a unique $w$. \Cref{prop: bias-in-stops missing OMS control} formalizes this equivalence.
\begin{prop}[Equivalence of Discrimination-in-Stops bound and missing control encounters]
\label{prop: bias-in-stops missing OMS control}
Under Assumptions \ref{assm: no-interference} -- \ref{assm: no-OWS}, fix a stratum $g \in \mathcal{G}^{\ast}$ with observed minority-civilian count $n_{g,1}$ and observed Always-Stop white-civilian count $n_{g,0,\mathrm{AS}}$. For each posited count $w \in \mathbb{Z}_{\geq 0}$ of missing Only-Minority-Stop control encounters, the map $w \mapsto w/(n_{g,1} + n_{g,0,\mathrm{AS}} + w)$ is a bijection from $\mathbb{Z}_{\geq 0}$ to the feasible domain $\mathcal{F}_{\underline{\rho}_g} \coloneqq \{w/(n_{g,1} + n_{g,0,\mathrm{AS}} + w) : w \in \mathbb{Z}_{\geq 0}\} \subset [0,1)$, with inverse $w = \underline{\rho}_g(n_{g,1} + n_{g,0,\mathrm{AS}})/(1 - \underline{\rho}_g)$.
\end{prop}
\noindent Specifying $\underline{\rho}_g \in \mathcal{F}_{\underline{\rho}_g}$ therefore determines a unique value of $w$, allowing the analyst to augment the observed data with the corresponding appended zero outcomes.

Because $\mathcal{F}_{\underline{\rho}_g}$ is a discrete set (one value for each nonnegative integer $w$), a researcher's chosen $\underline{\rho}_g \in [0, 1)$ may fall between feasible values. We define the corresponding posited count as the nonnegative integer whose implied lower bound is closest to the specified value:
\begin{align} \label{eq: tilde n g 0 OMS}
\tilde{n}_{g,0,\mathrm{OMS}}^{\underline{\rho}_g} & \coloneqq \argmin \limits_{w \in \mathbb{Z}_{\geq 0}} \left\lvert \underline{\rho}_g - \dfrac{w}{n_{g,1} + n_{g,0,\mathrm{AS}} + w} \right\rvert.
\end{align}
The choice of $\underline{\rho}_g$ thus specifies a realized augmented stratum size $\tilde{n}_g^{\underline{\rho}_g} \coloneqq n_{g,1} + n_{g,0,\mathrm{AS}} + \tilde{n}_{g,0,\mathrm{OMS}}^{\underline{\rho}_g}$, and the subsequent analysis varies $\bm{Z}_g$ over assignments in which the minority-civilian count equals the realized $n_{g,1}$.

With $\tilde{n}_{g,0,\mathrm{OMS}}^{\underline{\rho}_g}$ determined, we \textit{augment} stratum $g$'s observed data by appending $\tilde{n}_{g,0,\mathrm{OMS}}^{\underline{\rho}_g}$ zeros to the white-civilian use of force outcomes. The augmented Difference-in-Means is $\hat{\tau}_g^{\underline{\rho}_g} \coloneqq \hat{\bar{y}}_g(1) - \hat{\bar{y}}_g^{\underline{\rho}_g}(0)$, where $\hat{\bar{y}}_g(1)$ is the observed minority-civilian mean and $\hat{\bar{y}}_g^{\underline{\rho}_g}(0)$ is the augmented white-civilian mean: the observed sum of white-civilian force outcomes divided by the sum of the observed white-civilian count and $\tilde{n}_{g,0,\mathrm{OMS}}^{\underline{\rho}_g}$. Under Assumptions \ref{assm: no-interference} -- \ref{assm: no-OWS}, $\hat{\tau}_g^{\underline{\rho}_g}$ coincides with the Difference-in-Means computed from the full set of potentially stoppable encounters in stratum $g$ when the stipulated count $\tilde{n}_{g,0,\mathrm{OMS}}^{\underline{\rho}_g}$ equals the realized value of $n_{g,0,\mathrm{OMS}}(\bm{Z}_g)$ under the observed data (proof in Supplement Section~S.4.1).

Under \Cref{assm: no-force-without-stop}, every appended white-civilian encounter has a use of force outcome of $0$. These appended zeros dilute the white-civilian mean toward zero without changing the minority-civilian mean. The augmented Difference-in-Means $\hat{\tau}_g^{\underline{\rho}_g}$ is therefore monotonically increasing in $\underline{\rho}_g$ and converges to the minority-civilian mean as $\underline{\rho}_g \to 1$.

For example, consider a stratum with one stopped minority-civilian encounter and one stopped white-civilian encounter, both with force; the observed Difference-in-Means is $0$. Stipulating $\underline{\rho}_g = 1/2$ appends two nonforce white-civilian encounters that were not stopped. The appended zeros reduce the white-civilian mean from $1$ to $1/3$ while the minority-civilian mean remains $1$, so the augmented Difference-in-Means rises from $0$ to $2/3$. \Cref{fig: augmentation diagram} illustrates this augmentation procedure.
\begin{figure}[!h]
\centering
\includegraphics[width=0.85\linewidth, trim={0 300 0 50}, clip]{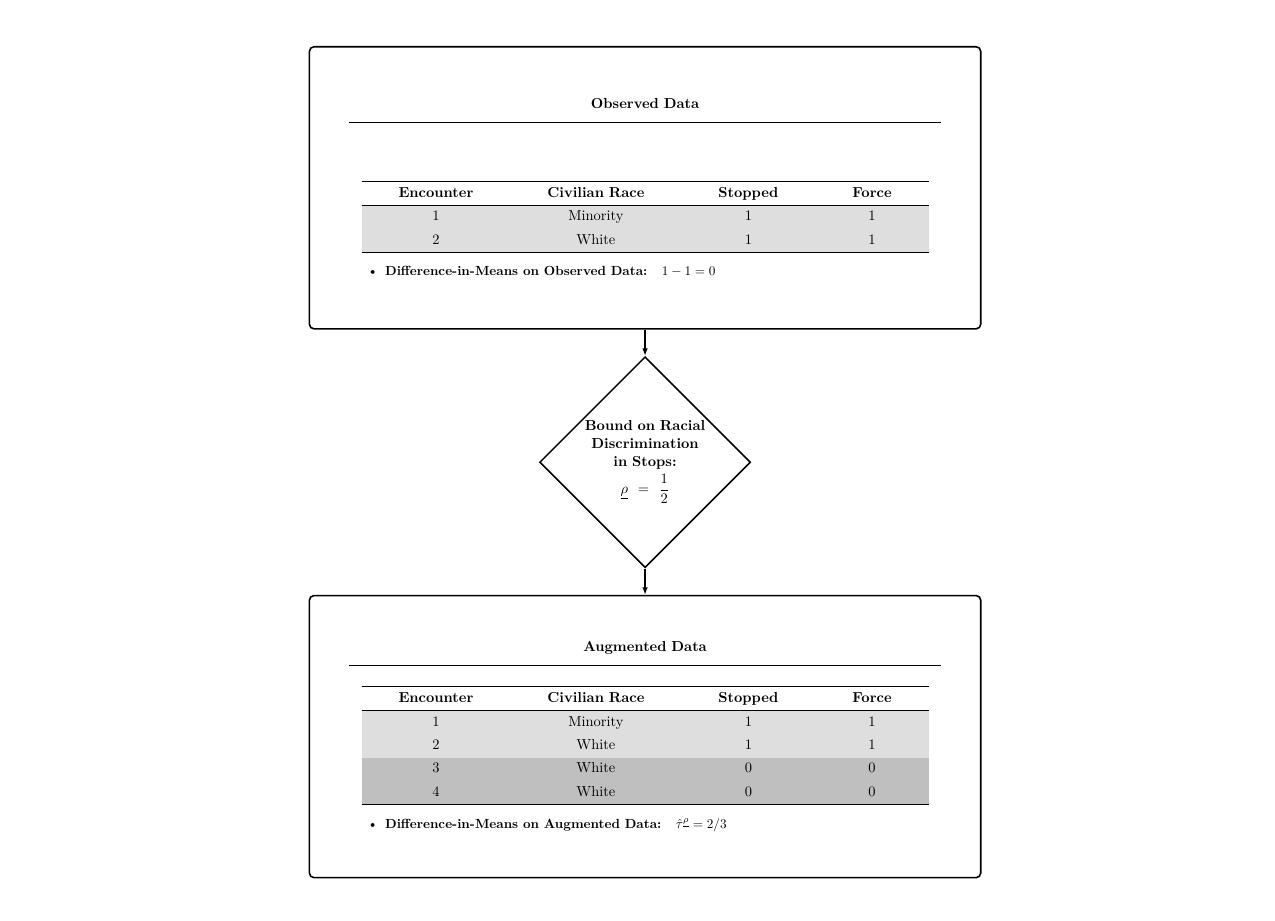}
\caption{The top panel shows the observed data, which omit all nonforce white-civilian encounters that were not stopped, and the corresponding Difference-in-Means computed on the observed data. The middle diamond specifies a lower bound on racial discrimination in stops ($\underline{\rho} = 1/2$), which determines how many nonforce white-civilian encounters that were not stopped to append to the observed data. The bottom panel shows the augmented observed data, with darker shading indicating the appended encounters. The Difference-in-Means $\hat{\tau}^{\underline{\rho}}$ is then computed on this augmented observed data.}
\label{fig: augmentation diagram}
\end{figure}

\subsection{\texorpdfstring{Step 2: Test for Discrimination in Force under Posited $\Gamma$}{Step 2: Test for Force Bias under Posited Gamma}}

Step 1 fixed the augmented stratum size $\tilde{n}_g^{\underline{\rho}_g} \coloneqq n_{g,1} + n_{g,0,\mathrm{AS}} + \tilde{n}_{g,0,\mathrm{OMS}}^{\underline{\rho}_g}$ and with it the augmented assignment space $\Omega_g^{\underline{\rho}_g}$, which contains $\binom{\tilde{n}_g^{\underline{\rho}_g}}{n_{g,1}}$ possible assignments. On this augmented data, we can ask whether the observed racial disparity in force is unusually extreme under the null hypothesis $\tau = \tau_0$, given a specified level of discrimination in stops $\bm{\underline{\rho}} \coloneqq (\underline{\rho}_g)_{g \in \mathcal{G}^{\ast}}$ --- with $\Omega^{\bm{\underline{\rho}}} \coloneqq \prod_{g \in \mathcal{G}^{\ast}} \Omega_g^{\underline{\rho}_g}$ the corresponding across-strata assignment space --- and a specified level of bias in encounters $\Gamma$. Under No-Bias-in-Encounters, assignment is uniform over $\Omega_g^{\underline{\rho}_g}$ and standard inference procedures yield valid tests of this null. Under departures from No-Bias-in-Encounters --- with assignment probabilities bounded by $\Gamma$ as in \eqref{eq: assign-model} --- a valid hypothesis test must instead control Type~I error for every assignment mechanism consistent with the $\Gamma$-bound.

Conducting such tests requires the conditional assignment probabilities of the true mechanism, which are unknown. Sensitivity analyses therefore evaluate inference under the worst-case assignment mechanism within the class defined by \eqref{eq: assign-model} --- the configuration that yields the largest $p$-value. A test that rejects under this worst-case configuration must also reject under the true mechanism, ensuring that the Type~I error rate is bounded above by the nominal level.

To compute worst-case $p$-values, we use the parametric submodel from \citet{rosenbaum1987a}. Within each stratum $g$, the model posits that $\log\{\pi_{g,i}/(1 - \pi_{g,i})\} = \kappa_g + \log(\Gamma)\, u_{g,i}$ for an unobserved covariate $u_{g,i} \in [0,1]$, sensitivity parameter $\Gamma \geq 1$, and stratum-specific intercept $\kappa_g$ that reflects the probability an encounter is with a minority civilian on the basis of observed covariates --- a shared parameter across all encounters in stratum $g$ under our exact post-stratification.

Conditioning on the realized count $n_{g,1}$ removes the nuisance intercept $\kappa_g$, and varying $\bm{u}_g$ over $[0,1]^{\tilde{n}_g^{\underline{\rho}_g}}$ reproduces exactly the class of conditional assignment distributions on $\Omega_g^{\underline{\rho}_g}$ induced by marginal Bernoulli probabilities satisfying the bound in \eqref{eq: assign-model} \citep{rosenbaum1995a}. The submodel therefore serves as a parametrization of this class --- not as a substantive description of how officers encounter civilians --- and the worst-case configuration of $\bm{u}_g$ in the submodel coincides with the worst-case assignment mechanism in the class, namely one that attains the $\Gamma$-bound sharply. Mechanisms satisfying the same $\Gamma$-bound but not attaining this sharp extreme generally yield smaller $p$-values \citep{hengsmall2021}. This worst-case construction has a minimax decision-theoretic interpretation \citep{cohenetal2020}: inference is valid under any mechanism in the class and conservative when the true mechanism does not sit at the sharp extreme.

Identifying the worst-case configuration is well developed for sharp null hypotheses --- those specifying an individual effect for every encounter \citep{rosenbaumkrieger1990,gastwirthetal2000,rosenbaum2018}. Our null hypothesis is composite (it concerns $\tau$, an average), so the worst case must be found over both assignment mechanisms satisfying \eqref{eq: assign-model} and all potential outcome configurations consistent with the null. \citet{fogarty2023} constructs a tilted test statistic from the centered Difference-in-Means $\hat{\tau}_i - \tau_0$ --- so named because the transformation tilts the centered quantity toward zero --- and develops it for upper-tailed tests in settings where each stratum contains exactly one treated unit and one or more controls. Combined with a suitable CLT and a consistently conservative variance estimator, the tilted statistic yields hypothesis tests that are asymptotically valid for all assignment mechanisms satisfying \eqref{eq: assign-model} and for all potential outcome configurations consistent with the composite null.

We generalize the tilting approach of \citet{fogarty2023} in two directions, to post-stratified designs with arbitrary numbers of treated and control units, and to tests against alternatives of both larger and smaller ATEs. The first step is to bound the probability of any assignment vector $\bm{z}_g^{\underline{\rho}_g} \in \Omega_g^{\underline{\rho}_g}$ under the restriction on the assignment model in \eqref{eq: assign-model}.
\begin{lem} \label{lem: prob bounds}
Under the restriction on the assignment model in \eqref{eq: assign-model}, the lower ($\underline{p}$) and upper ($\overline{p}$) bounds on the conditional probability of any $\bm{z}_g^{\underline{\rho}_g} \in \Omega_g^{\underline{\rho}_g}$ for $\Gamma \geq 1$ are
\begin{align} \label{eq: lb cond z prob I}
\underline{p}\left(\bm{z}_g^{\underline{\rho}_g}; \Gamma\right) & =
\dfrac{1}{\sum_{\bm{a}_g \in \Omega_g^{\underline{\rho}_g}}
\Gamma^{\bm{a}_g^{\top}\left(\bm{1}-\bm{z}_g^{\underline{\rho}_g}\right)}},
\\[0.75em]
\overline{p}\left(\bm{z}_g^{\underline{\rho}_g}; \Gamma\right) & =
\dfrac{\Gamma^{n_{g,1}}}{\sum_{\bm{a}_g \in \Omega_g^{\underline{\rho}_g}}
\Gamma^{\bm{a}_g^{\top}\bm{z}_g^{\underline{\rho}_g}}}. \label{eq: ub cond
z prob I}
\end{align}
\end{lem}

\Cref{lem: prob bounds} generalizes the probability bounds from the case of one minority-civilian encounter per stratum considered by \citet{fogarty2023} to arbitrary post-stratified designs; the Supplementary Materials verify that \Cref{lem: prob bounds} reduces to \citet[][Equation~5, p.~2201]{fogarty2023} in that special case.

The numerators in \eqref{eq: lb cond z prob I} and \eqref{eq: ub cond z prob I} do not depend on $\underline{\rho}_g$, but the denominators do because the set of feasible assignment vectors is determined by the postulated number of Only-Minority-Stop encounters with white civilians. Computing these denominators does not require enumerating all elements of $\Omega_g^{\underline{\rho}_g}$. The Supplementary Materials show how both can be computed in closed form for any values of $\underline{\rho}_g$ and $\Gamma$.

The tilted statistic substitutes, for each stratum, whichever probability bound shifts the centered Difference-in-Means $\hat{\tau}_g^{\underline{\rho}_g} - \tau_0$ toward zero, intentionally working against the direction of the alternative. Let $d \in \{+1, -1\}$ denote the direction of the alternative hypothesis, with $d = +1$ for an upper-tailed test (alternative: $\tau > \tau_0$) and $d = -1$ for a lower-tailed test (alternative: $\tau < \tau_0$). The stratum-level tilted statistic is
\begin{align} \label{eq: tilted IPW diff-in-means g}
\hat{\tau}_g^{\mathrm{tilt}}\left(\underline{\rho}_g; \Gamma, \tau_0, d\right) & =
\dfrac{1}{\abs{\Omega_g^{\underline{\rho}_g}}}
\left(\hat{\tau}_g^{\underline{\rho}_g} - \tau_0\right)
\begin{cases}
\overline{p}\left(\bm{z}_g^{\underline{\rho}_g}; \Gamma\right)^{-1}, &
\text{if } d \left(\hat{\tau}_g^{\underline{\rho}_g}-\tau_0\right) \geq 0,
\\[0.5em]
\underline{p}\left(\bm{z}_g^{\underline{\rho}_g}; \Gamma\right)^{-1}, &
\text{if } d\left(\hat{\tau}_g^{\underline{\rho}_g}-\tau_0\right) < 0,
\end{cases}
\end{align}
which applies the probability bound that moves the centered Difference-in-Means in the direction least favorable to the alternative. The tilted statistic for the study population averages these stratum-level contributions using weights $\tilde{n}_g^{\underline{\rho}_g}/\tilde{n}^{\ast}$, where $\tilde{n}^{\ast} \coloneqq \sum_{g \in \mathcal{G}^{\ast}} \tilde{n}_g^{\underline{\rho}_g}$, which depends on $\bm{\underline{\rho}}$ through its summands, though we leave this dependence implicit in the notation. The resulting statistic is
\begin{align} \label{eq: tilted IPW diff-in-means}
\hat{\tau}^{\mathrm{tilt}}\left(\bm{\underline{\rho}}; \Gamma, \tau_0, d\right) &
\coloneqq \sum_{g \in \mathcal{G}^{\ast}}
\left(\tilde{n}_g^{\underline{\rho}_g} / \tilde{n}^{\ast}\right) \,
\hat{\tau}_g^{\mathrm{tilt}}\left(\underline{\rho}_g; \Gamma, \tau_0, d\right).
\end{align}

All expectations below are taken with potential outcomes held fixed, over a distribution of $\bm{Z}_g$ on $\Omega_g^{\underline{\rho}_g}$ consistent with the restriction in \eqref{eq: assign-model}. \Cref{prop: EV tilted IPW diff-in-means bound 0} ensures that, under the null, the tilted statistic's expectation is shifted away from the alternative for all potential outcomes consistent with the null and all assignment mechanisms consistent with $\Gamma$.
\begin{prop} \label{prop: EV tilted IPW diff-in-means bound 0}
Under Assumptions \ref{assm: no-interference} -- \ref{assm: no-OWS}, the tilted statistic in \eqref{eq: tilted IPW diff-in-means} has nonpositive expectation under the null when the alternative is upper-tailed and nonnegative expectation under the null when the alternative is lower-tailed. Specifically, for any $\bm{\underline{\rho}} \in [0,1)^{\abs{\mathcal{G}^{\ast}}}$, with each stratum augmented by $\tilde{n}_{g,0,\mathrm{OMS}}^{\underline{\rho}_g}$ missing encounters as in \eqref{eq: tilde n g 0 OMS}, and any $\Gamma \geq 1$,
\begin{align*}
\E\left[\hat{\tau}^{\mathrm{tilt}}\left(\bm{\underline{\rho}}; \Gamma,
\tau_0, d = +1\right)\right] & \leq 0 \quad
\text{ (upper-tailed alternative)}, \\
\E\left[\hat{\tau}^{\mathrm{tilt}}\left(\bm{\underline{\rho}}; \Gamma,
\tau_0, d = -1\right)\right] & \geq 0 \quad
\text{ (lower-tailed alternative)}.
\end{align*}
\end{prop}
\noindent To convert \Cref{prop: EV tilted IPW diff-in-means bound 0} into a valid hypothesis test, we need a standard error that consistently upper bounds the true standard deviation of the tilted statistic under the null. In the Supplementary Materials, following \citet{fogarty2018a,fogarty2023}, we construct such a standard error, $\widehat{\Var}[\hat{\tau}^{\mathrm{tilt}}(\bm{\underline{\rho}}; \Gamma, \tau_0, d)]^{1/2}$, and show that it consistently upper bounds the true standard error for any values of $\bm{\underline{\rho}}$ and $\Gamma$. Dividing the tilted statistic by this conservative standard error yields a test statistic whose null distribution is stochastically dominated by the standard normal (under regularity conditions ensuring a CLT). The resulting $p$-value, computed from the standard normal reference distribution, is conservative.

\subsection{The Joint Sensitivity Framework}

Our approach introduces two sensitivity parameters that operate on distinct components of the data-generating process. The parameter $\bm{\underline{\rho}}$ governs the potential outcome structure, determining how many nonforce encounters with white civilians are missing from the observed data; $\Gamma$ governs racial bias in encounters, bounding how much officers within the same stratum may differ in their probabilities of encountering a minority civilian. The framework proceeds in two steps: for each postulated $\bm{\underline{\rho}}$, augment the data and construct $\Omega^{\bm{\underline{\rho}}}$; on that domain, compute worst-case probability bounds consistent with $\Gamma$, yielding $\hat{\tau}^{\mathrm{tilt}}\left(\bm{\underline{\rho}};\Gamma,\tau_0,d\right)$.

\Cref{fig: sens diagram} picks up where \Cref{fig: augmentation diagram} leaves off. Once $\bm{\underline{\rho}}$ fixes the augmented encounters and observed assignment, $\Gamma$ implies upper and lower probability bounds and hence the overall tilted statistic.

\begin{figure}[!h]
\centering
\includegraphics[width=0.85\linewidth, trim={0 300 0 50}, clip]{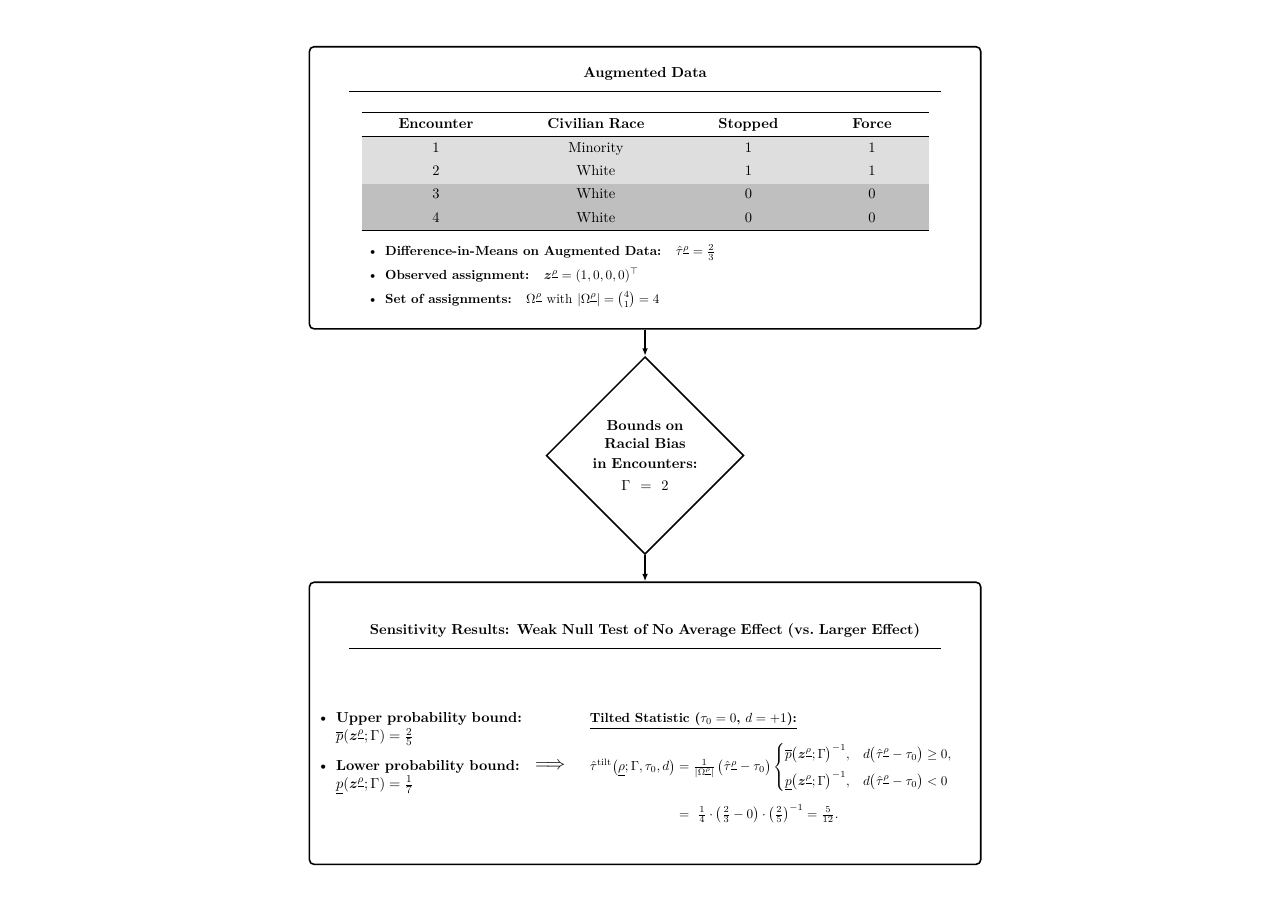}
\caption{The top panel displays the augmented data, observed assignment, and assignment set $\Omega^{\underline{\rho}}$. The middle diamond specifies a bound on racial bias in encounters ($\Gamma = 2$), which determines upper and lower probability bounds on the observed assignment. The bottom panel shows the tilted statistic for the weak null of no average effect against the alternative of a larger effect.}
\label{fig: sens diagram}
\end{figure}
\noindent The tilted statistic adjusts the centered Difference-in-Means $\hat{\tau}^{\underline{\rho}} - \tau_0$ toward zero by applying the probability bound from \Cref{lem: prob bounds}, ensuring valid inference under any encounter-process mechanism consistent with $\Gamma$.

In summary, $\bm{\underline{\rho}}$ fixes the augmented data and assignment space $\Omega^{\bm{\underline{\rho}}}$; $\Gamma$ quantifies how far assignment probabilities on that space may depart from uniformity. \Cref{prop: EV tilted IPW diff-in-means bound 0} establishes the worst-case expectation bound for the tilted statistic at each $(\bm{\underline{\rho}}, \Gamma)$, generalizing \citet{fogarty2023} to post-stratified designs with arbitrary numbers of minority-civilian and white-civilian encounters per stratum and to one-sided tests in either direction around $\tau_0$. Because the augmented stratum size is fixed across $\Omega^{\bm{\underline{\rho}}}$, introducing $\bm{\underline{\rho}}$ adds no inferential complication beyond the fixed-design guarantees available for $\Gamma$ alone.

Combined with a CLT and a consistently conservative variance estimator, this bound produces a test that controls Type~I error asymptotically at each $(\bm{\underline{\rho}}, \Gamma)$. Inverting the test yields a pointwise $(1 - \alpha)$ confidence set for $\tau$ at each postulated $(\bm{\underline{\rho}}, \Gamma)$, with coverage conditional on those postulated sensitivity parameters being correct, not simultaneous coverage across the grid.

\subsection{Sequential vs.\ Separate Sensitivity Analyses} \label{sec: sequential-vs-separate-sensitivity}

Because $\bm{\underline{\rho}}$ and $\Gamma$ govern conceptually distinct confounding channels --- racial discrimination in stops and racial bias in encounters --- one might expect that conducting each sensitivity analysis separately and then combining the results would be equivalent to our joint framework. It is not, for two related reasons: $\Gamma$ is structurally dependent on $\bm{\underline{\rho}}$, and the two parameters interact multiplicatively within the tilted statistic (Supplement Section~S.5.2 illustrates this with a worked example: a single stratum with one minority-civilian encounter, the closed-form expression, and how the tilted statistic responds to changes in either parameter).

Specifying $\bm{\underline{\rho}}$ determines the augmented assignment space $\Omega^{\bm{\underline{\rho}}}$, and $\Gamma$ bounds the odds ratios over assignments in that space --- so $\Gamma$ has no meaning until $\bm{\underline{\rho}}$ fixes the space. A researcher who conducts a $\Gamma$-sensitivity analysis on the observed data alone is therefore implicitly assuming $\bm{\underline{\rho}} = \bm{0}$, while a researcher who conducts a $\bm{\underline{\rho}}$-sensitivity analysis alone is implicitly assuming No-Bias-in-Encounters ($\Gamma = 1$). Neither analysis, taken on its own, captures what happens when both confounding channels are present.

\section{Application to NYPD's SQF Data} \label{sec: application}

We restrict attention to encounters with males aged 17--26, excluding radio runs in which a dispatcher directs the officer to a specific suspect whose reported race may predetermine the race of the civilian encountered. We analyze 2003--2013, when the NYPD operated SQF as a centralized, incentive-driven policing regime under Mayor Michael Bloomberg and Police Commissioner Raymond Kelly. This regime effectively ended after the August 2013 decision in \textit{Floyd v. City of New York} and the January 2014 inauguration of Mayor Bill de Blasio and Police Commissioner William Bratton; stop counts fell from more than 500{,}000 in 2012 to roughly 45{,}000 in 2014. Because SQF data are available only from 2003 onward, the 2003--2013 period remains the central empirical reference point in legal, political, and scholarly debates about SQF policy, though our results are not a definitive assessment of NYPD use of force overall.

Following the encounter-assignment process of \Cref{sec: encounter assignment}, we exactly stratify the encounters by combinations of observed covariates, so that every encounter in a stratum shares the same covariate values. Within each stratum, encounters then have the same probability of being with a minority civilian on the basis of observed covariates, and any remaining differences in minority encounter odds must be driven by unobserved covariates. The sensitivity parameter $\Gamma$ therefore captures only these unobserved differences, so smaller values of $\Gamma$ become more plausible after exact stratification than they would be without it.

We define strata using spatial covariates (precinct, sector, and beat), departing from existing research that typically conditions only on precinct \citep[e.g.,][]{fryer2019}. Since approximately \PctSectorMissing\% of sector values and \PctBeatMissing\% of beat values are missing, our poststratification includes missingness indicators, comparing encounters on the same beat within the same sector and precinct when both are observed, within the same sector and precinct when beat is missing, or within the same precinct alone otherwise. The strata also incorporate temporal covariates (year, season, tour, daytime), contextual features (indoors, transit, public housing, high-crime area, high-crime time, and NYPD Impact Zone, with zones from \citealp[][see Supplement]{macdonaldetal2016}), and officer proxies (uniformed officer, transit, and housing indicators). Officer rank, recorded only from 2017 onward, is unobserved here.

We then restrict attention to strata containing at least one minority and at least one white encounter, as described in \Cref{sec: setup}. Our causal target is $\tau$ in \eqref{eq: ate-force}, defined over these strata. Below, we infer this target over different combinations of $\bm{\rho}$ and $\Gamma$, and also over stratum-specific values of $\Gamma$ calibrated to the demographics of each stratum.

\subsection{Results} \label{sec: results}

The baseline difference in average use of force between minority and white encounters averaged over informative strata is substantively small --- about \BaselineTauPP\ percentage points (estimated S.E.\ $\approx \BaselineSE$). Under the baseline assumptions of no Bias-in-Stops ($\underline{\rho}_g = 0$ for every $g$, so $\underline{\rho} = 0$) and no Bias-in-Encounters ($\Gamma = 1$), we reject the null of no average causal effect in favor of discrimination in force against minority civilians. \Cref{fig: robustness frontier heatmap} summarizes how this conclusion changes as we relax each assumption.

\begin{figure}[!h]
\centering
\includegraphics[width=.84\linewidth]{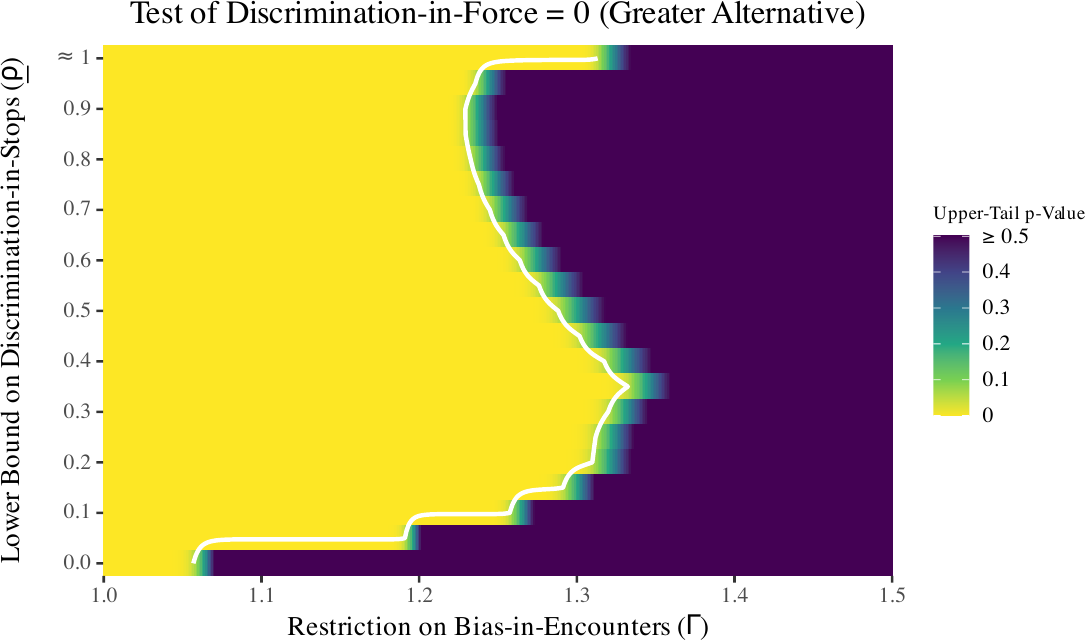}%
\caption{Upper-tail $p$-values for a one-sided test of no racial discrimination in force. The horizontal axis shows bias in encounters, governed by $\Gamma$; the vertical axis shows the common lower bound $\underline{\rho}$ on discrimination in stops ($\underline{\rho}_g = \underline{\rho}$ for every $g \in \mathcal{G}^{\ast}$). The white contour marks the 5\% critical boundary.}
\label{fig: robustness frontier heatmap}
\end{figure}

The interaction structure described in \Cref{sec: sequential-vs-separate-sensitivity} is borne out in the SQF application. When $\underline{\rho} = 0$, the test transitions from significant to insignificant at \(\Gamma \approx \GammaRhoZeroInsig\), so even a small degree of bias in encounters, with no discrimination in stops, would be enough to explain the observed disparity. A small increase in $\underline{\rho}$, however --- from 0 to \(\RhoMinOffsetAtGammaStar\) --- keeps the $p$-value below $\alpha = \AlphaLevel$, because the augmented Difference-in-Means initially grows faster than the tilting factor shrinks. For $\Gamma$ between $\PlateauGammaMin$ and $\PlateauGammaMax$, the contour boundary lies near $\underline{\rho} \approx \PlateauRho$, and the tilted estimator is nearly flat in $\Gamma$. As $\Gamma$ increases further, the rejection region expands upward along the $\underline{\rho}$ axis, reaching the largest $\underline{\rho}$ for which the test still rejects at $\Gamma \approx \GammaHighestRhoSigMax$, where the test remains significant up to $\underline{\rho} \approx \RhoSigMaxHighest$. Beyond that point, the test moves back into the non-rejection region, and once $\Gamma \approx \GammaAllInsigMin$, the test rejects for no value of $\underline{\rho}$.

The heatmap displays the joint sensitivity of the test across all $(\underline{\rho}, \Gamma)$ combinations. To connect this surface to empirically plausible discrimination in stops, we use the range $0.32$ to $0.34$ that \citet{knoxetal2020} derive from \citet{goeletal2016} and \citet{gelmanetal2007}, as noted in \Cref{sec: existing sens analysis}. Because $\underline{\rho}$ is a lower bound, we focus on $\underline{\rho} = 0.34$, the value in this range hardest for $\Gamma$ to overturn: larger $\underline{\rho}$ appends more zeros to white-civilian outcomes, increasing the augmented Difference-in-Means and resisting the shrinkage induced by $\Gamma$. Results at $\underline{\rho} = 0.32$ are nearly identical. \Cref{fig: conf sets plausible rho lbs} presents 95\% confidence sets across values of $\Gamma$ at $\underline{\rho} = 0.34$.
\begin{figure}[!h]
\centering
\includegraphics[width=.8\linewidth]{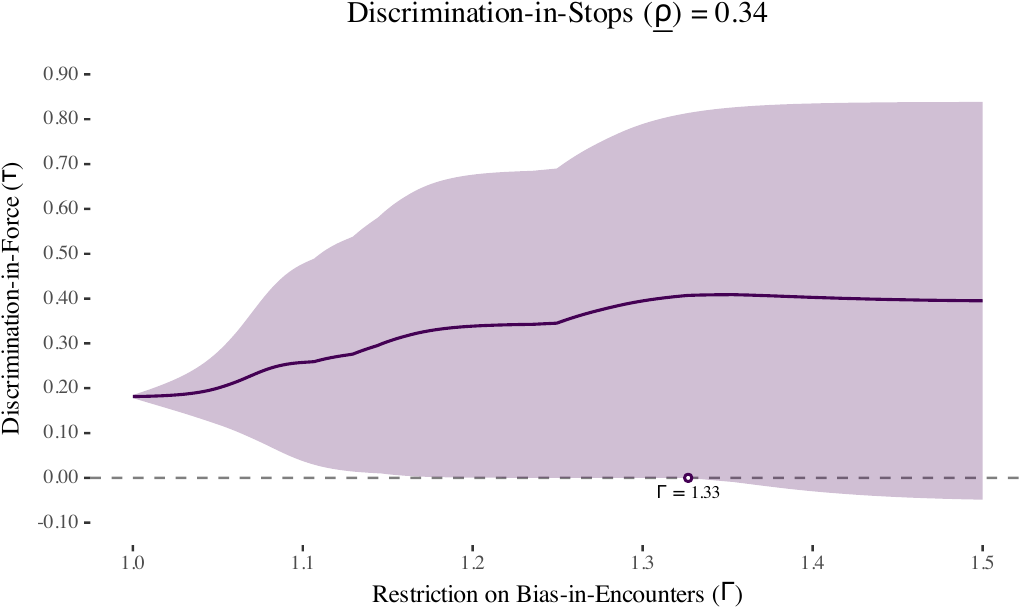}%
\caption{95\% confidence sets (shaded ribbon) and median nonrejected null hypotheses (solid curve) for Discrimination-in-Force across $\Gamma$, at $\underline{\rho} = 0.34$. The labeled point marks the smallest $\Gamma$ at which the confidence set first includes $0$.}
\label{fig: conf sets plausible rho lbs}
\end{figure}
The confidence set first includes $0$ at $\Gamma = \ChangePointUpper$. Past this point, the confidence set contains both positive and negative values of $\tau$. The entry of zero reflects growing uncertainty about the sign and magnitude of $\tau$ as hidden bias in encounters increases. This uncertainty does not imply that discrimination in force against minority civilians is absent; absence of evidence should not be interpreted as evidence of absence.

After the confidence set includes zero, it remains asymmetric, extending well above zero but only slightly below. This asymmetry is structural. Under \Cref{assm: no-OWS}, Only-Minority-Stop encounters cannot result in force against a white civilian --- who would not have been stopped --- so these encounters contribute nonnegative individual causal effects. As a result, the average effect is bounded away from large negative values regardless of $\Gamma$.

\subsection{Geographic Calibration of Sensitivity Bounds} \label{sec: interpretation}

The preceding sensitivity analysis imposes a uniform $\Gamma$ across all strata, conducting worst-case inference as though every stratum's hidden bias in encounters could be as large as $\Gamma$. This approach is conservative when bias varies across strata: the uniform bound penalizes every stratum as though its bias were as large as the worst stratum's. \citet{hengsmall2021} address a related problem --- interactions between observed and unobserved covariates --- and show that stratum-specific bounds can reduce this conservatism. We use the same principle but determine the bounds differently: rather than modeling such interactions, we exploit the fact that each stratum's encounter-odds ratio is physically constrained by the racial composition within its geographic footprint. Under the path-intersection model of \Cref{sec: encounter assignment}, officers can direct patrols toward locations with different racial compositions, but cannot change the composition at any given location. Census data therefore provide a demographic ceiling on how much minority encounter probabilities can vary within a stratum.

We implement this idea by constructing a demographic ceiling $\Gamma_g^{\mathrm{geo}}$ for each stratum $g$. In the sensitivity analysis, the operative bound for stratum $g$ is $\min(\Gamma,\Gamma_g^{\mathrm{geo}})$; therefore, as the global parameter $\Gamma$ increases, the operative sensitivity parameter for stratum $g$ is capped at its demographic ceiling $\Gamma_g^{\mathrm{geo}}$. As shown in a corollary to \Cref{prop: EV tilted IPW diff-in-means bound 0} in the Supplementary Material, the worst-case bound extends directly to stratum-specific sensitivity parameters $\Gamma_g$, with \Cref{prop: EV tilted IPW diff-in-means bound 0} corresponding to the special case $\Gamma_g = \Gamma$ for all $g$.

Following \citet{zhaoetal2022}, we construct $\Gamma_g^{\mathrm{geo}}(\xi)$ from 2010 Census block-group demographics within each stratum's geographic footprint. For each block group $b$, let $\eta_b$ denote the minority-to-white population odds. The ceiling is the ratio of the $(1-\xi)$-th to the $\xi$-th population-weighted quantile of $\eta_b$ across the stratum's block groups, where $\xi \in [0, 0.5)$ controls tail trimming. With $\xi = 0$ the ratio uses the most demographically extreme block groups; larger $\xi$ compares more typical locations. Supplement Section~S.9 gives the formal definition and the procedures for single-block-group strata and years (2003--2005) without encounter coordinates.

\begin{figure}[!h]
\centering
\includegraphics[width=.92\linewidth]{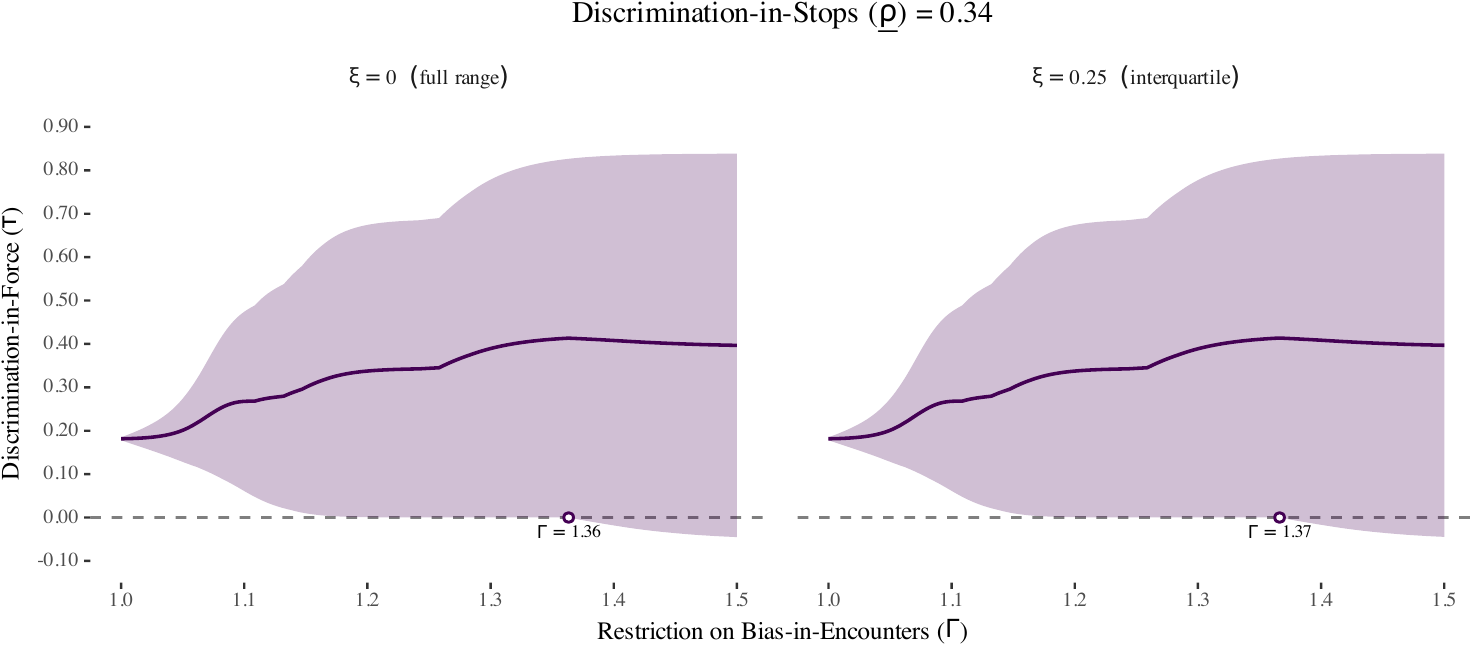}%
\caption{95\% confidence sets (shaded ribbon) and median nonrejected null hypotheses (solid curve) for Discrimination-in-Force across $\Gamma$ under geographic calibration at $\underline{\rho} = 0.34$. Panels show $\xi = 0$ (left) and $\xi = 0.25$ (right); labels mark the smallest $\Gamma$ values at which the confidence sets first include $0$.}
\label{fig: geog const conf sets plausible rho lbs}
\end{figure}

\Cref{fig: geog const conf sets plausible rho lbs} shows that restricting the sensitivity analysis to demographically feasible departures from No-Bias-in-Encounters shifts the changepoint modestly upward --- making the conclusion slightly more robust --- from $\Gamma = \ChangePointUpper$ to $\Gamma = \GeoChangePointXiZeroRhoUpper$ when $\xi = 0$.
Using $\xi = 0.25$ moves the changepoint slightly further to $\Gamma = \GeoChangePointXiTwentyFiveRhoUpper$. Values of $\xi$ beyond $0.25$ are difficult to motivate substantively: The median stratum contains encounters in only two block groups, so trimming more than $25\%$ from each tail of the within-stratum distribution of $\eta_b$ discards most of the demographic variation and forces $\Gamma_g^{\mathrm{geo}}$ toward $1$, which would mean that officers patrol exclusively in the demographically typical block group within their stratum.

To put these changepoints in substantive terms, a $\Gamma$ of roughly $1.37$ means that if one encounter has a $0.50$ probability that the civilian is a minority, another encounter in the same stratum could have a probability of about $0.58$. This modest difference could arise from where within a beat officers patrol or from which civilians are outdoors at a given time.

The changepoint values do not require implausibly large demographic contrasts. At $\xi = 0$, $\GeoCdfComplementXiZeroRhoUpper\%$ of encounters at $\underline{\rho} = 0.34$ occur in strata whose block-group-level racial composition could generate encounter-odds ratios at least as large as the respective changepoints. The geographic calibration therefore indicates that departures from No-Bias-in-Encounters of this magnitude are demographically feasible: In the majority of strata, the racial composition of the areas where encounters occur is heterogeneous enough to support differences of this size.

This interpretation has at least three caveats. First, census data measure residential population rather than the population present at a given time, so the bounds are less reliable for commercial districts or transit settings. Second, $\Gamma$ governs encounter probabilities among potentially stoppable civilians, whose racial composition may differ from the residential population. Third, the ceilings assigned to strata in 2003--2005 rely on geographic information from later years; if these compositions change slowly, this procedure preserves the census-implied geographic constraints rather than leaving those strata unconstrained.

\section{Conclusion}

Inferences about racial discrimination in police use of force depend on assumptions about two aspects of police behavior: which civilians officers would stop and the process by which officers encounter civilians in the first place. Our joint sensitivity analysis varies both sequentially. Applied to NYPD SQF data, the analysis shows that the two channels interact, so analyzing each separately would miss their combined implications.

The framework extends existing methods in three ways. First, we generalize the sensitivity analysis of \citet{fogarty2023} to post-stratified designs with arbitrary numbers of minority- and white-civilian encounters per stratum and to both greater-than and less-than alternatives. Second, we introduce a lower-bound parameterization of discrimination in stops together with a data-augmentation procedure that accounts for missing encounters, enabling encounter-bias sensitivity analysis that existing parameterizations do not allow. Finally, we calibrate stratum-specific sensitivity bounds using census demographics.

Applying this framework to NYPD SQF data yields a substantively important but fragile finding. Under plausible levels of discrimination in stops, the observed racial disparity in force supports rejecting the null in favor of discrimination in force against minority civilians. A modest degree of bias in the encounter process --- officers in the same patrol context differing by a factor of $\Gamma = \ChangePointLower$ to $\ChangePointUpper$ in their odds of encountering a minority civilian --- suffices to render the test statistically insignificant, and our census calibration confirms that \GeoCdfComplementXiZeroRhoLower\% of encounters are in strata where the demographic variation across block groups is large enough to produce odds ratios of this magnitude.

Whether officers' patrol behavior actually produces encounter-odds differences of this size --- rather than merely patrolling in areas where such differences are demographically possible --- is a question our framework cannot answer on its own. Two strategies could help. First, careful qualitative work --- ethnographic research on patrol patterns and officer decision-making --- could establish whether within-stratum variation in encounter probabilities is substantively meaningful. Second, improvements in design sensitivity --- for instance, through refined matching or stratification strategies --- would lower the stakes of this question by pushing the changepoints to larger values of $\Gamma$, ensuring that conclusions hold across a wider range of departures and making disagreements about plausibility less consequential.

The fragility of the finding should not be mistaken for evidence that racial discrimination in force is absent. Across the confidence sets, those beyond the changepoint include zero, but zero lies near the lower boundary while most values in the sets correspond to higher force rates if encounters were with minority rather than white civilians. The inclusion of zero reflects increasing uncertainty as the posited level of encounter bias grows, not a shift in the weight of evidence toward no discrimination, and treating the inclusion of zero as evidence against discrimination would amount to accepting a null that the analysis lacks the power to reject.

Beyond statistical structure, the framework has implications for policy responses to the observed disparity, which is itself a matter of normative concern, independent of what our framework establishes about causal sources. The joint sensitivity analysis nonetheless clarifies which responses fit which causal stories. If officers in comparable patrol contexts are equally likely to encounter minority and white civilians, the disparity reflects discrimination in officer responses --- calling for de-escalation training, use-of-force guidelines, and individual accountability. If instead certain officers are more likely to encounter minority civilians and are also more prone to use force, the response depends on whether deployment patterns or officers' own patrol choices drive that variation: the former calls for organizational reform, the latter for interventions on individual conduct. By varying both confounding channels sequentially, the joint analysis makes these distinctions visible in a way that separate analyses cannot.

\section*{Disclosure statement}\label{disclosure-statement}

The authors have no conflicts of interest to declare.

\section*{Data Availability Statement}\label{data-availability-statement}

The analysis uses public data: NYPD SQF data 2003--2013, 2010 Census block-group polygons and demographics for New York City, NYC police precinct shapefiles, and NYPD Operation Impact zone polygons georeferenced from Figure~1 of \citet{macdonaldetal2016}. The replication archive contains all code and intermediate datasets.

\begin{center}

{\large\bf SUPPLEMENTARY MATERIAL}

\end{center}

\begin{description}
\item[Supplement (PDF):] Proofs, derivations, and impact-zone polygons creation details.
\item[Replication code:] Scripts under \texttt{Code/} and \texttt{Impact\_Zones/}, organized by a \texttt{Makefile}, including data download, cleaning, post-stratification and sensitivity analysis.
\end{description}

\bibliography{bibliography.bib}

@book{taylor2004,
	address = {Cambridge, UK},
	author = {Taylor, Paul C},
	publisher = {Polity Press},
	title = {Race: A Philosophical Introduction},
	year = {2004}}

@article{macdonaldetal2016,
	author = {Mac{D}onald, John and Fagan, Jeffrey and Geller, Amanda},
	journal = {{PL}oS {ONE}},
	number = {6},
	pages = {e0157223},
	title = {The Effects of Local Police Surges on Crime and Arrests in {N}ew {Y}ork {C}ity},
	volume = {11},
	year = {2016}}

@article{piersonetal2020,
	author = {Pierson, Emma and Simoiu, Camelia and Overgoor, Jan and Corbett-Davies, Sam and Jenson, Daniel and Shoemaker, Amy and Ramachandran, Vignesh and Barghouty, Phoebe and Phillips, Cheryl and Shroff, Ravi and Goel, Sharad},
	journal = {Nature Human Behaviour},
	pages = {736-745},
	title = {A Large-Scale Analysis of Racial Disparities in Police Stops across the {U}nited {S}tates},
	volume = {4},
	year = {2020}}

@article{groggerridgeway2006,
	author = {Grogger, Jeffrey and Ridgeway, Greg},
	journal = {Journal of the American Statistical Association},
	number = {475},
	pages = {878-887},
	title = {Testing for Racial Profiling in Traffic Stops From Behind a Veil of Darkness},
	volume = {101},
	year = {2006}}

@article{cohenetal2020,
	author = {Cohen, Peter L and Olson, Matt A and Fogarty, Colin B},
	journal = {Biometrika},
	number = {4},
	pages = {809-825},
	title = {Multivariate One-sided Testing in Matched Observational Studies as an Adversarial Game},
	volume = {107},
	year = {2020}}

@article{hengsmall2021,
	author = {Heng, Siyu and Small, Dylan S},
	journal = {Statistica Sinica},
	pages = {2331-2353},
	title = {Sharpening the {R}osenbaum Sensitivity Bounds to Address Concerns about Interactions between Observed and Unobserved Covariates},
	volume = {31},
	year = {2021}}

@article{heckman1998,
	author = {Heckman, James J},
	journal = {Journal of Economic Perspectives},
	number = {2},
	pages = {101-116},
	title = {Detecting Discrimination},
	volume = {12},
	year = {1998}}

@article{hukohler-hausmann2025,
	author = {Hu, Lily and Kohler-Hausmann, Issa},
	journal = {Law \& Society Review},
	number = {2},
	pages = {239-264},
	title = {What is Perceived When Race is Perceived and Why It Matters for Causal Inference and Discrimination Studies},
	volume = {59},
	year = {2025}}

@article{bickelvanzwet1978,
	author = {Bickel, Peter J and {van Zwet}, Willem Rutger},
	journal = {Annals of Statistics},
	number = {5},
	pages = {937-1004},
	title = {Asymptotic Expansions for the Power of Distributionfree Tests in the Two-Sample Problem},
	volume = {6},
	year = {1978}}

@article{wilcoxcullen2018,
	author = {Wilcox, Pamela and Cullen, Francis T},
	journal = {Annual Review of Criminology},
	pages = {123-148},
	title = {Situational Opportunity Theories of Crime},
	volume = {1},
	year = {2018}}

@article{gaebleretal2022,
	author = {Gaebler, Johann and Cai, William and Basse, Guillaume and Shroff, Ravi and Goel, Sharad and Hill, Jennifer},
	journal = {Statistics and Public Policy},
	number = {1},
	pages = {26-48},
	title = {A Causal Framework for Observational Studies of Discrimination},
	volume = {9},
	year = {2022}}

@article{hu2025,
	author = {Hu, Lily},
	journal = {The Journal of Philosophy},
	title = {Normative Facts and Causal Structure},
	year = {2025}}

@article{zhaoetal2022,
	author = {Zhao, Qingyuan and Keele, Luke J and Small, Dylan S and Joffe, Marshall M},
	journal = {The American Political Science Review},
	number = {1},
	pages = {337-350},
	title = {A Note on Posttreatment Selection in Studying Racial Discrimination in Policing},
	volume = {116},
	year = {2022}}

@article{hu2023,
	author = {Hu, Lily},
	journal = {Journal of Moral Philosophy},
	title = {What is `Race' in Algorithmic Discrimination on the Basis of Race?},
	year = {2023}}

@article{fryer2018,
	author = {Fryer, Jr., Roland G},
	journal = {AEA Papers and Proceedings},
	pages = {228--233},
	title = {Reconciling Results on Racial Differences in Police Shootings},
	volume = {108},
	year = {2018}}

@article{fryer2019,
	author = {Fryer, Jr., Roland G},
	journal = {The Journal of Political Economy},
	number = {3},
	pages = {1210--1261},
	title = {An Empirical Analysis of Racial Differences in Police Use of Force},
	volume = {127},
	year = {2019}}

@article{goeletal2016,
	author = {Goel, Sharad and Rao, Justin M and Shroff, Ravi},
	journal = {Annals of Applied Statistics},
	number = {1},
	pages = {365--394},
	title = {Precinct or prejudice? {U}nderstanding racial disparities in {N}ew {Y}ork {C}ity's stop-and-frisk policy},
	volume = {10},
	year = {2016}}

@article{gelmanetal2007,
	author = {Gelman, Andrew and Fagan, Jeffrey and Kiss, Alex},
	journal = {Journal of the American Statistical Association},
	number = {479},
	pages = {813--823},
	title = {An Analysis of the {N}ew {Y}ork {C}ity {P}olice {D}epartment's ``{S}top-and-{F}risk'' Policy in the Context of Claims of Racial Bias},
	volume = {102},
	year = {2007}}

@article{knoxetal2020,
	author = {Knox, Dean and Lowe, Will and Mummolo, Jonathan},
	journal = {The American Political Science Review},
	number = {3},
	pages = {619--637},
	title = {Administrative Records Mask Racially Biased Policing},
	volume = {114},
	year = {2020}}

@article{greinerrubin2011,
	author = {Greiner, D James and Rubin, Donald B},
	journal = {The Review of Economics and Statistics},
	number = {3},
	pages = {775--785},
	title = {Causal Effects of Perceived Immutable Characteristics},
	volume = {93},
	year = {2011}}

@article{rosenbaum1995a,
	author = {Rosenbaum, Paul R},
	journal = {Journal of the American Statistical Association},
	number = {432},
	pages = {1424--1431},
	title = {Quantiles in Nonrandom Samples and Observational Studies},
	volume = {90},
	year = {1995}}

@article{fogarty2023,
	author = {Fogarty, Colin B},
	journal = {Biometrics},
	number = {3},
	pages = {2196-2207},
	title = {Testing Weak Nulls in Matched Observational Studies},
	volume = {79},
	year = {2023}}

@article{manski1999,
	author = {Manski, Charles F},
	journal = {Statistical Science},
	number = {3},
	pages = {279--281},
	title = {Choice as an Alternative to Control in Observational Studies: Comment},
	volume = {14},
	year = {1999}}

@article{fogarty2018a,
	author = {Fogarty, Colin B},
	journal = {Journal of the Royal Statistical Society: Series B (Statistical Methodology)},
	number = {5},
	pages = {1035--1056},
	title = {On Mitigating the Analytical Limitations of Finely Stratified Experiments},
	volume = {80},
	year = {2018}}

@article{rosenbaum1999a,
	author = {Rosenbaum, Paul R},
	journal = {Statistical Science},
	number = {3},
	pages = {259--304},
	title = {Choice as an Alternative to Control in Observational Studies},
	volume = {14},
	year = {1999}}

@article{rosenbaum1999b,
	author = {Rosenbaum, Paul R},
	journal = {Statistical Science},
	number = {3},
	pages = {300--304},
	title = {Choice as an Alternative to Control in Observational Studies: Rejoinder},
	volume = {14},
	year = {1999}}

@article{rosenbaum1987a,
	author = {Rosenbaum, Paul R},
	journal = {Biometrika},
	number = {1},
	pages = {13--26},
	title = {Sensitivity Analysis for Certain Permutation Inferences in Matched Observational Studies},
	volume = {74},
	year = {1987}}

@article{rosenbaumkrieger1990,
	author = {Rosenbaum, Paul R and Krieger, Abba M},
	journal = {Journal of the American Statistical Association},
	number = {410},
	pages = {493--498},
	title = {Sensitivity of Two-Sample Permutation Inferences in Observational Studies},
	volume = {85},
	year = {1990}}

@article{rosenbaum2018,
	author = {Rosenbaum, Paul R},
	journal = {Annals of Applied Statistics},
	number = {4},
	pages = {2312--2334},
	title = {Sensitivity Analysis for Stratified Comparisons in an Observational Study of the Effect of Smoking on Homocysteine Levels},
	volume = {12},
	year = {2018}}

@article{gastwirthetal2000,
	author = {Gastwirth, Joseph L and Krieger, Abba M and Rosenbaum, Paul R},
	journal = {Journal of the Royal Statistical Society: Series B (Statistical Methodology)},
	number = {3},
	pages = {545--555},
	title = {Asymptotic separability in sensitivity analysis},
	volume = {63},
	year = {2000}}

@article{haslanger2000,
	author = {Haslanger, Sally},
	journal = {No\^us},
	number = {1},
	pages = {31--55},
	title = {Gender and Race: (What) Are They? (What) Do We Want Them to Be?},
	volume = {34},
	year = {2000}}

@article{hainmuelleretal2014,
	author = {Hainmueller, Jens and Hopkins, Daniel J and Yamamoto, Teppei},
	journal = {Political Analysis},
	number = {1},
	pages = {1--30},
	title = {Causal Inference in Conjoint Analysis: Understanding Multidimensional Choices via Stated Preference Experiments},
	volume = {22},
	year = {2014}}

@article{zhangrubin2003,
	author = {Zhang, Junni L and Rubin, Donald B},
	journal = {Journal of Educational and Behavioral Statistics},
	number = {4},
	pages = {353--368},
	title = {Estimation of Causal Effects via Principal Stratification when Some Outcomes are Truncated by ``Death''},
	volume = {28},
	year = {2003}}

@book{rosenbaum2002a,
	address = {New York, NY},
	author = {Rosenbaum, Paul R},
	edition = {2nd},
	publisher = {Springer},
	title = {Observational Studies},
	year = {2002}}

@article{horvitzthompson1952,
	author = {Horvitz, Daniel G and Thompson, Donovan J},
	journal = {Journal of the American Statistical Association},
	number = {260},
	pages = {663--685},
	title = {A Generalization of Sampling without Replacement from a Finite Universe},
	volume = {47},
	year = {1952}}

@article{bookstein1989,
	author = {Bookstein, Fred L.},
	title = {Principal Warps: Thin-Plate Splines and the Decomposition of Deformations},
	journal = {{IEEE} Transactions on Pattern Analysis and Machine Intelligence},
	volume = {11},
	number = {6},
	pages = {567--585},
	year = {1989},
	doi = {10.1109/34.24792}}

@misc{macdonaldetal2016data,
	author = {Mac{D}onald, John and Fagan, Jeffrey and Geller, Amanda},
	title = {Replication Data and Code for ``The Effects of Local Police Surges on Crime and Arrests in {N}ew {Y}ork {C}ity''},
	year = {2016},
	howpublished = {\url{https://github.com/macdonaldjohn/Impact-Zone-Data}},
	note = {{S}tata replication code}}

\clearpage

\spacingset{1.9}
\setlength{\abovedisplayskip}{6pt plus 2pt minus 2pt}
\setlength{\belowdisplayskip}{6pt plus 2pt minus 2pt}
\setlength{\abovedisplayshortskip}{4pt plus 2pt minus 2pt}
\setlength{\belowdisplayshortskip}{4pt plus 2pt minus 2pt}

\titleformat{\paragraph}[runin]{\normalfont\normalsize\bfseries}{}{0pt}{}
\titlespacing*{\paragraph}{0pt}{2ex plus 0.5ex minus .2ex}{0.6em}

\setcounter{section}{0}
\setcounter{subsection}{0}
\setcounter{equation}{0}
\setcounter{figure}{0}
\setcounter{table}{0}
\setcounter{thm}{0}
\setcounter{prop}{0}
\setcounter{lem}{0}
\setcounter{cor}{0}
\setcounter{rmk}{0}
\renewcommand{\thesection}{S.\arabic{section}}
\renewcommand{\thesubsection}{S.\arabic{section}.\arabic{subsection}}
\renewcommand{\thesubsubsection}{S.\arabic{section}.\arabic{subsection}.\arabic{subsubsection}}
\renewcommand{\thefigure}{S.\arabic{figure}}
\renewcommand{\thetable}{S.\arabic{table}}
\numberwithin{equation}{section}
\numberwithin{thm}{section}
\numberwithin{prop}{section}
\numberwithin{lem}{section}
\numberwithin{cor}{section}

\phantomsection
\pdfbookmark[0]{Supplementary Material}{supplementary-material}
\begin{center}
{\Large\bfseries Supplementary Material}
\end{center}
\medskip

This supplement provides formal results, proofs, and additional discussion supporting the manuscript. Section~\ref{sec: supp causal comparison} argues that defining the causal comparison and conducting inference about it are separable tasks. Section~\ref{sec: supp prelim lemmas} establishes that potential outcomes reduce to functions of the civilian-race indicator alone within principal strata. Section~\ref{sec: supp formal section 3} presents the decomposition of the stratum-specific average causal effect among potentially stoppable encounters. Section~\ref{sec: supp formal section 4} presents results supporting Section~4 of the manuscript, including the equivalence of the augmented Difference-in-Means to the full-data Difference-in-Means, supporting lemmas on the hypergeometric distribution of the number of stopped control units and on conditional probabilities of treatment, and the finite-sample bias and consistency of the stratum-specific estimator $\hat{\tau}_g$. Section~\ref{sec: supp formal section 5} provides proofs of the lemma and propositions in Section~5 of the manuscript, a lemma establishing the equivalence of $\Gamma = 1$ and No-Bias-in-Encounters, efficient computation of probability bounds, the corollary extending the sensitivity analysis to stratum-specific bounds, and the conservative variance estimator.

\subsection*{Notation and setup}

Before presenting the formal results, we reiterate the notation and setup of the manuscript that are used throughout the supplement. We restrict attention to \textit{potentially stoppable encounters} --- encounters whose realized nonracial profile $\bm{v}_i$ belongs to $\mathcal{V}_i^{\mathrm{ps}} \coloneqq \mathcal{V}_i^{\mathrm{AS}} \cup \mathcal{V}_i^{\mathrm{OMS}} \cup \mathcal{V}_i^{\mathrm{OWS}}$, as defined in Section~3.2 of the manuscript. This restriction excludes the Never-Stop ($\mathrm{NS}$) principal stratum, since encounters whose nonracial profiles would not lead to a stop regardless of the civilian's race are structurally irrelevant to both stopping and use-of-force decisions.

Under Assumption~4 (No-Only-White-Stops), the Only-White-Stop ($\mathrm{OWS}$) principal stratum is further ruled out among the realized encounters, so every potentially stoppable encounter belongs to either the Always-Stop ($\mathrm{AS}$) or Only-Minority-Stop ($\mathrm{OMS}$) principal stratum --- that is, $r_{g,i} \in \{\mathrm{AS}, \mathrm{OMS}\}$ for all $i \in \{1, \ldots, n_g\}$ and all $g \in \mathcal{G}$. The $\mathrm{AS}$ label denotes encounters whose nonracial profile would lead to a stop regardless of the civilian's race ($s_{g,i}(1) = s_{g,i}(0) = 1$). The $\mathrm{OMS}$ label denotes encounters whose profile would lead to a stop only if the civilian is minority ($s_{g,i}(1) = 1$, $s_{g,i}(0) = 0$); in the latter case, a white civilian occupying the same encounter slot would not be stopped and therefore would not appear in police administrative data.

Within each stratum $g$, we index encounters by $i = 1, \ldots, n_g$, where $n_g$ is the total number of potentially stoppable encounters. We write $n_{g,1}$ and $n_{g,0} \coloneqq n_g - n_{g,1}$ for the minority-civilian and white-civilian counts, which are fixed across all assignments in $\Omega_g$ by the conditioning described in Section~3.4 of the manuscript. The principal stratum label $r_{g,i}$ is a fixed attribute of encounter slot $i$, so the total numbers of Always-Stop and Only-Minority-Stop encounters, $n_{g,\mathrm{AS}} \coloneqq \sum_{i=1}^{n_g} \mathbbm{1}\{r_{g,i} = \mathrm{AS}\}$ and $n_{g,\mathrm{OMS}} \coloneqq \sum_{i=1}^{n_g} \mathbbm{1}\{r_{g,i} = \mathrm{OMS}\}$, are also fixed across $\Omega_g$. However, the cross-tabulations $n_{g,z,r}(\bm{Z}_g)$ --- the numbers of encounters with civilian race $z$ and principal stratum $r$ --- are random variables under $\bm{Z}_g \in \Omega_g$, since different elements of $\Omega_g$ distribute the $n_{g,1}$ minority-civilian labels differently across Always-Stop and Only-Minority-Stop encounters. In particular, the number of stopped white-civilian encounters, $n_{g,0,\mathrm{AS}}(\bm{Z}_g) = \sum_{i:\,r_{g,i} = \mathrm{AS}} (1 - Z_{g,i})$, is a random variable under $\bm{Z}_g \in \Omega_g$.

The set of informative strata $\mathcal{G}^{\ast} \coloneqq \{g \in \mathcal{G}: n_{g,1} \geq 1 \text{ and } n_{g,0} \geq 1\}$ is as defined in Section~3.4 of the manuscript. All probabilities and expectations involving $\bm{Z}_g$ are conditional on the event $\bm{Z}_g \in \Omega_g$ (the assignment space defined in Section~3.4 of the manuscript). We leave this conditioning implicit throughout unless stated otherwise.

\medskip
\noindent\textit{Terminology.} Throughout the supplement, surrounding prose --- including section introductions, remarks, and motivating paragraphs --- uses the manuscript's application language: encounters with minority and white civilians among the potentially stoppable encounters within each stratum. In formal statements and proofs, we use standard causal-inference terminology: unit for encounter, treated for minority-civilian, and control for white-civilian. The principal strata labels --- Always-Stop ($\mathrm{AS}$), Only-Minority-Stop ($\mathrm{OMS}$), Only-White-Stop ($\mathrm{OWS}$), and Never-Stop ($\mathrm{NS}$) --- are retained throughout because they denote specific counterfactual stopping behaviors defined in Section 3.2. This convention emphasizes that the formal results are general statements, while the surrounding prose reflects the application.

\section{Specifying a Causal Target Versus Inferring It}\label{sec: supp causal comparison}

This section develops the manuscript's stance on the relation between two questions: \textit{which} contrast --- which particular comparison between potential outcomes --- to define as the target for causal inference, and \textit{whether} the defined contrast can be reliably inferred from the data. The first is a question of estimand choice. The second is a question of inference. We argue that the two are separable: defining what one is after and inferring it from the data are distinct tasks. The framework in the manuscript addresses the inference question for our chosen contrast, and the two confounding channels it addresses apply to any causal claim about race and use of force, so the same channel-decomposition logic could ground an analogous sensitivity analysis for a different target.

The fusion runs through \citet{heckman1998}'s influential framing. \citet{heckman1998} defines discrimination as ``a causal effect defined by a hypothetical \textit{ceteris paribus} conceptual experiment --- varying race but keeping all else constant'' \citep[][p.~102]{heckman1998}, so the all-else-equal contrast is simultaneously the target (varying race, holding nonracial attributes constant) and the proposed response to confounding. Confounding, on this framing, refers to what \citet{heckman1998} calls ``unobserved'' or ``omitted characteristics'' --- nonracial differences between individuals that could drive differential behavior toward them. The all-else-equal contrast has since been contested on substantive grounds in the broader literature \citep{hukohler-hausmann2025,hu2025}: holding nonracial attributes fixed across the racial comparison can be read as removing precisely the social position that racial categories index. We do not adjudicate this substantive debate. Our point is methodological and prior to it: target-definition and inference are separable tasks, so a framework that addresses inference can be used regardless of where any given researcher ultimately stands on the target.

To make the distinction concrete, consider two hypothetical experiments on officer use of force in which officers are assigned to vertical patrols of building stairwells. Each experiment uses the same coin flip per officer, but what the flip induces differs across the two. In the first experiment, the coin determines which of two stairwells the officer is sent to: on heads, a stairwell in an Upper East Side residential building staged with a white civilian; on tails, a stairwell in a Brownsville residential building staged with a minority civilian. In the second experiment, each officer has a fixed stairwell, and the coin determines which civilian is staged there: on heads, a white civilian; on tails, a minority civilian.

The two experiments target two different contrasts. The second is an all-else-equal contrast in the sense of \citet{heckman1998}: patrol context is held fixed across the racial comparison. The first is not: patrol context is permitted to move with race. Knowing only that a person is a racial minority, rather than white, corresponds to a greater probability that the person lives in Brownsville than on the Upper East Side; the first experiment's contrast preserves this indexing of social position --- here, neighborhood marginalization --- rather than controlling it out. Our point is not to adjudicate between the targets but to observe that the same two confounding channels --- encounter assignment and sample selection --- operate in both. Target-choice and response-to-confounding can therefore be separated. The rest of this section develops three points in turn: the contrast our framework adopts and why, the broader space of targets a researcher could adopt instead, and the formalization of the two confounding channels.

\subsection*{Adopting the Literature's Contrast}

We adopt the contrast implicit in the existing literature on racial discrimination in police use of force \citep{fryer2019,knoxetal2020}: an officer's use of force in an encounter that could be filled with a minority versus a white civilian, with the officer and patrol context held fixed through the stochastic path-intersection process of Section~2.1 of the manuscript. This contrast belongs to the all-else-equal family whose use as a target has been contested on substantive grounds in the broader literature \citep{hukohler-hausmann2025,hu2025}, and we do not defend it against those substantive arguments on first principles. The separability claim developed above is what lets us proceed without that defense: in our hands, the literature's contrast does only the work of target definition, and the manuscript's framework addresses confounding for that target independently of how the substantive grounds offered for it are eventually settled. A researcher persuaded by the substantive case against this contrast can target a different one; the same two confounding channels apply, and an analogous sensitivity analysis can be built on them.

Two considerations motivate the specific choice for this paper. First, this work is in direct conversation with the existing empirical literature on NYPD stop, question, and frisk practices \citep{fryer2019,knoxetal2020,gaebleretal2022}, which targets exactly this contrast. The methodological contribution of the manuscript is strongest when the inference it supports addresses the same object that literature has been debating; departing from that contrast for this paper would put the analysis outside that conversation without a payoff to set against the cost. Second, the empirical anchors for our two sensitivity parameters --- the plausible range $\underline{\rho} \in [0.32, 0.34]$ for discrimination in stops, drawn from \citet{goeletal2016} and \citet{gelmanetal2007} through \citet{knoxetal2020}, and the geographic ceiling $\Gamma_g^{\mathrm{geo}}(\xi)$ from 2010 Census block-group demographics --- are built on prior work that targets this same contrast. Anchoring our sensitivity analysis to that body of work places our calibrations in dialogue with prior estimates and makes the plausibility argument legible to readers of that literature.

The contrast also fixes certain nonracial attributes of the civilian. We weaken the all-else-equal condition as far as the principal stratification allows: Two civilians of different races who could occupy the same encounter are permitted to differ in their nonracial profiles $\bm{v}$, provided both profiles belong to the same principal stratum. Nonracial profiles may therefore vary across the minority-civilian and white-civilian conditions, but not in ways that would alter the officer's counterfactual stopping behavior.

The manuscript makes explicit an assumption that work in this literature \citep{fryer2019,knoxetal2020,gaebleretal2022,zhaoetal2022} leaves implicit: For each fixed race $z$, $y_i(z, \bm{v})$ is constant across all profiles $\bm{v}$ within a given principal stratum. The manuscript's contrast between races at a fixed principal stratum therefore takes the same value for every pair of profiles drawn from that principal stratum, and no marginalization over nonracial profiles is required. This contrast is a coarsened analogue of the component effect in the all-else-equal framework: Both hold a nonracial feature fixed across the racial comparison, but in the manuscript's contrast, that feature is the principal stratum rather than the full profile $\bm{v}$.

Which contrast is the right one to target is a substantive and often normative question that extends beyond statistical methodology \citep{hu2025}, and we do not attempt to settle it here. The next subsection surveys the broader space of targets a researcher could adopt instead, including the alternatives that motivate the criticisms of the all-else-equal frame. The final subsection formalizes the two confounding channels and shows that they apply regardless of which target is chosen.

\subsection*{The Space of Possible Targets}

The contrast our framework adopts is one specific point in the broader space of possible targets. The second experiment from the section opening exemplifies the type common in the empirical literature on discrimination: an all-else-equal contrast that holds nonracial attributes $\bm{v}$ fixed across race. Conjoint studies \citep{hainmuelleretal2014} formalize this convention with a progression of estimands, each building on the previous:
\begin{itemize}
    \item The \textit{component effect} (under Assumption 1 of No-Interference), in the notation of the manuscript, is $y_i(1, \bm{v}) - y_i(0, \bm{v})$: the contrast between decision-maker $i$'s responses to a minority ($z = 1$) and a white ($z = 0$) individual sharing the same nonracial profile $\bm{v}$.
    \item The \textit{marginal component effect}, $\sum_{\bm{v}} \left[y_i(1, \bm{v}) - y_i(0, \bm{v})\right] q(\bm{v})$, integrates the component effect over a researcher-specified probability mass function $q$ on $\bm{v}$, yielding a scalar contrast for each decision-maker $i$.
    \item The \textit{average marginal component effect}, $\frac{1}{N} \sum_{i=1}^{N} \sum_{\bm{v}} \left[y_i(1, \bm{v}) - y_i(0, \bm{v})\right] q(\bm{v})$, aggregates the marginal component effect across a population of $N$ decision-makers.
\end{itemize}

The first experiment exemplifies a different type: a contrast in which nonracial attributes are permitted to vary with race. The all-else-equal structure of the conjoint progression is a convention, not a definitional requirement. A researcher could target $y_i(1, \bm{v}) - y_i(0, \bm{v}^\prime)$ with $\bm{v} \neq \bm{v}^\prime$: the contrast between decision-maker $i$'s response to a minority individual with nonracial profile $\bm{v}$ and to a white individual with a different profile $\bm{v}^\prime$. A structurally different alternative abandons contrasts between fixed profiles altogether: Each potential outcome is first averaged under its own race-conditional distribution of $\bm{v}$, and the two race-specific averages are then contrasted --- reversing the all-else-equal order of contrast followed by marginalization. Letting $p_1(\bm{v})$ and $p_0(\bm{v})$ denote race-specific probability mass functions of $\bm{v}$, this alternative has two stages:
\begin{itemize}
    \item The \textit{within-race marginal effect}, $\sum_{\bm{v}} y_i(1, \bm{v}) p_1(\bm{v}) - \sum_{\bm{v}} y_i(0, \bm{v}) p_0(\bm{v})$: the contrast between decision-maker $i$'s expected response to a minority individual with nonracial profile distributed according to $p_1$ and the expected response to a white individual with profile distributed according to $p_0$.
    \item The \textit{average within-race marginal effect}, $\frac{1}{N} \sum_{i=1}^{N} \left[\sum_{\bm{v}} y_i(1, \bm{v}) p_1(\bm{v}) - \sum_{\bm{v}} y_i(0, \bm{v}) p_0(\bm{v})\right]$, aggregates the within-race marginal effect across a population of $N$ decision-makers.
\end{itemize}
Such contrasts preserve rather than remove the probabilistic differences in nonracial attributes that racial categories index. The distributions $p_1$ and $p_0$ may be specified, for example, to reflect empirical differences between racial groups in attributes of social position --- such as income, education, or residential patterns \citep{haslanger2000,hu2023,hukohler-hausmann2025,hu2025,taylor2004} --- although other specifications are possible.

\subsection*{The Confounding Channels Apply Regardless of Target}

In either experiment, the same two channels confound inference about the corresponding target. The first channel, \textit{encounter assignment}, is that the probability of tails --- the probability that the officer is sent to the minority-civilian condition --- may differ from one officer to another, possibly in ways associated with officers' force dispositions. The second, \textit{sample selection}, is that we observe whether an officer uses force only when the encounter results in a stop, and not all officers would stop the civilian under each race condition. Both channels are properties of the data-generating process rather than of the contrast, so they confound inference about the first experiment's target and about the second experiment's target in the same way. Encounter assignment confounds in both because officers vary in the probability of being assigned the minority-civilian condition, and that variation may correlate with variation in how those officers use force. Sample selection confounds in both because force is observed only when the encounter produces a stop, and stopping rates can vary across officers and across the race conditions.

The debate over which contrast is substantively or normatively appropriate may shape the interpretation of results, but the preceding argument shows that it does not change the two inferential problems we address in the manuscript: those problems arise from features of the data-generating process that the choice of contrast does not touch. Because the choice of contrast does not change what our framework does about confounding, our adoption of the literature's contrast does not commit us to its substantive or normative grounds. We adopt it because we build on and are in dialogue with an existing literature. A researcher who prefers a different target faces the same two confounding channels and could build an analogous sensitivity analysis addressing them.

\section{Reduction of Potential Outcomes to Functions of the Civilian-Race Indicator}\label{sec: supp prelim lemmas}

The following lemma shows that, under Assumption~1 (No Interference) and Assumption~3 (Use-of-Force Depends Only on Race within Principal Strata), the stopping and use-of-force potential outcomes both reduce to functions of the civilian-race indicator $z$ once we condition on the principal stratum to which the realized civilian's nonracial profile belongs. We use this reduction throughout the subsequent proofs.

\begin{lem}[Reduction of potential outcomes to functions of the civilian-race indicator]\label{lem: supp PO reduction}
Under Assumptions~1 (No Interference) and 3 (Use-of-Force Depends Only on Civilian-Race Indicator within Principal Strata), conditional on the principal stratum label $r_{g,i}$ of unit $i$ in stratum $g$:
\begin{enumerate}[label=\textup{(\alph*)}, leftmargin=2.5em]
\item The stopping potential outcome depends on $\bm{v}_{g,i}$ only through the principal stratum label $r_{g,i}$, so that $s_{g,i}(z, \bm{v}) = s_{g,i}(z, r_{g,i}) \coloneqq s_{g,i}(z)$ for any $\bm{v} \in \mathcal{V}_{g,i}^{r_{g,i}}$.
\item The use-of-force potential outcome depends on $\bm{v}_{g,i}$ only through $r_{g,i}$, so that $y_{g,i}(z, \bm{v}) = y_{g,i}(z, r_{g,i}) \coloneqq y_{g,i}(z)$ for any $\bm{v} \in \mathcal{V}_{g,i}^{r_{g,i}}$.
\end{enumerate}
Consequently, conditional on the principal stratum labels $\{r_{g,i}\}_{i=1}^{n_g}$, the assignment and potential outcomes within each stratum $g$ can be expressed solely in terms of the binary civilian-race indicator $z$.
\end{lem}

\begin{proof}
Part~(a) follows from Assumption~1 (No Interference) and the definition of the principal strata in Section~3.2 of the manuscript. The partition $\{\mathcal{V}_{g,i}^{\mathrm{AS}}, \mathcal{V}_{g,i}^{\mathrm{OMS}}, \mathcal{V}_{g,i}^{\mathrm{OWS}}, \mathcal{V}_{g,i}^{\mathrm{NS}}\}$ is defined by the counterfactual stopping behavior $(s_{g,i}(1, \bm{v}), s_{g,i}(0, \bm{v}))$: All $\bm{v}$ in the same element of the partition yield the same pair of stopping outcomes. Hence, conditional on $r_{g,i}$, the value of $s_{g,i}(z, \bm{v})$ does not depend on which $\bm{v} \in \mathcal{V}_{g,i}^{r_{g,i}}$ is realized.

Part~(b) is the content of Assumption~3, which states that for any $\bm{v}, \bm{v}' \in \mathcal{V}_{g,i}^{r_{g,i}}$, $y_{g,i}(z, \bm{v}) = y_{g,i}(z, \bm{v}')$. Thus $y_{g,i}(z, \bm{v})$ depends on $\bm{v}$ only through the principal stratum label.

The final claim follows from parts~(a) and~(b): once the labels $\{r_{g,i}\}$ are conditioned on, both $s_{g,i}$ and $y_{g,i}$ are functions of $z$ alone, and the assignment space $\Omega_g$ (defined in Section~3.4 of the manuscript) is defined solely in terms of the civilian-race indicator vector $\bm{z}_g$.
\end{proof}

\section{Formal Results and Proofs for Section 3}\label{sec: supp formal section 3}

\subsection{\texorpdfstring{$\Gamma = 1$ Implies No-Bias-in-Encounters}{Gamma = 1 Implies No-Bias-in-Encounters}} \label{sec: supp gamma one}

The manuscript states that $\Gamma = 1$ in the assignment model implies No-Bias-in-Encounters, which we now prove.

\begin{lem}[$\Gamma = 1$ Implies No-Bias-in-Encounters]\label{lem: supp gamma one}
Suppose the restriction on the assignment model described in Section~3.3 of the manuscript, in which $Z_{g,1}, \ldots, Z_{g,n_g}$ are independent Bernoulli random variables with $\Pr(Z_{g,i} = 1) = \pi_{g,i}$, and inference conditions on $\bm{Z}_g \in \Omega_g$. Write $\varphi_{g,i} \coloneqq \Pr(Z_{g,i} = 1 \given \bm{Z}_g \in \Omega_g)$ for the conditional treatment probability. If $\Gamma = 1$, then No-Bias-in-Encounters holds:
\begin{align*}
\varphi_{g,i} = \varphi_{g,j}
\quad \text{for all } i, j \in \{1, \ldots, n_g\}.
\end{align*}
\end{lem}

\begin{proof}
Suppose $\Gamma = 1$. The restriction on the assignment model in equation~(1) of the manuscript becomes
\begin{align*}
1 \leq \dfrac{\pi_{g,i}/(1 - \pi_{g,i})}{\pi_{g,j}/(1 - \pi_{g,j})} \leq 1 \quad \text{for all } i, j \in \{1, \ldots, n_g\},
\end{align*}
which implies $\pi_{g,i}/(1 - \pi_{g,i}) = \pi_{g,j}/(1 - \pi_{g,j})$ for all $i, j$. The function $x \mapsto x/(1-x)$ is strictly monotone on $(0, 1)$ --- and therefore injective --- so equality of odds implies equality of probabilities: $\pi_{g,i} = \pi_{g,j}$ for all $i, j$. Hence there exists $\pi \in (0,1)$ such that $\pi_{g,i} = \pi$ for all $i \in \{1, \ldots, n_g\}$.

For any $\bm{z}_g \in \{0,1\}^{n_g}$ with $\sum_{i=1}^{n_g} z_{g,i} = n_{g,1}$, independence of the $Z_{g,i}$ gives
\begin{align*}
\Pr(\bm{Z}_g = \bm{z}_g) = \prod_{i=1}^{n_g} \pi^{z_{g,i}}(1-\pi)^{1 - z_{g,i}} = \pi^{n_{g,1}}(1-\pi)^{n_{g,0}},
\end{align*}
which takes the same value for every $\bm{z}_g \in \Omega_g$. Conditioning on $\bm{Z}_g \in \Omega_g$ therefore yields the uniform distribution
\begin{align}\label{eq: supp uniform on Omega_g}
\Pr(\bm{Z}_g = \bm{z}_g \given \bm{Z}_g \in \Omega_g) = \dfrac{1}{|\Omega_g|}, \quad \bm{z}_g \in \Omega_g.
\end{align}

To compute $\varphi_{g,i}$, fix an encounter $i \in \{1, \ldots, n_g\}$ and write the conditional probability as
\begin{align*}
\varphi_{g,i} = \Pr(Z_{g,i} = 1 \given \bm{Z}_g \in \Omega_g) = \dfrac{\abs{\{\bm{z}_g \in \Omega_g: z_{g,i} = 1\}}}{\abs{\Omega_g}},
\end{align*}
which follows by summing \eqref{eq: supp uniform on Omega_g} over the subset of $\Omega_g$ with $z_{g,i} = 1$. An assignment $\bm{z}_g \in \Omega_g$ satisfies $z_{g,i} = 1$ if and only if the remaining $n_g - 1$ coordinates contain exactly $n_{g,1} - 1$ ones, so
\begin{align*}
\abs{\{\bm{z}_g \in \Omega_g: z_{g,i} = 1\}} & = \binom{n_g - 1}{n_{g,1} - 1} \text{ and } \abs{\Omega_g} = \binom{n_g}{n_{g,1}}.
\end{align*}
Therefore,
\begin{align*}
\varphi_{g,i} = \dfrac{\binom{n_g - 1}{n_{g,1} - 1}}{\binom{n_g}{n_{g,1}}} = \dfrac{n_{g,1}}{n_g}.
\end{align*}
Because this expression does not depend on $i$, $\varphi_{g,i} = \varphi_{g,j} = n_{g,1}/n_g$ for all $i, j \in \{1, \ldots, n_g\}$.
\end{proof}

\subsection{Decomposition of the Stratum-Specific ATE}\label{sec: supp decomp}


\begin{prop}[Stratum-specific ATE decomposition]\label{prop: supp decomp}
Under Assumptions~1--4, the stratum-specific ATE among potentially stoppable encounters in stratum $g \in \mathcal{G}$, as defined in equation~(5) of the manuscript, satisfies
\begin{align}\label{eq: supp bias-in-force decomp}
\tau_g = \bar{y}_g(1) - (1 - \rho_g)\,\bar{y}_g^{\mathrm{AS}}(0),
\end{align}
where
\begin{align*}
\bar{y}_g(z) &\coloneqq n_g^{-1}\sum_{i=1}^{n_g} y_{g,i}(z), \\
\bar{y}_g^{\mathrm{AS}}(z) &\coloneqq n_{g,\mathrm{AS}}^{-1} \sum_{i:\, r_{g,i} = \mathrm{AS}} y_{g,i}(z), \\
\rho_g & = \dfrac{n_{g,\mathrm{OMS}}}{n_g}.
\end{align*}
\end{prop}

\begin{proof}
Under Assumption~1 (No Interference), the ATE among potentially stoppable encounters in stratum $g$, $\tau_g = \bar{y}_g(1) - \bar{y}_g(0)$, where the notation $y_{g,i}(z)$ suppresses the dependence on the nonracial profile by Lemma~\ref{lem: supp PO reduction}. It remains to show that $\bar{y}_g(0) = (1 - \rho_g)\,\bar{y}_g^{\mathrm{AS}}(0)$.

Each unit has a fixed principal stratum label $r_{g,i}$ and, under Assumption~4 (No-Only-White-Stops), $r_{g,i} \in \{\mathrm{AS}, \mathrm{OMS}\}$ for all $i$, so we may partition the population average of control potential outcomes over the potentially stoppable encounters into exactly two principal strata:
\begin{align*}
\bar{y}_g(0) = \dfrac{1}{n_g}\sum_{i=1}^{n_g} y_{g,i}(0) = \dfrac{1}{n_g}\left[\sum_{\substack{i:\, r_{g,i} = \mathrm{AS}}} y_{g,i}(0) + \sum_{\substack{i:\, r_{g,i} = \mathrm{OMS}}} y_{g,i}(0)\right].
\end{align*}

By definition of the Only-Minority-Stop principal stratum, every encounter $i$ with $r_{g,i} = \mathrm{OMS}$ has $s_{g,i}(0) = 0$. Assumption~2 (No-Force-Without-Stop) then implies $y_{g,i}(0) \leq s_{g,i}(0) = 0$, and since $y_{g,i}(0) \in \{0, 1\}$, we have $y_{g,i}(0) = 0$ for all encounters with $r_{g,i} = \mathrm{OMS}$. Consequently,
\begin{align*}
\bar{y}_g(0) = \dfrac{n_{g,\mathrm{AS}}}{n_g}\,\bar{y}_g^{\mathrm{AS}}(0) = (1 - \rho_g)\,\bar{y}_g^{\mathrm{AS}}(0),
\end{align*}
where the second equality uses $n_{g,\mathrm{AS}} / n_g = 1 - \rho_g$, since Assumption~4 implies the potentially stoppable encounters in stratum $g$ are partitioned into $n_{g,\mathrm{AS}}$ Always-Stop and $n_{g,\mathrm{OMS}}$ Only-Minority-Stop encounters with $n_{g,\mathrm{AS}} + n_{g,\mathrm{OMS}} = n_g$. Substituting into $\tau_g = \bar{y}_g(1) - \bar{y}_g(0)$ yields \eqref{eq: supp bias-in-force decomp}.
\end{proof}

\section{Formal Results and Proofs for Section 4}\label{sec: supp formal section 4}

Throughout this section we use the plug-in estimators of $\bar{y}_g(1)$ and $\bar{y}_g(0)$:
\begin{align}
\hat{\bar{y}}_g(1) & \coloneqq \dfrac{\sum_{i=1}^{n_g} y_{g,i}(Z_{g,i})\, s_{g,i}(Z_{g,i})\, Z_{g,i}}{\sum_{i = 1}^{n_g} s_{g,i}(Z_{g,i})\, Z_{g,i}}, \label{eq: supp est bar y1 g} \\
\hat{\bar{y}}_g(0) & \coloneqq \dfrac{\sum_{i=1}^{n_g} y_{g,i}(Z_{g,i})\, s_{g,i}(Z_{g,i})\, (1 - Z_{g,i})}{\sum_{i = 1}^{n_g} s_{g,i}(Z_{g,i})\, (1 - Z_{g,i})}. \label{eq: supp est bar y0 g}
\end{align}
The sums range over all $n_g$ encounters in stratum $g$, but $s_{g,i}(Z_{g,i}) = 0$ for every encounter not in the dataset, so each ratio reduces to a sum over the observed encounters.

\subsection{Equivalence of the Augmented Difference-in-Means to the Full-Data Estimator}\label{sec: supp augmented DIM}

Under the sample selection structure described in the manuscript, police administrative data omit all encounters that were not stopped. Among the potentially stoppable encounters with white civilians, the missing observations are exactly those with $r_{g,i} = \mathrm{OMS}$ and $Z_{g,i} = 0$. By definition of the Only-Minority-Stop principal stratum, these encounters have $s_{g,i}(0) = 0$, and by Assumption 2 (No-Force-Without-Stop), their use-of-force outcomes are $y_{g,i}(0) = 0$. Consequently, if we knew the realized number of missing Only-Minority-Stop control encounters under a given assignment, we could reconstruct the full-data Difference-in-Means among all potentially stoppable encounters by appending the appropriate number of zeros to the white-civilian outcomes. \Cref{prop: supp aug DIM} below formalizes the construction for arbitrary stratum compositions and arbitrary postulated counts of missing Only-Minority-Stop control encounters; the manuscript's \Cref{fig: augmentation diagram} illustrates the construction graphically using the worked example from Section~5.1.

For convenience, we restate the augmented Difference-in-Means from Section~5.1 of the manuscript. With $\hat{\bar{y}}_g(1)$ the observed minority-civilian mean, the augmented Difference-in-Means is $\hat{\tau}_g^{\underline{\rho}_g} \coloneqq \hat{\bar{y}}_g(1) - \hat{\bar{y}}_g^{\underline{\rho}_g}(0)$, where the augmented white-civilian mean is
\begin{align}\label{eq: supp aug control mean}
\hat{\bar{y}}_g^{\underline{\rho}_g}(0) \coloneqq \dfrac{\sum_{i=1}^{n_g} y_{g,i}(Z_{g,i})\, s_{g,i}(Z_{g,i})\, (1 - Z_{g,i})}{\sum_{i=1}^{n_g} s_{g,i}(Z_{g,i})\, (1 - Z_{g,i}) + \tilde{n}_{g,0,\mathrm{OMS}}^{\underline{\rho}_g}}.
\end{align}

\begin{prop}[Augmented Difference-in-Means equals full-data estimator]\label{prop: supp aug DIM}
Under Assumptions~1, 2, and~4, fix a stratum $g \in \mathcal{G}^{\ast}$. On the event that the augmented count satisfies $\tilde{n}_{g,0,\mathrm{OMS}}^{\underline{\rho}_g} = n_{g,0,\mathrm{OMS}}(\bm{Z}_g)$, where
\begin{align*}
n_{g,0,\mathrm{OMS}}(\bm{Z}_g) \coloneqq \sum_{i:\,r_{g,i} = \mathrm{OMS}} (1 - Z_{g,i})
\end{align*}
is the realized number of Only-Minority-Stop control encounters under $\bm{Z}_g$, the augmented Difference-in-Means $\hat{\tau}_g^{\underline{\rho}_g}$ equals the full-data Difference-in-Means:
\begin{align}\label{eq: supp aug equals full}
\hat{\tau}_g^{\underline{\rho}_g} = \hat{\tau}_g^{\mathrm{full}},
\end{align}
where
\begin{align}\label{eq: supp full data DIM}
\hat{\tau}_g^{\mathrm{full}} \coloneqq \dfrac{\sum_{i=1}^{n_g} y_{g,i}(Z_{g,i})\, Z_{g,i}}{\sum_{i=1}^{n_g} Z_{g,i}} - \dfrac{\sum_{i=1}^{n_g} y_{g,i}(Z_{g,i})\, (1 - Z_{g,i})}{\sum_{i=1}^{n_g} (1 - Z_{g,i})}
\end{align}
computes treated and control means over all potentially stoppable encounters without conditioning on stop status.
\end{prop}

\begin{proof}
We show separately that the treated and control components of the augmented estimator $\hat{\tau}_g^{\underline{\rho}_g}$ coincide with those of $\hat{\tau}_g^{\mathrm{full}}$.

\medskip
\noindent\textit{Treated mean.} Under Assumptions~1 and~4, every potentially stoppable treated encounter is stopped, so $s_{g,i}(Z_{g,i}) = 1$ whenever $Z_{g,i} = 1$. Hence, the stop indicators in the treated mean $\hat{\bar{y}}_g(1)$ are redundant:
\begin{align*}
\dfrac{\sum_{i=1}^{n_g} y_{g,i}(Z_{g,i})\, s_{g,i}(Z_{g,i})\, Z_{g,i}}{\sum_{i=1}^{n_g} s_{g,i}(Z_{g,i})\, Z_{g,i}} = \dfrac{\sum_{i=1}^{n_g} y_{g,i}(Z_{g,i})\, Z_{g,i}}{\sum_{i=1}^{n_g} Z_{g,i}},
\end{align*}
which is the treated component of $\hat{\tau}_g^{\mathrm{full}}$ in \eqref{eq: supp full data DIM}.

\medskip
\noindent\textit{Control mean.} Under Assumptions~1, 2, and~4, the only unobserved control encounters are those with $r_{g,i} = \mathrm{OMS}$ and $Z_{g,i} = 0$. By definition of the Only-Minority-Stop principal stratum, $s_{g,i}(0) = 0$ for these encounters, and by Assumption~2, $y_{g,i}(0) = 0$. In the observed data, the control numerator in \eqref{eq: supp aug control mean} sums only over stopped control encounters (those with $r_{g,i} = \mathrm{AS}$ and $Z_{g,i} = 0$), and the denominator counts these stopped encounters plus the $\tilde{n}_{g,0,\mathrm{OMS}}^{\underline{\rho}_g}$ augmented encounters. On the event that $\tilde{n}_{g,0,\mathrm{OMS}}^{\underline{\rho}_g} = n_{g,0,\mathrm{OMS}}(\bm{Z}_g)$, the augmented denominator becomes
\begin{align*}
n_{g,0,\mathrm{AS}}(\bm{Z}_g) + n_{g,0,\mathrm{OMS}}(\bm{Z}_g) = n_{g,0} = \sum_{i=1}^{n_g} (1 - Z_{g,i}),
\end{align*}
where the first equality holds because, under Assumption~4, every control encounter is either Always-Stop or Only-Minority-Stop. The augmented numerator is unchanged because each appended encounter contributes $y_{g,i}(0) = 0$. Since the $n_{g,0,\mathrm{OMS}}(\bm{Z}_g)$ missing encounters also have $y_{g,i}(0) = 0$, the numerator equals $\sum_{i=1}^{n_g} y_{g,i}(Z_{g,i})(1 - Z_{g,i})$, and the augmented control mean equals
\begin{align*}
\dfrac{\sum_{i=1}^{n_g} y_{g,i}(Z_{g,i})\,(1 - Z_{g,i})}{\sum_{i=1}^{n_g} (1 - Z_{g,i})},
\end{align*}
which is the control component of $\hat{\tau}_g^{\mathrm{full}}$ in \eqref{eq: supp full data DIM}. Combining the treated and control components yields \eqref{eq: supp aug equals full}.
\end{proof}

\begin{rmk}
In practice, the realized number of Only-Minority-Stop control encounters $n_{g,0,\mathrm{OMS}}(\bm{z}_g)$ is unknown to the researcher, since it depends on both the unobserved principal stratum labels and the realized assignment. The sensitivity parameter $\underline{\rho}_g$ postulates a lower bound on the proportion of Only-Minority-Stop encounters, and Proposition~1 of the manuscript establishes the bijection between the researcher's posited count $\tilde{n}_{g,0,\mathrm{OMS}}^{\underline{\rho}_g}$ and this lower bound. Proposition~\ref{prop: supp aug DIM} establishes that if the posited count happens to equal the true realized count, the augmented estimator recovers the full-data estimator exactly.
\end{rmk}

\subsection{Supporting Lemmas}\label{sec: supp supporting lemmas}

The following five lemmas support the proofs of the finite-sample bias (Proposition~\ref{prop: supp bias}) and consistency (Proposition~\ref{prop: supp consistency}) of the stratum-specific estimator $\hat{\tau}_g$. Both proofs turn on the random variable $C_g(\bm{Z}_g)$, the number of stopped control units in stratum $g$: the finite-sample bias decomposes according to the realized value of $C_g(\bm{Z}_g)$, and consistency requires controlling its behavior as the stratum size grows.

Four of the lemmas characterize $C_g(\bm{Z}_g)$; the fifth characterizes the conditional treatment probabilities given $C_g(\bm{Z}_g)$. Lemma~\ref{lem: supp hypergeom} shows that $C_g(\bm{Z}_g)$ follows a hypergeometric distribution, a consequence of the uniform distribution on $\Omega_g$ under No-Bias-in-Encounters combined with the fact that the principal stratum labels $\{r_{g,i}\}$ are fixed slot attributes. The next three lemmas extract asymptotic consequences of this distributional result: Lemma~\ref{lem: supp Cg diverges} shows that $C_g(\bm{Z}_g) \xrightarrow{p} \infty$ under regularity conditions, Lemma~\ref{lem: supp prob zero denom} shows that the event $\mathcal{E}_g \coloneqq \{C_g(\bm{Z}_g) \geq 1\}$ --- on which $\hat{\tau}_g$ is well defined --- has probability approaching $1$, and Lemma~\ref{lem: supp EV inv Cg} shows that $\E[1/C_g(\bm{Z}_g) \given \mathcal{E}_g] \to 0$, which we use to bound the conditional variance of $\hat{\bar{y}}_g(0)$ in the consistency proof. Lemma~\ref{lem: supp cond trt prob} characterizes the conditional probabilities of treatment given the realized value of $C_g(\bm{Z}_g)$, which are used in the bias proof to compute the expectation of $\hat{\bar{y}}_g(1)$ and $\hat{\bar{y}}_g(0)$ stratum-by-stratum.

\begin{lem}[Hypergeometric distribution of the number of stopped controls]\label{lem: supp hypergeom}
Fix a stratum $g$ with $n_g$ units, of which $n_{g,\mathrm{AS}}$ are Always-Stop and $n_{g,\mathrm{OMS}}$ are Only-Minority-Stop. Under Assumptions~1 and~4, together with the No-Bias-in-Encounters condition in equation~(2) of the manuscript, define
\begin{align*}
C_g(\bm{Z}_g) \coloneqq \sum_{i:\, r_{g,i} = \mathrm{AS}} (1 - Z_{g,i}),
\end{align*}
the number of Always-Stop units assigned to control. Then
\begin{align}\label{eq: supp hypergeom}
C_g(\bm{Z}_g) \sim \mathrm{Hypergeometric}(n_g,\; n_{g,\mathrm{AS}},\; n_{g,0}),
\end{align}
with support
\begin{align}\label{eq: supp support Cg}
\{\max(0,\; n_{g,\mathrm{AS}} - n_{g,1}),\; \ldots,\; \min(n_{g,\mathrm{AS}},\; n_{g,0})\}.
\end{align}
\end{lem}
\begin{proof}
Under No-Bias-in-Encounters, the distribution of $\bm{Z}_g$ conditional on $\bm{Z}_g \in \Omega_g$ is uniform over $\Omega_g$. Since $\Omega_g$ consists of all binary vectors of length $n_g$ with exactly $n_{g,1}$ ones, the $n_{g,0} = n_g - n_{g,1}$ control assignments are a simple random sample without replacement from the $n_g$ encounters. Of these $n_g$ encounters, $n_{g,\mathrm{AS}}$ have $r_{g,i} = \mathrm{AS}$. The number of Always-Stop encounters assigned to the control condition --- i.e., $C_g(\bm{Z}_g)$ --- therefore follows a hypergeometric distribution with population size $n_g$, number of ``successes'' $n_{g,\mathrm{AS}}$, and draw size $n_{g,0}$.

The support of $C_g(\bm{Z}_g)$ is bounded above by both $n_{g,\mathrm{AS}}$ (total Always-Stop encounters) and $n_{g,0}$ (total control units), and bounded below by the number of Always-Stop encounters that must remain after the $n_{g,1}$ minority-civilian slots are filled, which is $\max(0, n_{g,\mathrm{AS}} - n_{g,1})$.

To verify the probability mass function, fix $c$ in the support of $C_g(\bm{Z}_g)$. The event $\{C_g(\bm{Z}_g) = c\}$ requires choosing $c$ control encounters from the $n_{g,\mathrm{AS}}$ Always-Stop encounters and $n_{g,0} - c$ control encounters from the $n_{g,\mathrm{OMS}}$ Only-Minority-Stop encounters. The number of such assignments is $\binom{n_{g,\mathrm{AS}}}{c}\binom{n_{g,\mathrm{OMS}}}{n_{g,0} - c}$. Since $\abs{\Omega_g} = \binom{n_g}{n_{g,0}}$, the uniform distribution yields
\begin{align*}
\Pr(C_g(\bm{Z}_g) = c) = \dfrac{\binom{n_{g,\mathrm{AS}}}{c}\binom{n_{g,\mathrm{OMS}}}{n_{g,0} - c}}{\binom{n_g}{n_{g,0}}},
\end{align*}
which is the hypergeometric probability mass function with the stated parameters.
\end{proof}

\begin{lem}[Divergence in probability of $C_g(\bm{Z}_g)$]\label{lem: supp Cg diverges}
Fix a stratum $g$ and suppose the conditions of Lemma~\ref{lem: supp hypergeom} hold. Under regularity conditions \textup{(R2)} and \textup{(R3)} of Proposition~\ref{prop: supp consistency}, $C_g(\bm{Z}_g) \xrightarrow{p} \infty$ as $n_g \to \infty$.
\end{lem}

\begin{proof}
By Lemma~\ref{lem: supp hypergeom}, $C_g(\bm{Z}_g) \sim \mathrm{Hypergeometric}(n_g, n_{g,\mathrm{AS}}, n_{g,0})$, which has expectation $\E[C_g(\bm{Z}_g)] = n_{g,0}\,n_{g,\mathrm{AS}}/n_g$ and variance
\begin{align*}
\Var\left[C_g(\bm{Z}_g)\right] = \underbrace{n_{g,0}\,\dfrac{n_{g,\mathrm{AS}}}{n_g}}_{= \E[C_g(\bm{Z}_g)]}\,\underbrace{\left(1 - \dfrac{n_{g,\mathrm{AS}}}{n_g}\right)}_{\text{proportion not AS}}\,\underbrace{\dfrac{n_g - n_{g,0}}{n_g - 1}}_{\text{finite population correction}}.
\end{align*}
Since the proportion not AS and the finite population correction are each at most $1$,
\begin{align}\label{eq: supp hypergeom var bound}
\Var\left[C_g(\bm{Z}_g)\right] \leq \E\left[C_g(\bm{Z}_g)\right].
\end{align}
Under (R2) -- (R3), $\E[C_g(\bm{Z}_g)] = n_{g,0} \cdot n_{g,\mathrm{AS}} / n_g \to \infty$ as $n_g \to \infty$.

Fix any $\delta \in (0, 1)$ and consider the event
\begin{align}\label{eq: supp deviation event I}
C_g(\bm{Z}_g) \leq \delta\,\E[C_g(\bm{Z}_g)].
\end{align}
Since $\delta < 1$, this event implies $C_g(\bm{Z}_g) \leq \E[C_g(\bm{Z}_g)]$, and hence
\begin{align*}
\abs{C_g(\bm{Z}_g) - \E[C_g(\bm{Z}_g)]} = \E[C_g(\bm{Z}_g)] - C_g(\bm{Z}_g) \geq (1 - \delta)\E[C_g(\bm{Z}_g)].
\end{align*}
Monotonicity of probability gives
\begin{align*}
\Pr\left(C_g(\bm{Z}_g) \leq \delta\,\E[C_g(\bm{Z}_g)]\right) \leq \Pr\left(\abs{C_g(\bm{Z}_g) - \E[C_g(\bm{Z}_g)]} \geq (1 - \delta)\E[C_g(\bm{Z}_g)]\right).
\end{align*}
Chebyshev's inequality with $\varrho \coloneqq (1 - \delta)\E[C_g(\bm{Z}_g)]$ --- which is eventually strictly positive by (R2)--(R3) --- and the variance bound \eqref{eq: supp hypergeom var bound} then yield
\begin{align*}
\Pr\left(C_g(\bm{Z}_g) \leq \delta\,\E[C_g(\bm{Z}_g)]\right) \leq \dfrac{\Var[C_g(\bm{Z}_g)]}{(1 - \delta)^2\,\E[C_g(\bm{Z}_g)]^2} \leq \dfrac{1}{(1 - \delta)^2\,\E[C_g(\bm{Z}_g)]} \to 0.
\end{align*}

Because $\E[C_g(\bm{Z}_g)] \to \infty$, for any fixed constant $B > 0$ we eventually have $B < \delta\,\E[C_g(\bm{Z}_g)]$, so the event $\{C_g(\bm{Z}_g) \leq B\}$ is contained in $\{C_g(\bm{Z}_g) \leq \delta\,\E[C_g(\bm{Z}_g)]\}$. Monotonicity of probability gives $\Pr(C_g(\bm{Z}_g) \leq B) \to 0$ for every fixed $B > 0$. By the definition of convergence in probability to $\infty$, $C_g(\bm{Z}_g) \xrightarrow{p} \infty$.
\end{proof}

\begin{lem}[Vanishing probability of zero stopped controls]\label{lem: supp prob zero denom}
Fix a stratum $g$ and suppose the conditions of Lemma~\ref{lem: supp hypergeom} hold. Under regularity conditions \textup{(R2)} and \textup{(R3)} of Proposition~\ref{prop: supp consistency}, $\Pr(C_g(\bm{Z}_g) = 0) \to 0$ as $n_g \to \infty$.
\end{lem}

\begin{proof}
By the hypergeometric model in Lemma~\ref{lem: supp hypergeom},
\begin{align}\label{eq: supp hypergeom c0}
\Pr(C_g(\bm{Z}_g) = 0) = \dfrac{\binom{n_{g,\mathrm{OMS}}}{n_{g,0}}}{\binom{n_g}{n_{g,0}}}.
\end{align}
Writing the ratio of binomial coefficients in sequential form yields
\begin{align*}
\Pr(C_g(\bm{Z}_g) = 0) = \prod_{t=0}^{n_{g,0} - 1} \dfrac{n_{g,\mathrm{OMS}} - t}{n_g - t}.
\end{align*}
Each factor satisfies $(n_{g,\mathrm{OMS}} - t)/(n_g - t) \leq n_{g,\mathrm{OMS}}/n_g = 1 - n_{g,\mathrm{AS}}/n_g$, so the product is bounded above by $(1 - n_{g,\mathrm{AS}}/n_g)^{n_{g,0}}$. Under (R2) -- (R3), $n_{g,\mathrm{AS}}/n_g \to \alpha_g^{\mathrm{AS}} \in (0, 1]$ and $n_{g,0} \to \infty$ as $n_g \to \infty$, so $(1 - n_{g,\mathrm{AS}}/n_g)^{n_{g,0}} \to 0$.
\end{proof}

\begin{lem}[Vanishing expectation of $1/C_g(\bm{Z}_g)$]\label{lem: supp EV inv Cg}
Fix a stratum $g$ and suppose the conditions of Lemma~\ref{lem: supp Cg diverges} hold. Then $\E[1/C_g(\bm{Z}_g) \given \mathcal{E}_g] \to 0$ as $n_g \to \infty$, where $\mathcal{E}_g \coloneqq \{C_g(\bm{Z}_g) \geq 1\}$.
\end{lem}

\begin{proof}
By Lemma~\ref{lem: supp Cg diverges}, $C_g(\bm{Z}_g) \xrightarrow{p} \infty$, so for any fixed $\varepsilon > 0$,
\begin{align}\label{eq: supp Cg conv prob III}
\Pr\left(\dfrac{1}{C_g(\bm{Z}_g)} > \varepsilon\right) = \Pr\left(C_g(\bm{Z}_g) < \dfrac{1}{\varepsilon}\right) \to 0.
\end{align}
Decompose $1/C_g(\bm{Z}_g)$ using indicator functions for the partition $\{1/C_g(\bm{Z}_g) \leq \varepsilon\}$ and $\{1/C_g(\bm{Z}_g) > \varepsilon\}$. Linearity of expectation conditional on $\mathcal{E}_g$ gives
\begin{equation}\label{eq: supp EV decomp}
\begin{split}
\E\left[\dfrac{1}{C_g(\bm{Z}_g)} \given \mathcal{E}_g\right] &= \E\left[\dfrac{1}{C_g(\bm{Z}_g)}\,\mathbbm{1}\left\{\dfrac{1}{C_g(\bm{Z}_g)} \leq \varepsilon\right\} \given \mathcal{E}_g\right] \\
&\quad + \E\left[\dfrac{1}{C_g(\bm{Z}_g)}\,\mathbbm{1}\left\{\dfrac{1}{C_g(\bm{Z}_g)} > \varepsilon\right\} \given \mathcal{E}_g\right].
\end{split}
\end{equation}
On the event $\{1/C_g(\bm{Z}_g) \leq \varepsilon\}$, the first term on the right-hand side of \eqref{eq: supp EV decomp} is at most $\varepsilon$. For the second term, $0 \leq 1/C_g(\bm{Z}_g) \leq 1$ on $\mathcal{E}_g$, so
\begin{align*}
\E\left[\dfrac{1}{C_g(\bm{Z}_g)}\,\mathbbm{1}\left\{\dfrac{1}{C_g(\bm{Z}_g)} > \varepsilon\right\} \given \mathcal{E}_g\right] \leq \Pr\left(\dfrac{1}{C_g(\bm{Z}_g)} > \varepsilon \given \mathcal{E}_g\right).
\end{align*}
Combining,
\begin{align}\label{eq: supp EV decomp bound}
\E\left[\dfrac{1}{C_g(\bm{Z}_g)} \given \mathcal{E}_g\right] \leq \varepsilon + \Pr\left(\dfrac{1}{C_g(\bm{Z}_g)} > \varepsilon \given \mathcal{E}_g\right).
\end{align}

By the definition of conditional probability and monotonicity of probability,
\begin{align*}
\Pr\left(\dfrac{1}{C_g(\bm{Z}_g)} > \varepsilon \given \mathcal{E}_g\right) = \dfrac{\Pr(1/C_g(\bm{Z}_g) > \varepsilon,\; \mathcal{E}_g)}{\Pr(\mathcal{E}_g)} \leq \dfrac{\Pr(1/C_g(\bm{Z}_g) > \varepsilon)}{\Pr(\mathcal{E}_g)}.
\end{align*}
By \eqref{eq: supp Cg conv prob III} and Lemma~\ref{lem: supp prob zero denom} --- which gives $\Pr(\mathcal{E}_g) \to 1$ --- the right-hand side converges to $0$. Taking $\limsup$ on both sides of \eqref{eq: supp EV decomp bound} yields
\begin{align*}
\limsup_{n_g \to \infty} \E\left[\dfrac{1}{C_g(\bm{Z}_g)} \given \mathcal{E}_g\right] \leq \varepsilon \quad \text{for every } \varepsilon > 0,
\end{align*}
and since $\E[1/C_g(\bm{Z}_g) \given \mathcal{E}_g] \geq 0$, this implies $\E[1/C_g(\bm{Z}_g) \given \mathcal{E}_g] \to 0$.
\end{proof}

\begin{lem}[Conditional probabilities of treatment given number of stopped control units]\label{lem: supp cond trt prob}
Fix a stratum $g$ and suppose Assumptions~1 and~4, together with the No-Bias-in-Encounters condition in equation~(2) of the manuscript. Let $C_g(\bm{Z}_g) = c$ for some $c$ in the support of $C_g(\bm{Z}_g)$ given in \eqref{eq: supp support Cg}. Then, for any unit $i$ with $r_{g,i} = \mathrm{AS}$,
\begin{align}\label{eq: supp cond prob AS}
\Pr\left(Z_{g,i} = 1 \given \bm{Z}_g \in \Omega_g,\; C_g(\bm{Z}_g) = c,\; r_{g,i} = \mathrm{AS}\right) = \dfrac{n_{g,\mathrm{AS}} - c}{n_{g,\mathrm{AS}}},
\end{align}
and, for any unit $i$ with $r_{g,i} = \mathrm{OMS}$,
\begin{align}\label{eq: supp cond prob OMS}
\Pr\left(Z_{g,i} = 1 \given \bm{Z}_g \in \Omega_g,\; C_g(\bm{Z}_g) = c,\; r_{g,i} = \mathrm{OMS}\right) = \dfrac{n_{g,1} - (n_{g,\mathrm{AS}} - c)}{n_{g,\mathrm{OMS}}}.
\end{align}
\end{lem}

\begin{proof}
Define the set of assignments consistent with $C_g(\bm{Z}_g) = c$ as
\begin{align*}
\Omega_g(c) \coloneqq \{\bm{z}_g \in \Omega_g: C_g(\bm{z}_g) = c\}.
\end{align*}
Under No-Bias-in-Encounters, the distribution of $\bm{Z}_g$ conditional on $\bm{Z}_g \in \Omega_g$ is uniform on $\Omega_g$. For any $\bm{z}_g \in \Omega_g(c) \subseteq \Omega_g$, the definition of conditional probability gives
\begin{align*}
\Pr(\bm{Z}_g = \bm{z}_g \given \bm{Z}_g \in \Omega_g,\; C_g(\bm{Z}_g) = c) = \dfrac{\Pr(\bm{Z}_g = \bm{z}_g,\; C_g(\bm{Z}_g) = c \given \bm{Z}_g \in \Omega_g)}{\Pr(C_g(\bm{Z}_g) = c \given \bm{Z}_g \in \Omega_g)}.
\end{align*}
Since $\bm{z}_g \in \Omega_g(c)$ implies $C_g(\bm{z}_g) = c$, the joint event $\{\bm{Z}_g = \bm{z}_g,\; C_g(\bm{Z}_g) = c\}$ reduces to $\{\bm{Z}_g = \bm{z}_g\}$, so the numerator equals $1/\abs{\Omega_g}$ by the uniform distribution on $\Omega_g$. The denominator equals $\abs{\Omega_g(c)}/\abs{\Omega_g}$, since $\{C_g(\bm{Z}_g) = c\}$ corresponds to the $\abs{\Omega_g(c)}$ elements of $\Omega_g$ with $C_g(\bm{z}_g) = c$. Taking the ratio,
\begin{align*}
\Pr(\bm{Z}_g = \bm{z}_g \given \bm{Z}_g \in \Omega_g,\; C_g(\bm{Z}_g) = c) = \dfrac{1/\abs{\Omega_g}}{\abs{\Omega_g(c)}/\abs{\Omega_g}} = \dfrac{1}{\abs{\Omega_g(c)}} \quad \text{for all } \bm{z}_g \in \Omega_g(c).
\end{align*}

Under Assumptions~1 and~4, the potentially stoppable units partition into $n_{g,\mathrm{AS}}$ Always-Stop and $n_{g,\mathrm{OMS}}$ Only-Minority-Stop units. Conditional on $C_g(\bm{Z}_g) = c$, exactly $c$ Always-Stop units are control and $n_{g,\mathrm{AS}} - c$ are treated, while $n_{g,0} - c$ Only-Minority-Stop units are control and $n_{g,1} - (n_{g,\mathrm{AS}} - c)$ are treated.

Fix a unit $i$ with $r_{g,i} = \mathrm{AS}$. If $Z_{g,i} = 1$ and $C_g(\bm{Z}_g) = c$, then exactly $c$ of the remaining $n_{g,\mathrm{AS}} - 1$ Always-Stop units must be assigned to control, and $n_{g,0} - c$ control assignments must come from the $n_{g,\mathrm{OMS}}$ Only-Minority-Stop units. The number of assignments in $\Omega_g(c)$ with $z_{g,i} = 1$ is
\begin{align*}
\abs{\{\bm{z}_g \in \Omega_g(c): z_{g,i} = 1\}} = \binom{n_{g,\mathrm{AS}} - 1}{c}\binom{n_{g,\mathrm{OMS}}}{n_{g,0} - c}.
\end{align*}
The total number of assignments in $\Omega_g(c)$ is
\begin{align*}
\abs{\Omega_g(c)} = \binom{n_{g,\mathrm{AS}}}{c}\binom{n_{g,\mathrm{OMS}}}{n_{g,0} - c}.
\end{align*}
The uniform distribution on $\Omega_g(c)$ then implies
\begin{align*}
\Pr(Z_{g,i} = 1 \given \bm{Z}_g \in \Omega_g,\; C_g(\bm{Z}_g) = c,\; r_{g,i} = \mathrm{AS}) = \dfrac{\binom{n_{g,\mathrm{AS}} - 1}{c}}{\binom{n_{g,\mathrm{AS}}}{c}} = \dfrac{n_{g,\mathrm{AS}} - c}{n_{g,\mathrm{AS}}},
\end{align*}
which establishes \eqref{eq: supp cond prob AS}.

Turning to \eqref{eq: supp cond prob OMS}, fix a unit $i$ with $r_{g,i} = \mathrm{OMS}$. Conditional on $C_g(\bm{Z}_g) = c$, exactly $n_{g,\mathrm{AS}} - c$ Always-Stop units are assigned to treatment. Since $\bm{Z}_g \in \Omega_g$ implies $\sum_{j = 1}^{n_g} Z_{g,j} = n_{g,1}$, the number of treated Only-Minority-Stop units is $n_{g,1} - (n_{g,\mathrm{AS}} - c)$. The uniform distribution on $\Omega_g(c)$ implies that these treatment assignments are selected uniformly at random from the $n_{g,\mathrm{OMS}}$ Only-Minority-Stop units. Consequently,
\begin{align*}
\Pr(Z_{g,i} = 1 \given \bm{Z}_g \in \Omega_g,\; C_g(\bm{Z}_g) = c,\; r_{g,i} = \mathrm{OMS}) = \dfrac{\binom{n_{g,\mathrm{OMS}} - 1}{n_{g,1} - (n_{g,\mathrm{AS}} - c) - 1}}{\binom{n_{g,\mathrm{OMS}}}{n_{g,1} - (n_{g,\mathrm{AS}} - c)}} = \dfrac{n_{g,1} - (n_{g,\mathrm{AS}} - c)}{n_{g,\mathrm{OMS}}},
\end{align*}
which establishes \eqref{eq: supp cond prob OMS}, thereby completing the proof.
\end{proof}

\subsection{Finite-Sample Bias of the Stratum-Specific Estimator}\label{sec: supp bias}

\begin{prop}[Finite-sample bias of $\hat{\tau}_g$]\label{prop: supp bias}
Fix a stratum $g \in \mathcal{G}^{\ast}$ and let $\mathcal{E}_g \coloneqq \{C_g(\bm{Z}_g) \geq 1\}$ denote the event that stratum $g$ contains at least one stopped control unit, so that $\hat{\tau}_g$ in equation~(8) of the manuscript is well defined. Under Assumptions 1--4, together with the No-Bias-in-Encounters condition in equation~(2) of the manuscript, the bias of $\hat{\tau}_g$ for $\tau_g$ among potentially stoppable units, conditional on $\mathcal{E}_g$, is
\footnotesize
\begin{align}\label{eq: supp bias tau hat}
\E\left[\hat{\tau}_g \given \mathcal{E}_g\right] - \tau_g = \sum_{c \geq 1} \left(\dfrac{n_{g,\mathrm{AS}} - c}{n_{g,1}} - (1 - \rho_g)\right)\left[\bar{y}_g^{\mathrm{AS}}(1) - \bar{y}_g^{\mathrm{OMS}}(1)\right] \Pr\left(C_g(\bm{Z}_g) = c \given \mathcal{E}_g\right),
\end{align}
\normalsize
where $\Pr(C_g(\bm{Z}_g) = c \given \mathcal{E}_g)$ is the hypergeometric distribution given in \eqref{eq: supp hypergeom} of Lemma~\ref{lem: supp hypergeom}, restricted to $c \geq 1$.
\end{prop}

\begin{proof}
Fix $c \geq 1$ in the support of $C_g(\bm{Z}_g)$ and condition on $\{C_g(\bm{Z}_g) = c\}$. Under Assumption~4 (No-Only-White-Stops), every potentially stoppable treated unit is stopped, so the denominator of $\hat{\bar{y}}_g(1)$ in \eqref{eq: supp est bar y1 g} equals $n_{g,1}$ over all $\bm{z}_g \in \Omega_g$. By construction of $\hat{\bar{y}}_g(0)$, the stopped control units are exactly the $c$ Always-Stop control units, so
\begin{align*}
\E\left[\hat{\bar{y}}_g(0) \given C_g(\bm{Z}_g) = c\right] = \E\left[\dfrac{1}{c}\sum_{i:\, r_{g,i} = \mathrm{AS}} (1 - Z_{g,i})\,y_{g,i}(0) \given C_g(\bm{Z}_g) = c\right],
\end{align*}
which, by linearity of expectation and Lemma~\ref{lem: supp cond trt prob}, equals $\bar{y}_g^{\mathrm{AS}}(0)$.

Turning to $\hat{\bar{y}}_g(1)$, by linearity of expectation together with Assumption~1 and Lemma~\ref{lem: supp PO reduction},
\begin{align*}
\E\left[\hat{\bar{y}}_g(1) \given C_g(\bm{Z}_g) = c\right] = \dfrac{1}{n_{g,1}}\sum_{i=1}^{n_g} \E\left[Z_{g,i} \given C_g(\bm{Z}_g) = c\right] y_{g,i}(1).
\end{align*}
By Lemma~\ref{lem: supp cond trt prob}, for $i$ with $r_{g,i} = \mathrm{AS}$,
\begin{align*}
\E\left[Z_{g,i} \given C_g(\bm{Z}_g) = c\right] = \dfrac{n_{g,\mathrm{AS}} - c}{n_{g,\mathrm{AS}}},
\end{align*}
and for $i$ with $r_{g,i} = \mathrm{OMS}$,
\begin{align*}
\E\left[Z_{g,i} \given C_g(\bm{Z}_g) = c\right] = \dfrac{n_{g,1} - (n_{g,\mathrm{AS}} - c)}{n_{g,\mathrm{OMS}}}.
\end{align*}
Substituting and partitioning the sum over $\{i: r_{g,i} = \mathrm{AS}\}$ and $\{i: r_{g,i} = \mathrm{OMS}\}$ yields
\begin{align*}
\E\left[\hat{\bar{y}}_g(1) \given C_g(\bm{Z}_g) = c\right] = \dfrac{n_{g,\mathrm{AS}} - c}{n_{g,1}}\,\bar{y}_g^{\mathrm{AS}}(1) + \dfrac{n_{g,1} - (n_{g,\mathrm{AS}} - c)}{n_{g,1}}\,\bar{y}_g^{\mathrm{OMS}}(1).
\end{align*}

Since $\tau_g = \bar{y}_g(1) - (1 - \rho_g)\,\bar{y}_g^{\mathrm{AS}}(0)$ by Proposition~\ref{prop: supp decomp} and $\E[\hat{\bar{y}}_g(0) \given C_g(\bm{Z}_g) = c] = \bar{y}_g^{\mathrm{AS}}(0)$, it follows that
\begin{align*}
\E\left[\hat{\tau}_g \given C_g(\bm{Z}_g) = c\right] - \tau_g = \left(\dfrac{n_{g,\mathrm{AS}} - c}{n_{g,1}} - (1 - \rho_g)\right)\left[\bar{y}_g^{\mathrm{AS}}(1) - \bar{y}_g^{\mathrm{OMS}}(1)\right].
\end{align*}

Applying iterated expectations --- first conditional on $C_g(\bm{Z}_g) = c$ and then averaging over $c \geq 1$ according to the hypergeometric distribution in Lemma~\ref{lem: supp hypergeom} conditional on $\mathcal{E}_g$ --- yields \eqref{eq: supp bias tau hat}.
\end{proof}

\begin{rmk}\label{rmk: supp bias zero rho}
An instructive special case arises when $\rho_g = 0$. Under Assumption~4, $\rho_g = n_{g,\mathrm{OMS}} / n_g$, so $\rho_g = 0$ implies that every potentially stoppable encounter in stratum $g$ is Always-Stop and $n_{g,\mathrm{AS}} = n_g$. In this case, $n_{g,\mathrm{AS}} - C_g(\bm{Z}_g) = n_{g,1}$ deterministically for any $\bm{z}_g \in \Omega_g$. The ratio $(n_{g,\mathrm{AS}} - c)/n_{g,1}$ therefore equals $1 = 1 - \rho_g$ for every $c$, so each summand in \eqref{eq: supp bias tau hat} is zero. Hence $\E[\hat{\tau}_g \given \mathcal{E}_g] = \tau_g$, and $\hat{\tau}_g$ is unbiased for $\tau_g$ in stratum $g \in \mathcal{G}^{\ast}$.
\end{rmk}

\begin{rmk}\label{rmk: supp bias equal PO}
Another case in which the bias is zero arises when $\bar{y}_g^{\mathrm{AS}}(1) = \bar{y}_g^{\mathrm{OMS}}(1)$. Substantively, this requires that the average level of force used in stops that would occur regardless of race equals the average level of force used in stops that would occur only for minority civilians. Under the typology of \citet{knoxetal2020}, Always-Stop encounters correspond to situations in which civilians are stopped due to suspicion of, for example, violent crime, whereas Only-Minority-Stop encounters correspond to situations in which stops arise from lower-threshold cues such as ``furtive movements'' (a checkbox on the NYPD's UF-250 stop report form). Officers are more likely to use force in encounters prompted by violent crime than in encounters prompted by lower-threshold cues, so these two averages are unlikely to be equal.
\end{rmk}

\subsection{Consistency of the Stratum-Specific Estimator}\label{sec: supp consistency}

We study the behavior of $\hat{\tau}_g$ as the stratum size $n_g$ grows, holding the assumptions and conditioning fixed throughout the asymptotic sequence. (All quantities --- $n_g$, $n_{g,\mathrm{AS}}$, $y_{g,i}(z)$, and so on --- carry an implicit sequence index, which we suppress to match the manuscript's notation.) In particular, $n_g$, $n_{g,1}$, $n_{g,0}$, $n_{g,\mathrm{AS}}$, and $n_{g,\mathrm{OMS}}$ are fixed across $\Omega_g$ at every stage, while the cross-tabulations $n_{g,z,r}(\bm{Z}_g)$ --- the numbers of encounters with civilian race $z$ and principal stratum $r$ --- and $C_g(\bm{Z}_g)$ vary across $\Omega_g$.

The proof of consistency draws on two additional lemmas concerning the asymptotic behavior of $C_g(\bm{Z}_g)$. Lemma~\ref{lem: supp Cg diverges} shows that $C_g(\bm{Z}_g) \xrightarrow{p} \infty$ under the regularity conditions, by showing that $C_g(\bm{Z}_g)$ is unlikely to fall far below its expectation and that this expectation diverges. Lemma~\ref{lem: supp EV inv Cg} then uses this divergence, together with Lemma~\ref{lem: supp prob zero denom}, to show that $\E[1/C_g(\bm{Z}_g) \given \mathcal{E}_g] \to 0$, where $\mathcal{E}_g \coloneqq \{C_g(\bm{Z}_g) \geq 1\}$ is the event that stratum $g$ contains at least one stopped white-civilian encounter. The white-civilian variance bound in the proof below uses these two lemmas directly.

The proof proceeds in four steps. First, under Assumptions~1 and~4 and No-Bias-in-Encounters, minority-civilian encounters form a simple random sample within each stratum, so the minority-civilian mean is unbiased and consistent for $\bar{y}_g(1)$ by a standard finite population variance argument. Second, the white-civilian mean is an average over Always-Stop encounters; conditioning on $\mathcal{E}_g$, the conditional treatment probabilities in Lemma~\ref{lem: supp cond trt prob} imply that the $c$ stopped white-civilian encounters are a simple random sample from the Always-Stop subpopulation, from which a standard finite population variance calculation yields the conditional expectation $\bar{y}_g^{\mathrm{AS}}(0)$ and a conditional variance bound proportional to $\E[1/C_g(\bm{Z}_g) \given \mathcal{E}_g]$. Third, Lemma~\ref{lem: supp EV inv Cg} implies that this conditional variance bound converges to $0$. Fourth, Chebyshev's inequality establishes consistency of the white-civilian mean for $\bar{y}_g^{\mathrm{AS}}(0)$ conditional on $\mathcal{E}_g$; the decomposition in Proposition~\ref{prop: supp decomp} together with the continuous mapping theorem yields consistency of $\hat{\tau}_g$ for $\tau_g$ under this conditioning; and Lemma~\ref{lem: supp prob zero denom} removes the conditioning to conclude unconditional consistency.

\begin{prop}[Consistency of $\hat{\tau}_g$ for $\tau_g$]\label{prop: supp consistency}
Suppose Assumptions~1 -- 4 and the No-Bias-in-Encounters condition in equation~(2) of the manuscript hold for each element of the asymptotic sequence. Fix $g \in \mathcal{G}$ and suppose $n_g \to \infty$ subject to:
\begin{enumerate}[label=\textup{(R\arabic*)}, leftmargin=2.5em]
\item \textbf{Bounded potential outcomes.} There exists $M < \infty$ such that $\abs{y_{g,i}(z)} \leq M$ for all $i$ and $z \in \{0,1\}$, uniformly along the asymptotic sequence.
\item \textbf{Nondegenerate racial composition.} $n_{g,1} / n_g \to v_g \in (0, 1)$ as $n_g \to \infty$, so that $n_{g,1} \to \infty$ and $n_{g,0} \to \infty$ as $n_g \to \infty$.
\item \textbf{Nonvanishing Always-Stop share.} $n_{g,\mathrm{AS}} / n_g \to \alpha_g^{\mathrm{AS}} \in (0, 1]$ as $n_g \to \infty$, so that $n_{g,\mathrm{AS}} \to \infty$ as $n_g \to \infty$.
\end{enumerate}
Then $\hat{\tau}_g \xrightarrow{p} \tau_g$.
\end{prop}

\begin{proof}
Fix a stratum $g$. Under Assumptions~1 and~4, $\hat{\bar{y}}_g(1)$ in \eqref{eq: supp est bar y1 g} is
\begin{align*}
\hat{\bar{y}}_g(1) = \dfrac{1}{n_{g,1}}\sum_{i=1}^{n_g} Z_{g,i}\, y_{g,i}(1),
\end{align*}
which, under No-Bias-in-Encounters, is a sample mean from a simple random sample without replacement of size $n_{g,1}$ from the $n_g$ potentially stoppable units. The expected value and variance are
\begin{align*}
\E\left[\hat{\bar{y}}_g(1)\right] = \bar{y}_g(1)
\end{align*}
and
\begin{align*}
\Var\left[\hat{\bar{y}}_g(1)\right] = \left(1 - \dfrac{n_{g,1}}{n_g}\right)\dfrac{S^2_{g,1}}{n_{g,1}},
\end{align*}
where $S^2_{g,1} \coloneqq (n_g - 1)^{-1}\sum_{i=1}^{n_g}[y_{g,i}(1) - \bar{y}_g(1)]^2$. By~(R1), each squared deviation satisfies $[y_{g,i}(1) - \bar{y}_g(1)]^2 \leq (2M)^2 = 4M^2$, so $S^2_{g,1} \leq 4M^2$. By~(R2), $n_{g,1} \to \infty$, so $\Var[\hat{\bar{y}}_g(1)] \to 0$. Chebyshev's inequality then implies
\begin{align}\label{eq: supp hat y1 cons}
\hat{\bar{y}}_g(1) \xrightarrow{p} \bar{y}_g(1).
\end{align}

Turning to $\hat{\bar{y}}_g(0)$ in \eqref{eq: supp est bar y0 g}, under Assumptions~1 and~4, a unit satisfies $s_{g,i}(0) = 1$ if and only if $r_{g,i} = \mathrm{AS}$, so
\begin{align*}
\hat{\bar{y}}_g(0) = \dfrac{1}{C_g(\bm{Z}_g)}\sum_{i:\, r_{g,i} = \mathrm{AS}} (1 - Z_{g,i})\, y_{g,i}(0),
\end{align*}
which is well defined on $\mathcal{E}_g \coloneqq \{C_g(\bm{Z}_g) \geq 1\}$.

By Lemma~\ref{lem: supp cond trt prob}, for any $c \geq 1$ and any $i$ with $r_{g,i} = \mathrm{AS}$, the conditional probability of assignment to the control condition is $\Pr(Z_{g,i} = 0 \given C_g(\bm{Z}_g) = c) = c / n_{g,\mathrm{AS}}$. Conditional on $\{C_g(\bm{Z}_g) = c\}$, the $c$ Always-Stop control units are therefore a simple random sample without replacement of size $c$ from the $n_{g,\mathrm{AS}}$ Always-Stop units. By the standard finite population formulas for sampling without replacement, this implies
\begin{align*}
\E\left[\hat{\bar{y}}_g(0) \given C_g(\bm{Z}_g) = c\right] = \bar{y}_g^{\mathrm{AS}}(0)
\end{align*}
and
\begin{align}\label{eq: supp Var control mean}
\Var\left[\hat{\bar{y}}_g(0) \given C_g(\bm{Z}_g) = c\right] = \left(1 - \dfrac{c}{n_{g,\mathrm{AS}}}\right)\dfrac{S^2_{g,0,\mathrm{AS}}}{c},
\end{align}
where $S^2_{g,0,\mathrm{AS}} \coloneqq (n_{g,\mathrm{AS}} - 1)^{-1}\sum_{i:\, r_{g,i} = \mathrm{AS}}[y_{g,i}(0) - \bar{y}_g^{\mathrm{AS}}(0)]^2$. Since the finite population correction factor satisfies $(1 - c/n_{g,\mathrm{AS}}) \leq 1$ for any $c \geq 1$, it follows that
\begin{align*}
\Var\left[\hat{\bar{y}}_g(0) \given C_g(\bm{Z}_g) = c\right] \leq \dfrac{S^2_{g,0,\mathrm{AS}}}{c}.
\end{align*}

Applying the law of total expectation conditional on $\mathcal{E}_g$ yields
\begin{align}\label{eq: supp EV est bar y0 g}
\E\left[\hat{\bar{y}}_g(0) \given \mathcal{E}_g\right] = \sum_{c \geq 1} \bar{y}_g^{\mathrm{AS}}(0)\,\Pr(C_g(\bm{Z}_g) = c \given \mathcal{E}_g) = \bar{y}_g^{\mathrm{AS}}(0),
\end{align}
and applying the law of total variance conditional on $\mathcal{E}_g$, and noting that the variance of the conditional expectation is $0$, yields
\begin{align}\label{eq: supp Var bound est bar y0 g}
\Var\left[\hat{\bar{y}}_g(0) \given \mathcal{E}_g\right]
&= \E\left[\Var\left(\hat{\bar{y}}_g(0) \given C_g(\bm{Z}_g), \mathcal{E}_g\right) \given \mathcal{E}_g\right]
+ \Var\left(\E\left[\hat{\bar{y}}_g(0) \given C_g(\bm{Z}_g), \mathcal{E}_g\right] \given \mathcal{E}_g\right) \nonumber \\
&= \E\left[\Var\left(\hat{\bar{y}}_g(0) \given C_g(\bm{Z}_g), \mathcal{E}_g\right) \given \mathcal{E}_g\right] \nonumber \\
&\leq \E\left[\dfrac{S^2_{g,0,\mathrm{AS}}}{C_g(\bm{Z}_g)} \given \mathcal{E}_g\right]
= S^2_{g,0,\mathrm{AS}}\,\E\left[\dfrac{1}{C_g(\bm{Z}_g)} \given \mathcal{E}_g\right].
\end{align}

By Lemma~\ref{lem: supp EV inv Cg}, $\E[1/C_g(\bm{Z}_g) \given \mathcal{E}_g] \to 0$. Since $S^2_{g,0,\mathrm{AS}} \leq 4M^2$ by (R1),
\begin{align*}
\Var\left[\hat{\bar{y}}_g(0) \given \mathcal{E}_g\right] \leq S^2_{g,0,\mathrm{AS}}\,\E\left[\dfrac{1}{C_g(\bm{Z}_g)} \given \mathcal{E}_g\right] \to 0.
\end{align*}
To complete the proof, apply Chebyshev's inequality to $\hat{\bar{y}}_g(0)$ conditional on $\mathcal{E}_g$:
\begin{align*}
\Pr\left(\left|\hat{\bar{y}}_g(0)-\E[\hat{\bar{y}}_g(0)\given\mathcal{E}_g]\right|>\vartheta \given \mathcal{E}_g\right) \leq \dfrac{\Var[\hat{\bar{y}}_g(0)\given\mathcal{E}_g]}{\vartheta^2}.
\end{align*}
Using \eqref{eq: supp EV est bar y0 g} to substitute $\E[\hat{\bar{y}}_g(0)\given\mathcal{E}_g]=\bar{y}_g^{\mathrm{AS}}(0)$ and the fact that $\Var[\hat{\bar{y}}_g(0)\given\mathcal{E}_g]\to0$ yields
\begin{align*}
\Pr\left(\left|\hat{\bar{y}}_g(0)-\bar{y}_g^{\mathrm{AS}}(0)\right|>\vartheta \given \mathcal{E}_g\right)
\to 0.
\end{align*}
Hence,
\begin{align}\label{eq: supp hat y0 cons}
\hat{\bar{y}}_g(0) \xrightarrow{p} \bar{y}_g^{\mathrm{AS}}(0)
\quad\text{under } \Pr(\cdot \given \mathcal{E}_g).
\end{align}

To show that \eqref{eq: supp hat y1 cons} also holds under $\Pr(\cdot \given \mathcal{E}_g)$, fix any $\varepsilon > 0$. By the definition of conditional probability,
\begin{align*}
\Pr\left(\abs{\hat{\bar{y}}_g(1) - \bar{y}_g(1)} > \varepsilon \given \mathcal{E}_g\right)
&= \dfrac{\Pr\left(\abs{\hat{\bar{y}}_g(1) - \bar{y}_g(1)} > \varepsilon,\; \mathcal{E}_g\right)}{\Pr(\mathcal{E}_g)}.
\end{align*}
Since the event
\begin{align*}
\left\{\abs{\hat{\bar{y}}_g(1) - \bar{y}_g(1)} > \varepsilon,\; \mathcal{E}_g\right\}
\end{align*}
is contained in
\begin{align*}
\left\{\abs{\hat{\bar{y}}_g(1) - \bar{y}_g(1)} > \varepsilon\right\},
\end{align*}
monotonicity of probability implies
\begin{align*}
\Pr\left(\abs{\hat{\bar{y}}_g(1) - \bar{y}_g(1)} > \varepsilon \given \mathcal{E}_g\right) \leq \dfrac{\Pr\left(\abs{\hat{\bar{y}}_g(1) - \bar{y}_g(1)} > \varepsilon\right)}{\Pr(\mathcal{E}_g)}.
\end{align*}
By \eqref{eq: supp hat y1 cons}, the numerator converges to $0$, and by \Cref{lem: supp prob zero denom}, $\Pr(\mathcal{E}_g) \to 1$, so
\begin{align}\label{eq: supp hat y1 cons cond}
\hat{\bar{y}}_g(1) \xrightarrow{p} \bar{y}_g(1) \quad\text{under } \Pr(\cdot \given \mathcal{E}_g).
\end{align}

Combining \eqref{eq: supp hat y1 cons cond} and \eqref{eq: supp hat y0 cons}, Proposition~\ref{prop: supp decomp}, and the continuous mapping theorem yields
\begin{align*}
\hat{\tau}_g = \hat{\bar{y}}_g(1) - (1 - \rho_g)\,\hat{\bar{y}}_g(0) \xrightarrow{p}
\tau_g \quad\text{under } \Pr(\cdot \given \mathcal{E}_g).
\end{align*}

Finally, since $\mathcal{E}_g^{\complement} = \{C_g(\bm{Z}_g)=0\}$ and $\Pr(\mathcal{E}_g^{\complement}) \to 0$ by \Cref{lem: supp prob zero denom}, for any $\varepsilon > 0$, partitioning the event $\{\abs{\hat{\tau}_g - \tau_g} > \varepsilon\}$ according to $\mathcal{E}_g$ and $\mathcal{E}_g^{\complement}$, applying the definition of conditional probability, and using monotonicity of probability yields
\begin{align*}
\Pr\left(\abs{\hat{\tau}_g - \tau_g} > \varepsilon\right)
&= \Pr\left(\abs{\hat{\tau}_g - \tau_g} > \varepsilon,\; \mathcal{E}_g\right)
  + \Pr\left(\abs{\hat{\tau}_g - \tau_g} > \varepsilon,\; \mathcal{E}_g^{\complement}\right) \\
&\leq \Pr\left(\abs{\hat{\tau}_g - \tau_g} > \varepsilon \given \mathcal{E}_g\right)\Pr(\mathcal{E}_g)
  + \Pr(\mathcal{E}_g^{\complement}) \\
&\leq \Pr\left(\abs{\hat{\tau}_g - \tau_g} > \varepsilon \given \mathcal{E}_g\right)
  + \Pr(\mathcal{E}_g^{\complement}) \to 0,
\end{align*}
which, by the definition of convergence in probability, establishes $\hat{\tau}_g \xrightarrow{p} \tau_g$.
\end{proof}

\begin{rmk}\label{rmk: supp Gstar equals G}
For strata outside $\mathcal{G}^{\ast}$ --- those containing encounters of only one race --- No-Bias-in-Encounters is trivially satisfied, since all encounters share a degenerate conditional probability of $0$ or $1$; the assumption is substantive only in strata where both races are present. Under regularity condition \textup{(R2)}, $n_{g,1}/n_g \to v_g \in (0,1)$ implies $n_{g,1} \to \infty$ and $n_{g,0} \to \infty$. Consequently, for every $g \in \mathcal{G}$, there exists a stage of the asymptotic sequence beyond which $n_{g,1} \geq 1$ and $n_{g,0} \geq 1$, so that $g \in \mathcal{G}^{\ast}$. Since the number of strata $\abs{\mathcal{G}}$ is held fixed, $\mathcal{G}^{\ast} = \mathcal{G}$ for all sufficiently large elements of the asymptotic sequence. The aggregate estimator, defined over $\mathcal{G}^{\ast}$, therefore asymptotically targets the average causal effect across all strata in $\mathcal{G}$.
\end{rmk}

\begin{rmk}\label{rmk: supp unidentified weights}
The stratum sizes $n_g$ and the population size $n^{\ast} = \sum_{g \in \mathcal{G}^{\ast}} n_g$ depend on the unobserved number of Only-Minority-Stop encounters in each stratum, so the weights $n_g / n^{\ast}$ that define the aggregate estimand $\tau$ in equation~(6) of the manuscript are not identified from the observed data alone. In the sequential sensitivity framework, specifying $\bm{\underline{\rho}}$ determines the augmented stratum sizes $\tilde{n}_g^{\underline{\rho}_g}$ and hence the augmented weights.
\end{rmk}

\begin{rmk}\label{rmk: supp alternative regime}
Even if the population weights were known, the consistency of the aggregate estimator $\sum_{g \in \mathcal{G}^{\ast}} (n_g / n^{\ast})\,\hat{\tau}_g$ for $\tau$ would not necessarily hold in an alternative asymptotic regime in which $\abs{\mathcal{G}^{\ast}}$ grows while stratum sizes $n_g$ remain uniformly bounded. In the standard setting without sample selection, where all potential outcomes are observed, the within-stratum Difference-in-Means $\hat{\bar{y}}_g(1) - \hat{\bar{y}}_g(0)$ is unbiased for $\tau_g$ in every stratum under complete randomization regardless of stratum size. The aggregate estimator is therefore unbiased for $\tau$ at every stage, and as $\abs{\mathcal{G}^{\ast}}$ grows the variance shrinks to zero, yielding consistency.

The estimator $\hat{\tau}_g = \hat{\bar{y}}_g(1) - (1 - \rho_g)\,\hat{\bar{y}}_g(0)$ does not share this property: As Proposition~\ref{prop: supp bias} shows, $\hat{\tau}_g$ is generally biased for $\tau_g$ in finite samples. In the regime of Proposition~\ref{prop: supp consistency}, this bias vanishes as $n_g \to \infty$. When stratum sizes remain bounded, however, the bias term $(n_{g,\mathrm{AS}} - C_g(\bm{Z}_g))/n_{g,1} - (1 - \rho_g)$ need not vanish on average across strata. Consequently, even as $\abs{\mathcal{G}^{\ast}}$ grows and the variance of the aggregate estimator shrinks to zero, the aggregate estimator converges in probability to a quantity that need not equal $\tau$.
\end{rmk}

\section{Formal Results and Proofs for Section 5}\label{sec: supp formal section 5}

\subsection{Proof of Proposition~1}\label{sec: supp proof prop 1}

Proposition 1 of the manuscript establishes a one-to-one correspondence between a researcher's posited value of $\underline{\rho}_g$ in the feasible domain $\mathcal{F}_{\underline{\rho}_g}$ and the implied number of missing Only-Minority-Stop white-civilian encounters $\tilde{n}_{g,0,\mathrm{OMS}}^{\underline{\rho}_g} \in \mathbb{Z}_{\geq 0}$, where $\mathbb{Z}_{\geq 0} \coloneqq \{0, 1, 2, \ldots\}$ denotes the set of nonnegative integers. The result applies to any stratum $g \in \mathcal{G}$ with at least one observed encounter ($n_{g,1} + n_{g,0,\mathrm{AS}} \geq 1$), which is more general than $g \in \mathcal{G}^{\ast}$ since $\mathcal{G}^{\ast}$ requires both $n_{g,1} \geq 1$ and $n_{g,0} \geq 1$.

We restate the proposition for reference.
\begin{prop}[Restatement of Proposition 1 of the manuscript]
\label{prop: bias-in-stops missing OMS control restate}
Under Assumptions 1 -- 4, fix a stratum $g \in \mathcal{G}^{\ast}$ with observed minority-civilian count $n_{g,1}$ and observed Always-Stop control count $n_{g,0,\mathrm{AS}}$. For $w \in \mathbb{Z}_{\geq 0}$ posited missing Only-Minority-Stop control units, the lower bound on racial discrimination in stops is
\begin{align}
\underline{\rho}_g & = \dfrac{w}{n_{g,1} + n_{g,0,\mathrm{AS}} + w},
\end{align}
with inverse
\begin{align}
w & = \dfrac{\underline{\rho}_g}{1 - \underline{\rho}_g} \left(n_{g,1} + n_{g,0,\mathrm{AS}}\right).
\end{align}
Moreover, the map
\begin{align} \label{eq: map OMS control to LB rho}
w \mapsto \dfrac{w}{n_{g,1} + n_{g,0,\mathrm{AS}} + w}
\end{align}
is a bijection from $\mathbb{Z}_{\geq 0}$ to the feasible domain of $\underline{\rho}_g$
\begin{align}
\mathcal{F}_{\underline{\rho}_g} & \coloneqq \left\{\dfrac{w}{n_{g,1} + n_{g,0,\mathrm{AS}} + w}: w \in \mathbb{Z}_{\geq 0} \right\} \subset [0,1).
\end{align}
\end{prop}

\begin{proof}
Fix a stratum $g \in \mathcal{G}_{\geq 1} \coloneqq \{g \in \mathcal{G}: n_{g,1} + n_{g,0,\mathrm{AS}} \geq 1\}$ and let $\ell \coloneqq n_{g,1} + n_{g,0,\mathrm{AS}} \geq 1$.

Under Assumptions~1 -- 4, for any posited number $w \in \mathbb{Z}_{\geq 0}$ of missing Only-Minority-Stop control units, the number of Only-Minority-Stop units in stratum $g$ is at least $w$ (from the $w$ posited control units) and at most $w + n_{g,1}$ (if every treated unit is also Only-Minority-Stop). The proportion of Only-Minority-Stop units among all $\ell + w$ potentially stoppable units is therefore
\begin{align*}
\rho_g = \dfrac{w + n_{g,1,\mathrm{OMS}}}{\ell + w},
\end{align*}
where $n_{g,1,\mathrm{OMS}} \in \{0, \ldots, n_{g,1}\}$ is unknown. Since $\rho_g$ is increasing in $n_{g,1,\mathrm{OMS}}$ for fixed $w$ and $\ell$, the minimum is attained at $n_{g,1,\mathrm{OMS}} = 0$, yielding
\begin{align*}
\underline{\rho}_g = \dfrac{w}{\ell + w}.
\end{align*}
Solving for $w$ by cross-multiplying $\underline{\rho}_g(\ell + w) = w$ and rearranging gives
\begin{align*}
w = \dfrac{\underline{\rho}_g}{1 - \underline{\rho}_g}\,\ell.
\end{align*}

Consider the map $w \mapsto w/(\ell + w)$ from $\mathbb{Z}_{\geq 0}$ to $[0, 1)$. Since $\ell \geq 1 > 0$, this map is well defined. We show it is a bijection from $\mathbb{Z}_{\geq 0}$ onto the feasible domain $\mathcal{F}_{\underline{\rho}_g}$ by establishing that the map is injective and that its image coincides with $\mathcal{F}_{\underline{\rho}_g}$.

\medskip
\noindent\textit{Injectivity.} Suppose $w_1/(\ell + w_1) = w_2/(\ell + w_2)$ for $w_1, w_2 \in \mathbb{Z}_{\geq 0}$. Cross-multiplying yields $w_1(\ell + w_2) = w_2(\ell + w_1)$, which simplifies to $w_1 \ell = w_2 \ell$. Since $\ell > 0$, $w_1 = w_2$.

\medskip
\noindent\textit{Image equals $\mathcal{F}_{\underline{\rho}_g}$.} The feasible domain is $\mathcal{F}_{\underline{\rho}_g} \coloneqq \{w / (\ell + w): w \in \mathbb{Z}_{\geq 0}\} \subset [0, 1)$. By construction, the image of the map --- i.e., $\{w/(\ell + w): w \in \mathbb{Z}_{\geq 0}\}$ --- is identical to $\mathcal{F}_{\underline{\rho}_g}$: Every element of $\mathcal{F}_{\underline{\rho}_g}$ is $w/(\ell + w)$ for some $w \in \mathbb{Z}_{\geq 0}$, and every $w/(\ell + w)$ belongs to $\mathcal{F}_{\underline{\rho}_g}$.

\medskip
Since the map is injective and its image equals $\mathcal{F}_{\underline{\rho}_g}$, it is a bijection from $\mathbb{Z}_{\geq 0}$ to $\mathcal{F}_{\underline{\rho}_g}$.
\end{proof}

\subsection{\texorpdfstring{Worked Example: Interaction of $\bm{\underline{\rho}}$ and $\Gamma$ in a Single Stratum}{Worked Example: Interaction of rho and Gamma in a Single Stratum}}\label{sec: supp tilt one minority worked example}

The interaction of $\bm{\underline{\rho}}$ and $\Gamma$ can be seen directly from the tilted statistic. To illustrate, consider a single stratum with one minority-civilian encounter ($n_{g,1} = 1$), an upper-tailed test ($d = +1$) with $\tau_0 = 0$, and suppose $\hat{\tau}_g^{\underline{\rho}_g} - \tau_0 \geq 0$. The tilted statistic simplifies to $\hat{\tau}_g^{\mathrm{tilt}}(\underline{\rho}_g; \Gamma, \tau_0, d = +1) = (\hat{\tau}_g^{\underline{\rho}_g} - \tau_0)(\tilde{n}_g^{\underline{\rho}_g} - 1 + \Gamma)/(\Gamma\, \tilde{n}_g^{\underline{\rho}_g})$. The first factor, $\hat{\tau}_g^{\underline{\rho}_g} - \tau_0$, increases with $\underline{\rho}_g$ because adding zeros to the white-civilian outcomes lowers the white-civilian mean and therefore raises the Difference-in-Means. The second factor applies the probability bound from Lemma~1 of the manuscript, tilting the centered Difference-in-Means toward zero to ensure valid inference under any assignment mechanism consistent with $\Gamma$ over the augmented assignment space. Both parameters appear in the same multiplicative expression: $\underline{\rho}_g$ enters through both $\hat{\tau}_g^{\underline{\rho}_g}$ and $\tilde{n}_g^{\underline{\rho}_g}$, while $\Gamma$ enters directly. The tilted statistic cannot be decomposed into separate, additive components for each.

Because the two parameters appear in the same multiplicative expression, the effect of changing one depends on the value of the other. Consider increasing $\underline{\rho}_g$ while holding $\Gamma$ fixed. The first factor grows --- appending more zeros to the white-civilian outcomes pulls the white-civilian mean toward zero, which raises the Difference-in-Means $\hat{\tau}_g^{\underline{\rho}_g}$ and therefore increases the centered Difference-in-Means $\hat{\tau}_g^{\underline{\rho}_g} - \tau_0$. The second factor shrinks toward $1/\Gamma$ as the augmented stratum size $\tilde{n}_g^{\underline{\rho}_g}$ grows. When $\Gamma$ is small, the second factor shrinks slowly and the first factor dominates, so the tilted statistic rises with $\underline{\rho}_g$. When $\Gamma$ is large, the second factor shrinks rapidly and can eventually outweigh the first factor, so the tilted statistic rises and then falls. A symmetric pattern holds if $\Gamma$ varies with $\underline{\rho}_g$ fixed: The change in $\Gamma$ is muted when $\underline{\rho}_g$ is small but amplified once the augmented stratum has grown. Isolating one parameter and holding the other at a default value therefore misses these joint effects.

\section{Proof of Lemma 1} \label{sec: supp proof lem 1}

The proof of Lemma 1 in the manuscript rests on two preliminary results. The first (\Cref{lem: submodel cond dist}) derives the conditional distribution of assignments within a stratum under the parametric submodel \eqref{eq: submodel} below. The second (\Cref{lem: submodel equiv}) shows that this submodel generates the same class of conditional distributions as the general $\Gamma$-bound in equation (1) of the manuscript and attains the bound sharply. Consequently, optimizing over the submodel suffices to bound probabilities over the full class.

\begin{lem} \label{lem: submodel cond dist}
Fix a stratum $g \in \mathcal{G}^{\ast}$ with augmented stratum size
$\tilde{n}_g^{\underline{\rho}_g}$ and augmented assignment space
$\Omega_g^{\underline{\rho}_g}$. Suppose
$Z_{g,1}, \ldots, Z_{g,\tilde{n}_g^{\underline{\rho}_g}}$ are mutually
independent Bernoulli random variables with
\begin{equation} \label{eq: submodel}
\log\left\{\dfrac{\pi_{g,i}}{1 - \pi_{g,i}}\right\} = \kappa_g + \log(\Gamma)\, u_{g,i}
\end{equation}
for some $\kappa_g \in \R$, $\bm{u}_g \coloneqq \left(u_{g,1}, \ldots,
u_{g,\tilde{n}_g^{\underline{\rho}_g}}\right)^{\top} \in
[0,1]^{\tilde{n}_g^{\underline{\rho}_g}}$, and $\Gamma \geq 1$. Then the
conditional distribution of $\bm{Z}_g$ given the event
$\bm{Z}_g \in \Omega_g^{\underline{\rho}_g}$ is
\begin{equation} \label{eq: submodel cond dist}
\Pr\left(\bm{Z}_g = \bm{z}_g^{\underline{\rho}_g} \given \bm{Z}_g \in
\Omega_g^{\underline{\rho}_g}\right) =
\dfrac{\Gamma^{\bm{z}_g^{\underline{\rho}_g \top}\bm{u}_g}}{\sum_{\bm{a}_g
\in \Omega_g^{\underline{\rho}_g}} \Gamma^{\bm{a}_g^{\top}\bm{u}_g}}
\end{equation}
for every $\bm{z}_g^{\underline{\rho}_g} \in \Omega_g^{\underline{\rho}_g}$.
\end{lem}

\begin{proof}
From \eqref{eq: submodel}, the odds that unit $i$ in stratum $g$ is
treated are
$\pi_{g,i}/(1 - \pi_{g,i}) = \exp\left(\kappa_g + \log(\Gamma)\,
u_{g,i}\right)$, so
\begin{equation*}
\pi_{g,i} = \dfrac{\exp\left(\kappa_g + \log(\Gamma)\,
u_{g,i}\right)}{1 + \exp\left(\kappa_g + \log(\Gamma)\, u_{g,i}\right)}.
\end{equation*}
For any $\bm{z}_g \in \{0,1\}^{\tilde{n}_g^{\underline{\rho}_g}}$,
independence of the $Z_{g,i}$ across $i = 1, \ldots,
\tilde{n}_g^{\underline{\rho}_g}$ within stratum $g$ gives
\begin{align}
\Pr\left(\bm{Z}_g = \bm{z}_g\right)
& = \prod_{i=1}^{\tilde{n}_g^{\underline{\rho}_g}}
\pi_{g,i}^{z_{g,i}}\left(1 - \pi_{g,i}\right)^{1 - z_{g,i}} \notag \\
& = \prod_{i=1}^{\tilde{n}_g^{\underline{\rho}_g}}
\dfrac{\exp\left(\kappa_g + \log(\Gamma)\, u_{g,i}\right)^{z_{g,i}}}{1 +
\exp\left(\kappa_g + \log(\Gamma)\, u_{g,i}\right)}. \label{eq: joint prob
product}
\end{align}
We now separate each factor in the numerator of \eqref{eq: joint prob
product}. Writing $\exp\left(\kappa_g + \log(\Gamma)\,
u_{g,i}\right)^{z_{g,i}} = \exp\left(\kappa_g\, z_{g,i}\right)
\Gamma^{u_{g,i}\, z_{g,i}}$ and taking the product over
$i = 1, \ldots, \tilde{n}_g^{\underline{\rho}_g}$ yields
\begin{equation} \label{eq: joint prob separated}
\Pr\left(\bm{Z}_g = \bm{z}_g\right) =
\left\{\prod_{i = 1}^{\tilde{n}_g^{\underline{\rho}_g}} \dfrac{1}{1 +
\exp\left(\kappa_g + \log(\Gamma)\, u_{g,i}\right)}\right\}
\exp\left(\kappa_g \sum_{i=1}^{\tilde{n}_g^{\underline{\rho}_g}}
z_{g,i}\right) \Gamma^{\bm{z}_g^{\top}\bm{u}_g}.
\end{equation}
By the definition of conditional probability,
\begin{equation} \label{eq: cond prob ratio}
\Pr\left(\bm{Z}_g = \bm{z}_g^{\underline{\rho}_g} \given \bm{Z}_g \in
\Omega_g^{\underline{\rho}_g}\right) =
\dfrac{\Pr\left(\bm{Z}_g =
\bm{z}_g^{\underline{\rho}_g}\right)}{\sum_{\bm{a}_g \in
\Omega_g^{\underline{\rho}_g}} \Pr\left(\bm{Z}_g = \bm{a}_g\right)}.
\end{equation}
Two factors in \eqref{eq: joint prob separated} are common to every
$\bm{a}_g \in \Omega_g^{\underline{\rho}_g}$. The first is the product
\begin{align*}
\prod_{i = 1}^{\tilde{n}_g^{\underline{\rho}_g}} \dfrac{1}{1 +
\exp\left(\kappa_g + \log(\Gamma)\, u_{g,i}\right)},
\end{align*}
which depends on $\kappa_g$ and $\bm{u}_g$ but not on $\bm{z}_g$. The
second is the intercept term: Every
$\bm{a}_g \in \Omega_g^{\underline{\rho}_g}$ satisfies
$\sum_{i = 1}^{\tilde{n}_g^{\underline{\rho}_g}} a_{g,i} = n_{g,1}$, so
$\exp\left(\kappa_g \sum_i a_{g,i}\right) = \exp\left(\kappa_g\,
n_{g,1}\right)$ takes the same value for every element of
$\Omega_g^{\underline{\rho}_g}$. Because both factors appear identically in
every term of the numerator and denominator of \eqref{eq: cond prob ratio},
they cancel, leaving
\begin{equation*}
\Pr\left(\bm{Z}_g = \bm{z}_g^{\underline{\rho}_g} \given \bm{Z}_g \in
\Omega_g^{\underline{\rho}_g}\right) =
\dfrac{\Gamma^{\bm{z}_g^{\underline{\rho}_g \top}\bm{u}_g}}{\sum_{\bm{a}_g
\in \Omega_g^{\underline{\rho}_g}} \Gamma^{\bm{a}_g^{\top}\bm{u}_g}},
\end{equation*}
which is \eqref{eq: submodel cond dist}.
\end{proof}

\begin{lem} \label{lem: submodel equiv}
The set of conditional distributions on $\Omega_g^{\underline{\rho}_g}$
induced by the parametric submodel \eqref{eq: submodel} as $\bm{u}_g$
ranges over $[0,1]^{\tilde{n}_g^{\underline{\rho}_g}}$ and $\kappa_g$
ranges over $\R$ equals the set of conditional distributions on
$\Omega_g^{\underline{\rho}_g}$ induced by mutually independent Bernoulli
assignments $\{\pi_{g,i}\}_{i=1}^{\tilde{n}_g^{\underline{\rho}_g}}$
satisfying the bound in equation~(1) of the manuscript.
\end{lem}

\begin{proof}
We establish both inclusions.

\emph{Submodel $\subseteq$ $\Gamma$-class.} Under \eqref{eq: submodel}, the
odds $\pi_{g,i}/(1-\pi_{g,i}) = \exp(\kappa_g) \Gamma^{u_{g,i}}$, so for any
$i, j \in \{1, \ldots, \tilde{n}_g^{\underline{\rho}_g}\}$,
\begin{equation*}
\dfrac{\pi_{g,i}/(1-\pi_{g,i})}{\pi_{g,j}/(1-\pi_{g,j})} = \Gamma^{u_{g,i} -
u_{g,j}}.
\end{equation*}
Because $u_{g,i}, u_{g,j} \in [0,1]$, the exponent $u_{g,i} - u_{g,j} \in
[-1, 1]$, so the ratio lies in $[1/\Gamma, \Gamma]$. The Bernoulli
assignments $\{\pi_{g,i}\}$ therefore satisfy the bound in equation~(1) of
the manuscript.

\emph{$\Gamma$-class $\subseteq$ Submodel.} Let
$(\pi_{g,1}, \ldots, \pi_{g,\tilde{n}_g^{\underline{\rho}_g}})$ be any
tuple of Bernoulli probabilities satisfying the bound in equation~(1) of
the manuscript. We construct $(\kappa_g, \bm{u}_g)$ with $\bm{u}_g \in
[0,1]^{\tilde{n}_g^{\underline{\rho}_g}}$ such that \eqref{eq: submodel}
recovers every $\pi_{g,i}$ in this tuple, thereby showing that the
conditional distribution on $\Omega_g^{\underline{\rho}_g}$ induced by the
given tuple is also induced by the submodel. If $\Gamma = 1$, the bound
forces all $\pi_{g,i}$ to be equal; set $\kappa_g \coloneqq
\log\left\{\pi_{g,1}/(1-\pi_{g,1})\right\}$ and $u_{g,i} \coloneqq 0$ for
all $i$. For $\Gamma > 1$, let $\pi_{g,\min} \coloneqq \min_{i}
\pi_{g,i}$, set $\kappa_g \coloneqq
\log\left\{\pi_{g,\min}/(1-\pi_{g,\min})\right\}$, and define
\begin{equation*}
u_{g,i} \coloneqq
\dfrac{\log\left\{\pi_{g,i}/(1-\pi_{g,i})\right\} -
\log\left\{\pi_{g,\min}/(1-\pi_{g,\min})\right\}}{\log(\Gamma)}.
\end{equation*}
Since $\log(\Gamma) > 0$ and $\pi_{g,i} \geq \pi_{g,\min}$, the numerator
is non-negative, so $u_{g,i} \geq 0$. The upper bound in equation~(1) of
the manuscript, applied with $j$ equal to the index attaining
$\pi_{g,\min}$, gives
\begin{equation*}
\log\left\{\pi_{g,i}/(1-\pi_{g,i})\right\} -
\log\left\{\pi_{g,\min}/(1-\pi_{g,\min})\right\} \leq \log(\Gamma),
\end{equation*}
so $u_{g,i} \leq 1$. Hence $\bm{u}_g \in
[0,1]^{\tilde{n}_g^{\underline{\rho}_g}}$. By construction, $\kappa_g +
\log(\Gamma)\, u_{g,i} = \log\left\{\pi_{g,i}/(1-\pi_{g,i})\right\}$, so
\eqref{eq: submodel} holds for each $i = 1, \ldots,
\tilde{n}_g^{\underline{\rho}_g}$. Because both parametrizations specify
the same marginal probabilities $\{\pi_{g,i}\}$, they induce the same
joint distribution on $\{0,1\}^{\tilde{n}_g^{\underline{\rho}_g}}$ and
hence the same conditional distribution on
$\Omega_g^{\underline{\rho}_g}$. Since the starting tuple was arbitrary,
every conditional distribution in the $\Gamma$-class belongs to the
submodel class.
\end{proof}

\begin{rmk} \label{rem: sharp attainment}
The submodel attains the $\Gamma$-bound sharply: Setting $u_{g,i} = 1$ and
$u_{g,j} = 0$ yields an odds ratio of exactly $\Gamma$, so the class of
conditional distributions generated by the submodel includes mechanisms at
the boundary of the $\Gamma$-class.
\end{rmk}

We now restate and prove Lemma 1 of the manuscript, drawing on Lemmas \ref{lem: submodel cond dist} and \ref{lem: submodel equiv}.

\begin{lem} \label{lem: prob bounds restate}
Under the restriction on the assignment model in equation~(1) of the manuscript, the lower ($\underline{p}$) and upper ($\overline{p}$) bounds on the conditional probability of any $\bm{z}_g^{\underline{\rho}_g} \in \Omega_g^{\underline{\rho}_g}$ for $\Gamma \geq 1$ are
\begin{align} \label{eq: lb cond z prob}
\underline{p}\left(\bm{z}_g^{\underline{\rho}_g}; \Gamma\right) & =
\dfrac{1}{\sum_{\bm{a}_g \in \Omega_g^{\underline{\rho}_g}}
\Gamma^{\bm{a}_g^{\top}\left(\bm{1}-\bm{z}_g^{\underline{\rho}_g}\right)}},
\\[0.75em]
\overline{p}\left(\bm{z}_g^{\underline{\rho}_g}; \Gamma\right) & =
\dfrac{\Gamma^{n_{g,1}}}{\sum_{\bm{a}_g \in \Omega_g^{\underline{\rho}_g}}
\Gamma^{\bm{a}_g^{\top}\bm{z}_g^{\underline{\rho}_g}}}. \label{eq: ub cond
z prob}
\end{align}
\end{lem}

\begin{proof}
By \Cref{lem: submodel equiv}, every conditional distribution on
$\Omega_g^{\underline{\rho}_g}$ that satisfies equation~(1) of the
manuscript can be generated by the parametric submodel \eqref{eq: submodel}
for some $\bm{u}_g \in [0,1]^{\tilde{n}_g^{\underline{\rho}_g}}$.
Bounding the conditional probability over the submodel therefore bounds it
over the entire class. By \Cref{lem: submodel cond dist}, the conditional
probability under the submodel takes the form
\begin{equation} \label{eq: cond prob u form}
\Pr\left(\bm{Z}_g = \bm{z}_g^{\underline{\rho}_g} \given \bm{Z}_g \in
\Omega_g^{\underline{\rho}_g}\right) =
\dfrac{\Gamma^{\bm{z}_g^{\underline{\rho}_g \top}\bm{u}_g}}{\sum_{\bm{a}_g
\in \Omega_g^{\underline{\rho}_g}} \Gamma^{\bm{a}_g^{\top}\bm{u}_g}}.
\end{equation}
It remains to optimize \eqref{eq: cond prob u form} over
$\bm{u}_g \in [0,1]^{\tilde{n}_g^{\underline{\rho}_g}}$.

\emph{Lower bound.} Dividing numerator and denominator of
\eqref{eq: cond prob u form} by
$\Gamma^{\bm{z}_g^{\underline{\rho}_g \top}\bm{u}_g}$ gives
\begin{equation} \label{eq: cond prob rewritten}
\Pr\left(\bm{Z}_g = \bm{z}_g^{\underline{\rho}_g} \given \bm{Z}_g \in
\Omega_g^{\underline{\rho}_g}\right) = \dfrac{1}{\sum_{\bm{a}_g \in
\Omega_g^{\underline{\rho}_g}} \Gamma^{\left(\bm{a}_g -
\bm{z}_g^{\underline{\rho}_g}\right)^{\top}\bm{u}_g}}.
\end{equation}
Minimizing the left side over $\bm{u}_g$ is equivalent to maximizing the
denominator on the right side. Writing the denominator as a product over
$i = 1, \ldots, \tilde{n}_g^{\underline{\rho}_g}$ within stratum $g$,
\begin{equation*}
\sum_{\bm{a}_g \in \Omega_g^{\underline{\rho}_g}}
\Gamma^{\left(\bm{a}_g -
\bm{z}_g^{\underline{\rho}_g}\right)^{\top}\bm{u}_g}
= \sum_{\bm{a}_g \in \Omega_g^{\underline{\rho}_g}}
\prod_{i=1}^{\tilde{n}_g^{\underline{\rho}_g}}
\Gamma^{\left(a_{g,i} - z_{g,i}^{\underline{\rho}_g}\right) u_{g,i}}.
\end{equation*}
For each $\bm{a}_g \in \Omega_g^{\underline{\rho}_g}$ and each
$i = 1, \ldots, \tilde{n}_g^{\underline{\rho}_g}$ within stratum $g$, the
coefficient $a_{g,i} - z_{g,i}^{\underline{\rho}_g} \in \{-1, 0, 1\}$, and
its sign is determined by $z_{g,i}^{\underline{\rho}_g}$ alone:
\begin{itemize}
  \item If $z_{g,i}^{\underline{\rho}_g} = 0$, then $a_{g,i} -
  z_{g,i}^{\underline{\rho}_g} = a_{g,i} \in \{0,1\}$, a non-negative
  coefficient.
  \item If $z_{g,i}^{\underline{\rho}_g} = 1$, then $a_{g,i} -
  z_{g,i}^{\underline{\rho}_g} = a_{g,i} - 1 \in \{-1, 0\}$, a
  non-positive coefficient.
\end{itemize}
Since $\Gamma \geq 1$, the choice $\bm{u}_g = \bm{1} - \bm{z}_g^{\underline{\rho}_g}$ maximizes every factor $\Gamma^{(a_{g,i} - z_{g,i}^{\underline{\rho}_g})u_{g,i}}$ simultaneously. Under this choice, $u_{g,i} = 1$ whenever $z_{g,i}^{\underline{\rho}_g} = 0$ and $u_{g,i} = 0$ whenever $z_{g,i}^{\underline{\rho}_g} = 1$, for all $\bm{a}_g$ and all $i = 1,\ldots,\tilde{n}_g^{\underline{\rho}_g}$ in stratum $g$. A common coordinate-wise maximizer of each summand's factored form is a maximizer of the sum, so $\bm{u}_g = \bm{1} - \bm{z}_g^{\underline{\rho}_g}$ maximizes the denominator of \eqref{eq: cond prob rewritten}.

Substituting $\bm{u}_g = \bm{1} - \bm{z}_g^{\underline{\rho}_g}$ into
\eqref{eq: cond prob u form}, the numerator becomes
\begin{equation*}
\Gamma^{\bm{z}_g^{\underline{\rho}_g \top}\left(\bm{1} -
\bm{z}_g^{\underline{\rho}_g}\right)} = \Gamma^{\sum_{i}
z_{g,i}^{\underline{\rho}_g}\left(1 - z_{g,i}^{\underline{\rho}_g}\right)}
= \Gamma^0 = 1,
\end{equation*}
where the middle equality uses $z_{g,i}^{\underline{\rho}_g}\left(1 -
z_{g,i}^{\underline{\rho}_g}\right) = 0$ because
$z_{g,i}^{\underline{\rho}_g} \in \{0,1\}$. The denominator becomes
$\sum_{\bm{a}_g \in \Omega_g^{\underline{\rho}_g}}
\Gamma^{\bm{a}_g^{\top}\left(\bm{1} -
\bm{z}_g^{\underline{\rho}_g}\right)}$, yielding \eqref{eq: lb cond z prob}.

\emph{Upper bound.} By the parallel argument, maximizing
\eqref{eq: cond prob u form} over $\bm{u}_g$ is equivalent to minimizing
the denominator of \eqref{eq: cond prob rewritten}. The choice $\bm{u}_g = \bm{z}_g^{\underline{\rho}_g}$ minimizes every factor $\Gamma^{(a_{g,i} - z_{g,i}^{\underline{\rho}_g})u_{g,i}}$ simultaneously. Under this choice, $u_{g,i} = 1$ whenever $z_{g,i}^{\underline{\rho}_g} = 1$ and $u_{g,i} = 0$ whenever $z_{g,i}^{\underline{\rho}_g} = 0$, for all $\bm{a}_g$ and all $i = 1,\ldots,\tilde{n}_g^{\underline{\rho}_g}$ in stratum $g$. Indices with coefficient $a_{g,i}-1 \leq 0$ therefore receive $u_{g,i} = 1$, which makes the exponent as negative as possible. Indices with coefficient $a_{g,i} \geq 0$ receive $u_{g,i} = 0$, so their exponent contribution is zero. A common coordinate-wise minimizer of each summand's factored form is a minimizer of the sum, so $\bm{u}_g = \bm{z}_g^{\underline{\rho}_g}$ minimizes the denominator of \eqref{eq: cond prob rewritten}.

Substituting $\bm{u}_g = \bm{z}_g^{\underline{\rho}_g}$ into
\eqref{eq: cond prob u form}, the numerator becomes
\begin{equation*}
\Gamma^{\bm{z}_g^{\underline{\rho}_g \top}
\bm{z}_g^{\underline{\rho}_g}} = \Gamma^{\sum_{i}
\left(z_{g,i}^{\underline{\rho}_g}\right)^2} = \Gamma^{n_{g,1}},
\end{equation*}
using $\left(z_{g,i}^{\underline{\rho}_g}\right)^2 =
z_{g,i}^{\underline{\rho}_g}$ and
$\sum_{i = 1}^{\tilde{n}_g^{\underline{\rho}_g}}
z_{g,i}^{\underline{\rho}_g} = n_{g,1}$. The denominator becomes
$\sum_{\bm{a}_g \in \Omega_g^{\underline{\rho}_g}}
\Gamma^{\bm{a}_g^{\top} \bm{z}_g^{\underline{\rho}_g}}$, yielding
\eqref{eq: ub cond z prob}.
\end{proof}

\paragraph{Reduction to Fogarty's one-per-stratum case.} \Cref{lem: prob bounds restate} generalizes the probability bounds from the case of one minority-civilian encounter per stratum considered by \citet{fogarty2023} to arbitrary post-stratified designs. We verify the reduction. In Fogarty's special case, $\abs{\Omega_g^{\underline{\rho}_g}} = \tilde{n}_g^{\underline{\rho}_g}$, and the inner products $\bm{a}_g^{\top}(\bm{1} - \bm{z}_g^{\underline{\rho}_g})$ and $\bm{a}_g^{\top}\bm{z}_g^{\underline{\rho}_g}$ take only the values $0$ and $1$. If $\bm{a}_g \neq \bm{z}_g^{\underline{\rho}_g}$, then $\bm{a}_g^{\top}(\bm{1} - \bm{z}_g^{\underline{\rho}_g}) = 1$, since $\bm{z}_g^{\underline{\rho}_g}$ places its single $1$ on exactly one unit; for the unique $\bm{a}_g = \bm{z}_g^{\underline{\rho}_g}$, we instead have $\bm{a}_g^{\top}(\bm{1} - \bm{z}_g^{\underline{\rho}_g}) = 0$. Substituting into the lower bound from \Cref{lem: prob bounds restate} yields
\begin{align}\label{eq: supp lb cond z prob fogarty}
\underline{p}\left(\bm{z}_g^{\underline{\rho}_g}; \Gamma\right) & = \dfrac{1}{\sum_{\bm{a}_g \in \Omega_g^{\underline{\rho}_g}} \Gamma^{\bm{a}_g^{\top}\left(\bm{1}-\bm{z}_g^{\underline{\rho}_g}\right)}} = \dfrac{1}{\Gamma (\tilde{n}_g^{\underline{\rho}_g} - 1) + 1},
\end{align}
which matches the left-hand side of the bound in (5) of \citet[][p.~2201]{fogarty2023}. A parallel simplification occurs for the upper bound. With one minority-civilian encounter, $\bm{a}_g^{\top}\bm{z}_g^{\underline{\rho}_g} = 0$ for all $\bm{a}_g \neq \bm{z}_g^{\underline{\rho}_g}$ and equals $1$ only for $\bm{a}_g = \bm{z}_g^{\underline{\rho}_g}$. Substituting yields
\begin{align}\label{eq: supp ub cond z prob fogarty}
\overline{p}\left(\bm{z}_g^{\underline{\rho}_g}; \Gamma\right) & = \dfrac{\Gamma^{n_{g,1}}}{\sum_{\bm{a}_g \in \Omega_g^{\underline{\rho}_g}} \Gamma^{\bm{a}_g^{\top}\bm{z}_g^{\underline{\rho}_g}}} = \dfrac{\Gamma}{(\tilde{n}_g^{\underline{\rho}_g} - 1) + \Gamma},
\end{align}
matching the right-hand side of the bound in (5) of \citet[][p.~2201]{fogarty2023}.

\subsection{Efficient Computation of Probability Bounds}\label{sec: supp efficient computation}

The lower and upper probability bounds in equations~(10) and~(11) of the manuscript are sums over all assignment vectors in $\Omega_g^{\underline{\rho}_g}$. Because $\abs{\Omega_g^{\underline{\rho}_g}} = \binom{\tilde{n}_g}{n_{g,1}}$ is large, direct enumeration is computationally prohibitive. We now show that both denominators can be computed in closed form.

\begin{prop}\label{prop: supp efficient denom}
For any $\bm{z}_g^{\underline{\rho}_g} \in \Omega_g^{\underline{\rho}_g}$ with $n_{g,1}$ treated units and $\tilde{n}_g - n_{g,1}$ control units, the denominators of the lower and upper probability bounds are
\begin{align}
\sum_{\bm{a}_g \in \Omega_g^{\underline{\rho}_g}} \Gamma^{\bm{a}_g^{\top}(\bm{1} - \bm{z}_g^{\underline{\rho}_g})} &= \sum_{j=0}^{n_{g,1}} \binom{n_{g,1}}{j}\binom{\tilde{n}_g - n_{g,1}}{n_{g,1} - j} \Gamma^{n_{g,1} - j}, \label{eq: supp denom lb} \\[0.5em]
\sum_{\bm{a}_g \in \Omega_g^{\underline{\rho}_g}} \Gamma^{\bm{a}_g^{\top}\bm{z}_g^{\underline{\rho}_g}} &= \sum_{j=0}^{n_{g,1}} \binom{n_{g,1}}{j}\binom{\tilde{n}_g - n_{g,1}}{n_{g,1} - j} \Gamma^{j}. \label{eq: supp denom ub}
\end{align}
\end{prop}

\begin{proof}
For the lower bound denominator, observe that $\bm{a}_g^{\top}(\bm{1} - \bm{z}_g^{\underline{\rho}_g})$ counts the number of units assigned to the treatment condition in $\bm{a}_g$ but not in $\bm{z}_g^{\underline{\rho}_g}$ --- i.e., the number of positions at which $\bm{a}_g$ places a $1$ where $\bm{z}_g^{\underline{\rho}_g}$ has a $0$. Let $j$ denote the number of positions at which both $\bm{a}_g$ and $\bm{z}_g^{\underline{\rho}_g}$ have a $1$ (i.e., the overlap). Then $\bm{a}_g^{\top}(\bm{1} - \bm{z}_g^{\underline{\rho}_g}) = n_{g,1} - j$, since $\bm{a}_g$ has $n_{g,1}$ ones in total and $j$ of them coincide with those of $\bm{z}_g^{\underline{\rho}_g}$.

Since both $\bm{a}_g$ and $\bm{z}_g^{\underline{\rho}_g}$ contain exactly $n_{g,1}$ ones, the number of overlapping ones $j = \bm{a}_g^\top \bm{z}_g^{\underline{\rho}_g}$ must lie between $0$ and $n_{g,1}$. For a given overlap $j \in \{0,\ldots,n_{g,1}\}$, consider assignment vectors $\bm{a}_g \in \Omega_g^{\underline{\rho}_g}$ with exactly $j$ positions at which both $\bm a_g$ and $\bm z_g^{\underline{\rho}_g}$ equal $1$. Since $\bm z_g^{\underline{\rho}_g}$ has $n_{g,1}$ ones, there are $n_{g,1}$ positions at which an overlap can occur, and choosing the $j$ overlapping positions can be done in $\binom{n_{g,1}}{j}$ ways. The vector $\bm a_g$ must contain $n_{g,1}$ ones in total, so the remaining $n_{g,1}-j$ ones must occur among the $\tilde n_g - n_{g,1}$ positions at which $\bm{z}_g^{\underline{\rho}_g}$ has a $0$, which can be chosen in $\binom{\tilde{n}_g - n_{g,1}}{n_{g,1} - j}$ ways. Thus the number of assignment vectors with overlap $j$ equals
$\binom{n_{g,1}}{j}\binom{\tilde{n}_g - n_{g,1}}{n_{g,1} - j}$, and grouping the sum
by overlap $j$ yields \eqref{eq: supp denom lb}.

The upper bound denominator follows identically: $\bm{a}_g^{\top}\bm{z}_g^{\underline{\rho}_g} = j$, and the same combinatorial argument yields \eqref{eq: supp denom ub}.
\end{proof}

Each sum in Proposition~\ref{prop: supp efficient denom} has at most $n_{g,1} + 1$ terms, so both probability bounds can be computed in time proportional to $n_{g,1}$.

\subsection{Proof of Proposition 2}\label{sec: supp proof prop 2}

Proposition~2 of the manuscript establishes that the tilted IPW Difference-in-Means is conservative under the null. We restate the proposition for reference.

\begin{prop}[Restatement of Proposition~2 of the manuscript]\label{prop: supp tilt bound restate}
Under Assumptions~1--4 of the manuscript, the tilted statistic in equation~(13) of the manuscript has nonpositive expectation under the null when the alternative is upper-tailed and nonnegative expectation under the null when the alternative is lower-tailed. Specifically, for any $\bm{\underline{\rho}} \in [0,1)^{\abs{\mathcal{G}^{\ast}}}$, with each stratum augmented by $\tilde{n}_{g,0,\mathrm{OMS}}^{\underline{\rho}_g}$ missing units as in equation~(9) of the manuscript, and any $\Gamma \geq 1$,
\begin{align}
\E\left[\hat{\tau}^{\mathrm{tilt}}\left(\bm{\underline{\rho}};\Gamma,\tau_0,d=+1\right)\right] &\leq 0 \quad\text{(upper-tailed alternative)},\label{eq: supp tilt upper}\\
\E\left[\hat{\tau}^{\mathrm{tilt}}\left(\bm{\underline{\rho}};\Gamma,\tau_0,d=-1\right)\right] &\geq 0 \quad\text{(lower-tailed alternative)}.\label{eq: supp tilt lower}
\end{align}
\end{prop}

\begin{proof}
We prove the upper-tailed case $d = +1$; the lower-tailed case follows by a symmetric argument.

\medskip
\noindent\textit{Step 1: Decompose the stratum-level expectation.}\quad Fix a stratum $g \in \mathcal{G}^{\ast}$. The stratum-level tilted statistic \eqref{eq: tilted IPW diff-in-means g} with $d = +1$ applies $\overline{p}(\bm{z}_g^{\underline{\rho}_g};\Gamma)^{-1}$ when $\hat{\tau}_g^{\underline{\rho}_g} - \tau_0 \geq 0$ and $\underline{p}(\bm{z}_g^{\underline{\rho}_g};\Gamma)^{-1}$ when $\hat{\tau}_g^{\underline{\rho}_g} - \tau_0 < 0$. Writing $p(\bm{z}_g^{\underline{\rho}_g}) \coloneqq \Pr(\bm{Z}_g = \bm{z}_g^{\underline{\rho}_g})$ for the true (unknown) assignment probability, the expectation over the assignment mechanism is
\scriptsize
\begin{equation}\label{eq: supp tilt expectation decomp}
\begin{split}
\E\left[\hat{\tau}_g^{\mathrm{tilt}}\left(\underline{\rho}_g;\Gamma,\tau_0,+1\right)\right]
= \dfrac{1}{\abs{\Omega_g^{\underline{\rho}_g}}}
&\sum_{\bm{z}_g^{\underline{\rho}_g} \in \Omega_g^{\underline{\rho}_g}}
p\left(\bm{z}_g^{\underline{\rho}_g}\right)
\left(\hat{\tau}_g^{\underline{\rho}_g} - \tau_0\right) \left[
  \dfrac{\mathbbm{1}\left\{\hat{\tau}_g^{\underline{\rho}_g} - \tau_0 \geq 0\right\}}
       {\overline{p}\left(\bm{z}_g^{\underline{\rho}_g};\Gamma\right)}
  \;+\;
  \dfrac{\mathbbm{1}\left\{\hat{\tau}_g^{\underline{\rho}_g} - \tau_0 < 0\right\}}
       {\underline{p}\left(\bm{z}_g^{\underline{\rho}_g};\Gamma\right)}
\right].
\end{split}
\end{equation}
\normalsize
\medskip
\noindent\textit{Step 2: Upper-bound each summand.}\quad We show that every summand in \eqref{eq: supp tilt expectation decomp} is bounded above by $(\hat{\tau}_g^{\underline{\rho}_g} - \tau_0)$. This yields the largest possible value of the right-hand side of \eqref{eq: supp tilt expectation decomp}: if even this upper bound is nonpositive under the null, the true expectation must be as well, regardless of the assignment mechanism.

By Lemma~1 of the manuscript, any assignment mechanism satisfying the model in equation~(1) of the manuscript satisfies $\underline{p}(\bm{z}_g^{\underline{\rho}_g};\Gamma) \leq p(\bm{z}_g^{\underline{\rho}_g}) \leq \overline{p}(\bm{z}_g^{\underline{\rho}_g};\Gamma)$ for all $\bm{z}_g^{\underline{\rho}_g} \in \Omega_g^{\underline{\rho}_g}$. For each assignment, exactly one indicator in \eqref{eq: supp tilt expectation decomp} is nonzero.

\begin{enumerate}
\item[\textit{(i)}] When $\hat{\tau}_g^{\underline{\rho}_g} - \tau_0 \geq 0$, the active indicator selects $\overline{p}^{-1}$. Because $p(\bm{z}_g^{\underline{\rho}_g}) \leq \overline{p}(\bm{z}_g^{\underline{\rho}_g};\Gamma)$, the ratio $p/\overline{p} \leq 1$. Multiplying a nonnegative quantity by a ratio at most one yields a weakly smaller value:
\begin{align}\label{eq: supp bound first sum}
\dfrac{p\left(\bm{z}_g^{\underline{\rho}_g}\right)}
     {\overline{p}\left(\bm{z}_g^{\underline{\rho}_g};\Gamma\right)}\,
\left(\hat{\tau}_g^{\underline{\rho}_g} - \tau_0\right)
\;\leq\;
\left(\hat{\tau}_g^{\underline{\rho}_g} - \tau_0\right).
\end{align}

\item[\textit{(ii)}] When $\hat{\tau}_g^{\underline{\rho}_g} - \tau_0 < 0$, the active indicator selects $\underline{p}^{-1}$. Because $p(\bm{z}_g^{\underline{\rho}_g}) \geq \underline{p}(\bm{z}_g^{\underline{\rho}_g};\Gamma)$, the ratio $p/\underline{p} \geq 1$. Multiplying a negative quantity by a ratio at least one makes the product more negative --- further from zero and further in the direction opposite the upper-tailed alternative --- yielding a weakly smaller value:
\begin{align}\label{eq: supp bound second sum}
\dfrac{p\left(\bm{z}_g^{\underline{\rho}_g}\right)}
     {\underline{p}\left(\bm{z}_g^{\underline{\rho}_g};\Gamma\right)}\,
\left(\hat{\tau}_g^{\underline{\rho}_g} - \tau_0\right)
\;\leq\;
\left(\hat{\tau}_g^{\underline{\rho}_g} - \tau_0\right).
\end{align}
\end{enumerate}

\noindent In both cases the summand is bounded above by $(\hat{\tau}_g^{\underline{\rho}_g} - \tau_0)$. The bound holds because the quantity being scaled by the probability ratio --- namely $(\hat{\tau}_g^{\underline{\rho}_g} - \tau_0)$ --- has the same sign as the condition that selects the ratio: nonnegative quantities are paired with $p/\overline{p} \leq 1$, and negative quantities are paired with $p/\underline{p} \geq 1$. Centering at $\tau_0$ is what ensures this alignment. Without centering, the tilting would operate on $\hat{\tau}_g^{\underline{\rho}_g}$ directly. In the region $0 \leq \hat{\tau}_g^{\underline{\rho}_g} < \tau_0$, the condition $\hat{\tau}_g^{\underline{\rho}_g} - \tau_0 < 0$ would select the ratio $p/\underline{p} \geq 1$, but $\hat{\tau}_g^{\underline{\rho}_g}$ itself is nonnegative, so multiplying by $p/\underline{p} \geq 1$ would \emph{increase} the product, violating the required upper bound.

\medskip
\noindent\textit{Step 3: Recombine.}\quad Substituting \eqref{eq: supp bound first sum} and \eqref{eq: supp bound second sum} into \eqref{eq: supp tilt expectation decomp} and using $\mathbbm{1}\{\hat{\tau}_g^{\underline{\rho}_g} - \tau_0 \geq 0\} + \mathbbm{1}\{\hat{\tau}_g^{\underline{\rho}_g} - \tau_0 < 0\} = 1$ to recombine the two cases into a single sum over all assignments:
\begin{align}\label{eq: supp tilt ub stratum}
\E\left[\hat{\tau}_g^{\mathrm{tilt}}\left(\underline{\rho}_g;\Gamma,\tau_0,+1\right)\right]
\;\leq\;
\dfrac{1}{\abs{\Omega_g^{\underline{\rho}_g}}}
\sum_{\bm{z}_g^{\underline{\rho}_g} \in \Omega_g^{\underline{\rho}_g}}
\left(\hat{\tau}_g^{\underline{\rho}_g} - \tau_0\right).
\end{align}
The right-hand side of \eqref{eq: supp tilt ub stratum} is a deterministic quantity: it averages the augmented Difference-in-Means $\hat{\tau}_g^{\underline{\rho}_g}$ over all $\abs{\Omega_g^{\underline{\rho}_g}}$ elements of $\Omega_g^{\underline{\rho}_g}$ with equal weight, then subtracts $\tau_0$. Because each unit appears as treated in the same number of assignment vectors, this uniform average equals the stratum-level average treatment effect $\tau_g$ by a combinatorial identity:
\begin{align}\label{eq: supp combinatorial identity}
\dfrac{1}{\abs{\Omega_g^{\underline{\rho}_g}}}
\sum_{\bm{z}_g^{\underline{\rho}_g} \in \Omega_g^{\underline{\rho}_g}}
\hat{\tau}_g^{\underline{\rho}_g}
\;=\; \tau_g.
\end{align}
Substituting \eqref{eq: supp combinatorial identity} into \eqref{eq: supp tilt ub stratum} gives
\begin{align}\label{eq: supp tilt stratum centered}
\E\left[\hat{\tau}_g^{\mathrm{tilt}}\left(\underline{\rho}_g;\Gamma,\tau_0,+1\right)\right]
\;\leq\;
\tau_g - \tau_0.
\end{align}

\medskip
\noindent\textit{Step 4: Aggregate over strata.}\quad The aggregate tilted statistic in equation~(13) of the manuscript is $\hat{\tau}^{\mathrm{tilt}} = \sum_{g \in \mathcal{G}^{\ast}} (\tilde{n}_g^{\underline{\rho}_g}/\tilde{n}^{\ast})\,\hat{\tau}_g^{\mathrm{tilt}}$, where $\tilde{n}^{\ast} \coloneqq \sum_{g \in \mathcal{G}^{\ast}} \tilde{n}_g^{\underline{\rho}_g}$. Applying \eqref{eq: supp tilt stratum centered} to each term,
\begin{align*}
\E\left[\hat{\tau}^{\mathrm{tilt}}\left(\bm{\underline{\rho}};\Gamma,\tau_0,+1\right)\right]
&\leq \sum_{g \in \mathcal{G}^{\ast}}
  \dfrac{\tilde{n}_g^{\underline{\rho}_g}}{\tilde{n}^{\ast}}\left(\tau_g - \tau_0\right)
= \sum_{g \in \mathcal{G}^{\ast}}
  \dfrac{\tilde{n}_g^{\underline{\rho}_g}}{\tilde{n}^{\ast}}\,\tau_g
  \;-\; \tau_0
= \tau - \tau_0,
\end{align*}
where the second equality uses $\sum_{g \in \mathcal{G}^{\ast}} \tilde{n}_g^{\underline{\rho}_g}/\tilde{n}^{\ast} = 1$ and the definition $\tau \coloneqq \sum_{g \in \mathcal{G}^{\ast}} (\tilde{n}_g^{\underline{\rho}_g}/\tilde{n}^{\ast})\,\tau_g$. Under the null $\tau \leq \tau_0$, this gives $\tau - \tau_0 \leq 0$, establishing \eqref{eq: supp tilt upper}.

\medskip
\noindent\textit{Lower-tailed case.}\quad When $d = -1$, the tilting applies $\underline{p}^{-1}$ when $\hat{\tau}_g^{\underline{\rho}_g} - \tau_0 \leq 0$ and $\overline{p}^{-1}$ when $\hat{\tau}_g^{\underline{\rho}_g} - \tau_0 > 0$. The analogous argument --- $p/\underline{p} \geq 1$ on nonpositive quantities leaves them unchanged or pushes them closer to zero; $p/\overline{p} \leq 1$ on positive quantities makes them less positive --- shows each summand is bounded \emph{below} by $(\hat{\tau}_g^{\underline{\rho}_g} - \tau_0)$. Recombining via \eqref{eq: supp combinatorial identity} gives $\E[\hat{\tau}_g^{\mathrm{tilt}}(\underline{\rho}_g;\Gamma,\tau_0,-1)] \geq \tau_g - \tau_0$, and aggregation yields \eqref{eq: supp tilt lower} under the null $\tau \geq \tau_0$.
\end{proof}

\subsection{Corollary for Stratum-Specific Sensitivity Parameters}
\label{sec: supp corollary}

Proposition~\ref{prop: supp tilt bound restate} is stated for a uniform sensitivity parameter $\Gamma$ that applies identically to all strata. As described in Section~6.2 of the manuscript, the geographic calibration analysis replaces $\Gamma$ with stratum-specific bounds $\bm{\Gamma} \coloneqq (\Gamma_g)_{g \in \mathcal{G}^{\ast}}$, where the operative bound for stratum $g$ is $\min(\Gamma, \Gamma_g^{\mathrm{geo}})$. The following corollary extends Proposition~\ref{prop: supp tilt bound restate} to this setting.

\begin{cor}[Extension to stratum-specific sensitivity parameters]
\label{cor: supp stratum Gamma}
Under Assumptions 1--4 of the manuscript, the conclusion of Proposition~\ref{prop: supp tilt bound restate} continues to hold when the uniform sensitivity parameter $\Gamma$ is replaced by stratum-specific parameters $\bm{\Gamma} = (\Gamma_g)_{g \in \mathcal{G}^{\ast}}$, with each $\Gamma_g \geq 1$. Specifically, define the stratum-level tilted statistic with stratum-specific bounds as
\begin{align}\label{eq: supp tilt stratum Gamma}
\hat{\tau}_g^{\mathrm{tilt}}\left(\underline{\rho}_g;\Gamma_g,\tau_0,d\right)
= \dfrac{1}{\abs{\Omega_g^{\underline{\rho}_g}}}
\left(\hat{\tau}_g^{\underline{\rho}_g} - \tau_0\right)
\begin{cases}
\overline{p}\left(\bm{z}_g^{\underline{\rho}_g};\Gamma_g\right)^{-1},
& \text{if } d\left(\hat{\tau}_g^{\underline{\rho}_g} - \tau_0\right) \geq 0,\\[0.5em]
\underline{p}\left(\bm{z}_g^{\underline{\rho}_g};\Gamma_g\right)^{-1},
& \text{if } d\left(\hat{\tau}_g^{\underline{\rho}_g} - \tau_0\right) < 0,
\end{cases}
\end{align}
and the aggregate tilted statistic as
\begin{align}\label{eq: supp tilt agg stratum Gamma}
\hat{\tau}^{\mathrm{tilt}}\left(\bm{\underline{\rho}};\bm{\Gamma},\tau_0,d\right)
\coloneqq \sum_{g \in \mathcal{G}^{\ast}}
\left(\tilde{n}_g^{\underline{\rho}_g} / \tilde{n}^{\ast}\right)\,
\hat{\tau}_g^{\mathrm{tilt}}\left(\underline{\rho}_g;\Gamma_g,\tau_0,d\right).
\end{align}
Then, for any $\bm{\underline{\rho}} \in [0,1)^{\abs{\mathcal{G}^{\ast}}}$ and any $\bm{\Gamma} \in [1,\infty)^{\abs{\mathcal{G}^{\ast}}}$,
\begin{align}
\E\left[\hat{\tau}^{\mathrm{tilt}}\left(\bm{\underline{\rho}};\bm{\Gamma},\tau_0,d=+1\right)\right]
&\leq 0 \quad\text{(upper-tailed alternative)},\label{eq: supp cor tilt upper}\\
\E\left[\hat{\tau}^{\mathrm{tilt}}\left(\bm{\underline{\rho}};\bm{\Gamma},\tau_0,d=-1\right)\right]
&\geq 0 \quad\text{(lower-tailed alternative)}.\label{eq: supp cor tilt lower}
\end{align}
\end{cor}

\begin{proof}
The proof of Proposition~\ref{prop: supp tilt bound restate} establishes \eqref{eq: supp cor tilt upper} and \eqref{eq: supp cor tilt lower} stratum-by-stratum, and each step uses only the within-stratum probability bounds from Lemma~1 of the manuscript. We sketch the modifications needed for stratum-specific $\Gamma_g$ and refer to the proof of Proposition~\ref{prop: supp tilt bound restate} for details.

The restriction on the assignment model in equation~(1) of the manuscript constrains odds ratios among pairs of units \emph{within the same stratum}. Replacing the uniform $\Gamma$ with stratum-specific $\Gamma_g$ gives, for each $g \in \mathcal{G}^{\ast}$,
\begin{align*}
\dfrac{1}{\Gamma_g}
\;\leq\;
\dfrac{\pi_{g,i}(1-\pi_{g,j})}{\pi_{g,j}(1-\pi_{g,i})}
\;\leq\;
\Gamma_g
\qquad \text{for all } i, j \in \{1,\ldots,n_g\}.
\end{align*}
This is a strictly weaker restriction than requiring a uniform $\Gamma = \max_{g} \Gamma_g$. Lemma~1 of the manuscript, whose proof operates entirely within a single stratum, therefore holds with $\Gamma_g$ in place of $\Gamma$, yielding stratum-specific probability bounds:
\begin{align*}
\underline{p}\left(\bm{z}_g^{\underline{\rho}_g};\Gamma_g\right)
\;\leq\; p\left(\bm{z}_g^{\underline{\rho}_g}\right)
\;\leq\; \overline{p}\left(\bm{z}_g^{\underline{\rho}_g};\Gamma_g\right)
\qquad\text{for all } \bm{z}_g^{\underline{\rho}_g} \in \Omega_g^{\underline{\rho}_g}.
\end{align*}

Steps~1--3 of the proof of Proposition~\ref{prop: supp tilt bound restate} use only the within-stratum probability bounds in stratum $g$. Substituting $\Gamma_g$ for $\Gamma$ throughout yields the stratum-level bound
\begin{align}\label{eq: supp cor stratum bound}
\E\left[\hat{\tau}_g^{\mathrm{tilt}}\left(\underline{\rho}_g;\Gamma_g,\tau_0,+1\right)\right]
\;\leq\; \tau_g - \tau_0,
\end{align}
the stratum-specific analog of equation~\eqref{eq: supp tilt stratum centered}. Step~4 then aggregates via linearity of expectation --- which requires no assumption about cross-stratum dependence --- and does not depend on any sensitivity parameter:
\begin{align*}
\E\left[\hat{\tau}^{\mathrm{tilt}}\left(\bm{\underline{\rho}};\bm{\Gamma},\tau_0,+1\right)\right]
&\leq \sum_{g \in \mathcal{G}^{\ast}}
  \dfrac{\tilde{n}_g^{\underline{\rho}_g}}{\tilde{n}^{\ast}}\left(\tau_g - \tau_0\right)
= \tau - \tau_0
\;\leq\; 0,
\end{align*}
where the last inequality holds under the null $\tau \leq \tau_0$. This establishes \eqref{eq: supp cor tilt upper}. The argument for \eqref{eq: supp cor tilt lower} is symmetric.
\end{proof}

\subsection{Conservative variance estimation} \label{sec: consv var}

For reference, the tilted statistic from equation~(13) of the manuscript is
\begin{equation} \label{eq: tilted stat restate}
\hat{\tau}^{\mathrm{tilt}}\left(\bm{\underline{\rho}}; \Gamma, \tau_0,
d\right) = \sum_{g \in \mathcal{G}^{\ast}}
\left(\tilde{n}_g^{\underline{\rho}_g}/\tilde{n}^{\ast}\right)
\hat{\tau}_g^{\mathrm{tilt}}\left(\underline{\rho}_g; \Gamma, \tau_0,
d\right),
\end{equation}
where $\tilde{n}^{\ast} \coloneqq \sum_{g \in \mathcal{G}^{\ast}}
\tilde{n}_g^{\underline{\rho}_g}$ and each stratum-level tilted statistic
$\hat{\tau}_g^{\mathrm{tilt}}$ is defined in equation~(12) of the
manuscript. Because the civilian-race indicators are mutually independent
Bernoulli random variables under the assignment model in Section~3.3 of
the manuscript, assignments in distinct strata are independent and the
variance of the tilted statistic decomposes as
\begin{equation} \label{eq: true var}
\Var\left[\hat{\tau}^{\mathrm{tilt}}\right] = \sum_{g \in
\mathcal{G}^{\ast}}
\left(\tilde{n}_g^{\underline{\rho}_g}/\tilde{n}^{\ast}\right)^{2}
\Var\left[\hat{\tau}_g^{\mathrm{tilt}}\right].
\end{equation}

Following \citet{fogarty2018a,fogarty2023}, we estimate this variance with
an HC2-style estimator. Define the scaled stratum weight $w_g \coloneqq
\abs{\mathcal{G}^{\ast}}\,
\tilde{n}_g^{\underline{\rho}_g}/\tilde{n}^{\ast}$, the weight vector
$\bm{Q} \coloneqq \left(w_g\right)_{g \in \mathcal{G}^{\ast}}$, the hat
matrix $\bm{H}_{\bm{Q}} \coloneqq
\bm{Q}\left(\bm{Q}^{\top}\bm{Q}\right)^{-1}\bm{Q}^{\top}$, and the sum
of squared weights $S_w \coloneqq \sum_{g \in \mathcal{G}^{\ast}} w_g^2$.
Let $h_g \coloneqq w_g^2/S_w$ denote the $g$th diagonal entry of
$\bm{H}_{\bm{Q}}$, and define the leverage-adjusted vector
$\tilde{\bm{y}} \coloneqq \left(w_g\,
\hat{\tau}_g^{\mathrm{tilt}}/\sqrt{1 - h_g}\right)_{g \in
\mathcal{G}^{\ast}}$. The variance estimator is
\begin{equation} \label{eq: var est}
\widehat{\mathrm{se}}^2 \coloneqq
\dfrac{1}{\abs{\mathcal{G}^{\ast}}^2}\, \tilde{\bm{y}}^{\top}\left(\bm{I}
- \bm{H}_{\bm{Q}}\right)\tilde{\bm{y}}.
\end{equation}

Assumptions~1--4 of the manuscript justify the construction of the
augmented data and the tilted statistic in \eqref{eq: tilted stat restate},
but the conservativeness of the variance estimator itself relies only on
cross-stratum independence and on having at least two informative strata.

\begin{prop} \label{prop: consv var}
Consider the tilted statistic in equation~(13) of the manuscript, restated
in \eqref{eq: tilted stat restate}. Suppose the treatment indicators
$\{Z_{g,i} : i = 1, \ldots, \tilde{n}_g^{\underline{\rho}_g},\; g \in
\mathcal{G}^{\ast}\}$ are mutually independent Bernoulli random variables,
as specified by the assignment model in Section~3.3 of the manuscript, and
that $\abs{\mathcal{G}^{\ast}} \geq 2$. Then the variance estimator in
\eqref{eq: var est} satisfies
\begin{equation*}
\E\left[\widehat{\mathrm{se}}^2\right] \geq
\Var\left[\hat{\tau}^{\mathrm{tilt}}\left(\bm{\underline{\rho}}; \Gamma,
\tau_0, d\right)\right]
\end{equation*}
for all $\bm{\underline{\rho}} \in [0,1)^{\abs{\mathcal{G}^{\ast}}}$, all
$\Gamma \geq 1$, all potential outcome configurations, and all assignment
mechanisms consistent with the $\Gamma$-bound in equation~(1) of the
manuscript.
\end{prop}

The variance estimator is conservative for the following reason. The
estimator is built from the squared stratum-level tilted statistics
$\hat{\lambda}_g^2$. The expectation of each squared statistic decomposes
as
\begin{equation*}
\E[\hat{\lambda}_g^2] = \Var[\hat{\lambda}_g] +
\E[\hat{\lambda}_g]^2,
\end{equation*}
so a weighted sum of $\hat{\lambda}_g^2$ terms has expectation equal to
the true variance plus a non-negative excess from the squared
stratum-level expectations $\E[\hat{\lambda}_g]^2$. If these expectations
were known, one could eliminate the excess by centering each statistic
before squaring --- using $(\hat{\lambda}_g - \E[\hat{\lambda}_g])^2$ in
place of $\hat{\lambda}_g^2$. Under the composite null $\tau = \tau_0$,
however, many configurations of individual potential outcomes are
consistent with the same overall average $\tau_0$, and different
configurations produce different stratum-level expectations. Because these
expectations depend on unknowable potential outcomes, the centering cannot
be performed, and the excess remains.

The variance estimator in \eqref{eq: var est} partially corrects for this
excess. The estimator computes
$\tilde{\bm{y}}^{\top}(\bm{I} - \bm{H}_{\bm{Q}})\tilde{\bm{y}}$, where
$\tilde{\bm{y}}$ is a vector of leverage-adjusted stratum-level statistics
and $\bm{H}_{\bm{Q}}$ is the hat matrix formed from the stratum weight
vector $\bm{Q}$. The residual-maker $\bm{I} - \bm{H}_{\bm{Q}}$
has diagonal entries $1 - w_g^2/S_w$, which shrink each
stratum's squared contribution, and off-diagonal entries
$-w_g w_{g^{\prime}}/S_w$, which subtract pairwise cross-products of the
stratum-level statistics.

The consequence of premultiplying and postmultiplying
$\tilde{\bm{y}}$ by this matrix is to project out the overall weighted
average: The diagonal entries of $\bm{I} - \bm{H}_{\bm{Q}}$, each equal to
$1 - w_g^2/S_w$, scale down each stratum's squared contribution but do
not on their own remove the squared-expectation excess. The off-diagonal
entries, each equal to $-w_g w_{g^{\prime}}/S_w$, are negative, and when
the vector $\tilde{\bm{y}}$ is multiplied on both sides of the matrix ---
that is, in the product
$\tilde{\bm{y}}^{\top}(\bm{I} - \bm{H}_{\bm{Q}})\tilde{\bm{y}}$ --- each
off-diagonal entry gets multiplied by $\tilde{y}_g$ from the left and
$\tilde{y}_{g^{\prime}}$ from the right, producing terms proportional to
$-\hat{\lambda}_g \hat{\lambda}_{g^{\prime}}$. Taking expectations of these
$-\hat{\lambda}_g \hat{\lambda}_{g^{\prime}}$ terms, cross-stratum
independence implies
$\E[\hat{\lambda}_g \hat{\lambda}_{g^{\prime}}] =
\E[\hat{\lambda}_g]\,\E[\hat{\lambda}_{g^{\prime}}]$, so the
off-diagonal subtractions remove quantities built from the same
stratum-level expectations whose squares constitute the diagonal excess.

These subtractions, however, never fully absorb the excess. By the AM-GM
inequality, the sum of the cross-products is always strictly smaller in
magnitude than the sum of the diagonal squared-expectation terms whenever
$\abs{\mathcal{G}^{\ast}} \geq 2$, so a non-negative remainder is left
over. The proof below makes this argument precise.

\begin{proof}
For brevity, write $\hat{\lambda}_g \coloneqq
\hat{\tau}_g^{\mathrm{tilt}}\left(\underline{\rho}_g; \Gamma, \tau_0,
d\right)$ for each $g \in \mathcal{G}^{\ast}$. Each $\hat{\lambda}_g$ is
a function of the assignment vector $\bm{Z}_g = (Z_{g,1}, \ldots,
Z_{g,\tilde{n}_g^{\underline{\rho}_g}})$ and the fixed potential outcomes
within stratum $g$. Because the $Z_{g,i}$ are mutually independent across
all strata and all units by assumption, the vectors $\bm{Z}_g$ and
$\bm{Z}_{g^{\prime}}$ are independent for $g \neq g^{\prime}$, and
therefore $\hat{\lambda}_g$ and $\hat{\lambda}_{g^{\prime}}$ are
independent for $g \neq g^{\prime}$.

The hat matrix $\bm{H}_{\bm{Q}}$ has entries $w_g w_{g^{\prime}}/S_w$, so
the residual matrix $\bm{I} - \bm{H}_{\bm{Q}}$ has diagonal entries
$1 - w_g^2/S_w = (S_w - w_g^2)/S_w$ and off-diagonal entries
$-w_g w_{g^{\prime}}/S_w$. Substituting
$\tilde{y}_g = w_g \hat{\lambda}_g/\sqrt{1 - h_g} = w_g
\hat{\lambda}_g \sqrt{S_w}/\sqrt{S_w - w_g^2}$ and expanding
$\tilde{\bm{y}}^{\top}(\bm{I} - \bm{H}_{\bm{Q}})\tilde{\bm{y}}$ yields
\begin{equation} \label{eq: quadratic expanded}
\tilde{\bm{y}}^{\top}\left(\bm{I} -
\bm{H}_{\bm{Q}}\right)\tilde{\bm{y}} = \sum_{g \in \mathcal{G}^{\ast}}
w_g^2 \hat{\lambda}_g^2 - 2 \sum_{g \in \mathcal{G}^{\ast}}
\sum_{\substack{g^{\prime} \in \mathcal{G}^{\ast} \\ g^{\prime} > g}}
\dfrac{w_g^2 w_{g^{\prime}}^2 \hat{\lambda}_g
\hat{\lambda}_{g^{\prime}}}{\sqrt{S_w - w_g^2}\sqrt{S_w -
w_{g^{\prime}}^2}}.
\end{equation}
Dividing by $\abs{\mathcal{G}^{\ast}}^2$ and recalling that $w_g =
\abs{\mathcal{G}^{\ast}}\,
\tilde{n}_g^{\underline{\rho}_g}/\tilde{n}^{\ast}$, the variance
estimator becomes
\begin{equation} \label{eq: var est expanded}
\widehat{\mathrm{se}}^2 = \sum_{g \in \mathcal{G}^{\ast}}
\left(\dfrac{\tilde{n}_g^{\underline{\rho}_g}}{\tilde{n}^{\ast}}\right)^{2}
\hat{\lambda}_g^2 - \dfrac{2}{\abs{\mathcal{G}^{\ast}}^2} \sum_{g \in
\mathcal{G}^{\ast}} \sum_{\substack{g^{\prime} \in \mathcal{G}^{\ast} \\
g^{\prime} > g}} \dfrac{w_g^2 w_{g^{\prime}}^2 \hat{\lambda}_g
\hat{\lambda}_{g^{\prime}}}{\sqrt{S_w - w_g^2}\sqrt{S_w -
w_{g^{\prime}}^2}}.
\end{equation}

Taking expectations of \eqref{eq: var est expanded}, the cross-stratum
independence established above implies
$\E[\hat{\lambda}_g \hat{\lambda}_{g^{\prime}}] =
\E[\hat{\lambda}_g]\,\E[\hat{\lambda}_{g^{\prime}}]$ for $g \neq
g^{\prime}$, and the standard second-moment decomposition gives
$\E[\hat{\lambda}_g^2] = \Var[\hat{\lambda}_g] +
\E[\hat{\lambda}_g]^2$. Applying both identities yields
\footnotesize
\begin{equation} \label{eq: E var est}
\E\!\left[\widehat{\mathrm{se}}^2\right] = \sum_{g \in
\mathcal{G}^{\ast}}
\left(\dfrac{\tilde{n}_g^{\underline{\rho}_g}}{\tilde{n}^{\ast}}\right)^{2}
\!\left(\Var\!\left[\hat{\lambda}_g\right] +
\E\!\left[\hat{\lambda}_g\right]^2\right) -
\dfrac{2}{\abs{\mathcal{G}^{\ast}}^2} \sum_{g \in \mathcal{G}^{\ast}}
\sum_{\substack{g^{\prime} \in \mathcal{G}^{\ast} \\ g^{\prime} > g}}
\dfrac{w_g^2 w_{g^{\prime}}^2 \E[\hat{\lambda}_g]\,
\E[\hat{\lambda}_{g^{\prime}}]}{\sqrt{S_w - w_g^2}\sqrt{S_w -
w_{g^{\prime}}^2}}.
\end{equation}
\normalsize

Subtracting the true variance in \eqref{eq: true var} from \eqref{eq: E
var est}, the $\Var[\hat{\lambda}_g]$ terms cancel, leaving the remainder
\footnotesize
\begin{equation} \label{eq: remainder}
R \coloneqq \E\!\left[\widehat{\mathrm{se}}^2\right] -
\Var\!\left[\hat{\tau}^{\mathrm{tilt}}\right] = \sum_{g \in
\mathcal{G}^{\ast}}
\left(\dfrac{\tilde{n}_g^{\underline{\rho}_g}}{\tilde{n}^{\ast}}\right)^{2}
\E\!\left[\hat{\lambda}_g\right]^2 - \dfrac{2}{\abs{\mathcal{G}^{\ast}}^2}
\sum_{g \in \mathcal{G}^{\ast}} \sum_{\substack{g^{\prime} \in
\mathcal{G}^{\ast} \\ g^{\prime} > g}} \dfrac{w_g^2 w_{g^{\prime}}^2
\E[\hat{\lambda}_g]\,
\E[\hat{\lambda}_{g^{\prime}}]}{\sqrt{S_w - w_g^2}\sqrt{S_w -
w_{g^{\prime}}^2}}.
\end{equation}
\normalsize

It remains to show that $R \geq 0$. By the AM-GM inequality, for any
$g \neq g^{\prime}$ in $\mathcal{G}^{\ast}$,
\begin{equation*}
\dfrac{w_g^2\, \E[\hat{\lambda}_g]}{\sqrt{S_w - w_g^2}} \cdot
\dfrac{w_{g^{\prime}}^2\,
\E[\hat{\lambda}_{g^{\prime}}]}{\sqrt{S_w - w_{g^{\prime}}^2}} \leq
\dfrac{1}{2}\left[\dfrac{w_g^4\,
\E[\hat{\lambda}_g]^2}{S_w - w_g^2} + \dfrac{w_{g^{\prime}}^4\,
\E[\hat{\lambda}_{g^{\prime}}]^2}{S_w - w_{g^{\prime}}^2}\right].
\end{equation*}
Summing over all $\abs{\mathcal{G}^{\ast}}(\abs{\mathcal{G}^{\ast}} -
1)/2$ distinct pairs and noting that each index $g$ appears in
$\abs{\mathcal{G}^{\ast}} - 1$ pairs, the double sum in \eqref{eq:
remainder} is bounded above:
\begin{equation} \label{eq: AM-GM bound}
\sum_{g \in \mathcal{G}^{\ast}} \sum_{g^{\prime} \in \mathcal{G}^{\ast} :\, g^{\prime} > g} \dfrac{w_g^2 w_{g^{\prime}}^2
\E[\hat{\lambda}_g]\,
\E[\hat{\lambda}_{g^{\prime}}]}{\sqrt{S_w - w_g^2}\sqrt{S_w -
w_{g^{\prime}}^2}} \leq \dfrac{\abs{\mathcal{G}^{\ast}} - 1}{2} \sum_{g
\in \mathcal{G}^{\ast}} \dfrac{w_g^4\,
\E[\hat{\lambda}_g]^2}{S_w - w_g^2}.
\end{equation}

It therefore suffices to show that
\begin{equation} \label{eq: sufficient ineq}
\sum_{g \in \mathcal{G}^{\ast}}
\left(\dfrac{\tilde{n}_g^{\underline{\rho}_g}}{\tilde{n}^{\ast}}\right)^{2}
\E\!\left[\hat{\lambda}_g\right]^2 \geq
\dfrac{\abs{\mathcal{G}^{\ast}} - 1}{\abs{\mathcal{G}^{\ast}}^2} \sum_{g
\in \mathcal{G}^{\ast}} \dfrac{w_g^4\,
\E[\hat{\lambda}_g]^2}{S_w - w_g^2}.
\end{equation}
To simplify the right side, note that $S_w - w_g^2 = \sum_{g^{\prime} \neq
g} w_{g^{\prime}}^2$, and substituting $w_g =
\abs{\mathcal{G}^{\ast}}\,\tilde{n}_g^{\underline{\rho}_g}/\tilde{n}^{\ast}$
throughout gives $w_g^4/\abs{\mathcal{G}^{\ast}}^2 =
\abs{\mathcal{G}^{\ast}}^2
(\tilde{n}_g^{\underline{\rho}_g})^4/(\tilde{n}^{\ast})^4$ and $S_w -
w_g^2 = \abs{\mathcal{G}^{\ast}}^2 \sum_{g^{\prime} \neq g}
(\tilde{n}_{g^{\prime}}^{\underline{\rho}_{g^{\prime}}})^2/(\tilde{n}^{\ast})^2$.
After cancellation, \eqref{eq: sufficient ineq} reduces to
\begin{equation} \label{eq: reduced ineq}
\sum_{g \in \mathcal{G}^{\ast}}
\left(\tilde{n}_g^{\underline{\rho}_g}\right)^{2}
\E\!\left[\hat{\lambda}_g\right]^2 \left(\sum_{g^{\prime} \neq g}
\left(\tilde{n}_{g^{\prime}}^{\underline{\rho}_{g^{\prime}}}\right)^{2}\right)
\geq 0.
\end{equation}
Every factor on the left side of \eqref{eq: reduced ineq} is non-negative:
$(\tilde{n}_g^{\underline{\rho}_g})^2 \geq 0$ and
$\E[\hat{\lambda}_g]^2 \geq 0$ hold trivially, and
$\sum_{g^{\prime} \neq g}
(\tilde{n}_{g^{\prime}}^{\underline{\rho}_{g^{\prime}}})^2 > 0$ because
$\abs{\mathcal{G}^{\ast}} \geq 2$ ensures that at least one other stratum
exists. The inequality therefore holds, and $R \geq 0$.
\end{proof}

\begin{rmk}[Consistent upper bound on the standard error]
\label{rem: consistent se}
\Cref{prop: consv var} is a finite-sample result: For any fixed
$\abs{\mathcal{G}^{\ast}} \geq 2$, the expected value of
$\widehat{\mathrm{se}}^2$ is at least as large as the true variance,
regardless of stratum sizes or the number of strata. The standard error
estimator $\widehat{\mathrm{se}} \coloneqq
(\widehat{\mathrm{se}}^2)^{1/2}$ --- the square root of the variance
estimator --- enters the denominator of the test statistic
in Section~5.2 of the manuscript. For the asymptotic validity of
this test, we require that the ratio
$\widehat{\mathrm{se}}^2 / \Var[\hat{\tau}^{\mathrm{tilt}}]$ converges in
probability to a limit that is at least~1, which in turn implies that
$\widehat{\mathrm{se}}$ consistently upper-bounds the true standard
error $\Var[\hat{\tau}^{\mathrm{tilt}}]^{1/2}$. The relevant
asymptotic regime has the number of informative strata
$\abs{\mathcal{G}^{\ast}} \to \infty$ while stratum sizes
$\tilde{n}_g^{\underline{\rho}_g}$ remain bounded --- reflecting the fine
exact stratification in Section~6 of the manuscript, in which many small
strata are formed from combinations of spatial, temporal, and contextual
covariates. This regime contrasts with the one invoked for consistency of
the augmented Difference-in-Means
$\hat{\tau}_g^{\underline{\rho}_g}$ in Section~4 of the manuscript, where
the number of strata is held fixed and the number of encounters within
each stratum grows. \citet{fogarty2018a} establishes the asymptotic
consistency of the HC2 leverage adjustment for the
special case in which each stratum contains exactly one treated unit. Our
setting involves post-stratified designs with arbitrary numbers of treated
and control units per stratum, but the same argument applies: The leverage
adjustment ensures that the ratio
$\widehat{\mathrm{se}}^2 / \Var[\hat{\tau}^{\mathrm{tilt}}]$ converges in
probability to a limit at least as large as 1 under the many-strata
regime, so that $\widehat{\mathrm{se}}$ consistently upper-bounds the true
standard deviation.
\end{rmk}

\begin{rmk}[Asymptotically valid inference] \label{rem: clt inference}
The test statistic described in Section~5.2 of the manuscript is
$\hat{\tau}^{\mathrm{tilt}} / \widehat{\mathrm{se}}$. Under the
many-strata regime of \Cref{rem: consistent se}, its null distribution is
stochastically dominated by the standard normal. The argument proceeds in
three steps.

First, a Lindeberg-type central limit theorem ensures that centering the
tilted statistic at its expectation and dividing by its true standard
error yields a quantity that converges in distribution to a standard
normal:
\begin{equation*}
\dfrac{\hat{\tau}^{\mathrm{tilt}} -
\E[\hat{\tau}^{\mathrm{tilt}}]}{\Var[\hat{\tau}^{\mathrm{tilt}}]^{1/2}}
\xrightarrow{d} N(0,1).
\end{equation*}

Second, Proposition~2 of the manuscript establishes that
$\E[\hat{\tau}^{\mathrm{tilt}}] \leq 0$ under the null for an
upper-tailed test. Since subtracting a nonpositive number adds a
nonnegative quantity, the centered numerator
$\hat{\tau}^{\mathrm{tilt}} - \E[\hat{\tau}^{\mathrm{tilt}}]$ is weakly
larger than $\hat{\tau}^{\mathrm{tilt}}$ alone. The uncentered ratio
$\hat{\tau}^{\mathrm{tilt}} /
\Var[\hat{\tau}^{\mathrm{tilt}}]^{1/2}$ is therefore weakly smaller than
the centered version, and hence stochastically dominated by the standard
normal.

Third, \Cref{prop: consv var} and the consistency result in \Cref{rem:
consistent se} together imply that $\widehat{\mathrm{se}} \geq
\Var[\hat{\tau}^{\mathrm{tilt}}]^{1/2}$ in the limit. Replacing the true
standard deviation with the larger $\widehat{\mathrm{se}}$ in the
denominator makes the ratio weakly smaller still, preserving the
stochastic domination.

Combining all three steps, the $p$-value computed from the standard normal
reference distribution is conservative, and the hypothesis test controls
the Type~I error rate at the nominal level asymptotically, for all
potential outcome configurations consistent with the null and all
assignment mechanisms consistent with the $\Gamma$-bound.
\end{rmk}

\section{The NYPD UF-250 Stop, Question, and Frisk Report Worksheet} \label{supp: uf-250}

The NYPD's UF-250 form, which is reproduced in \Cref{fig: UF250}, records administrative data for each stop conducted under the SQF program. Information on the form includes the officer's description of the encounter circumstances, the civilian's demographic characteristics, and whether force was used.

\begin{figure}[H]
\centering
\includegraphics[width=0.731\linewidth]{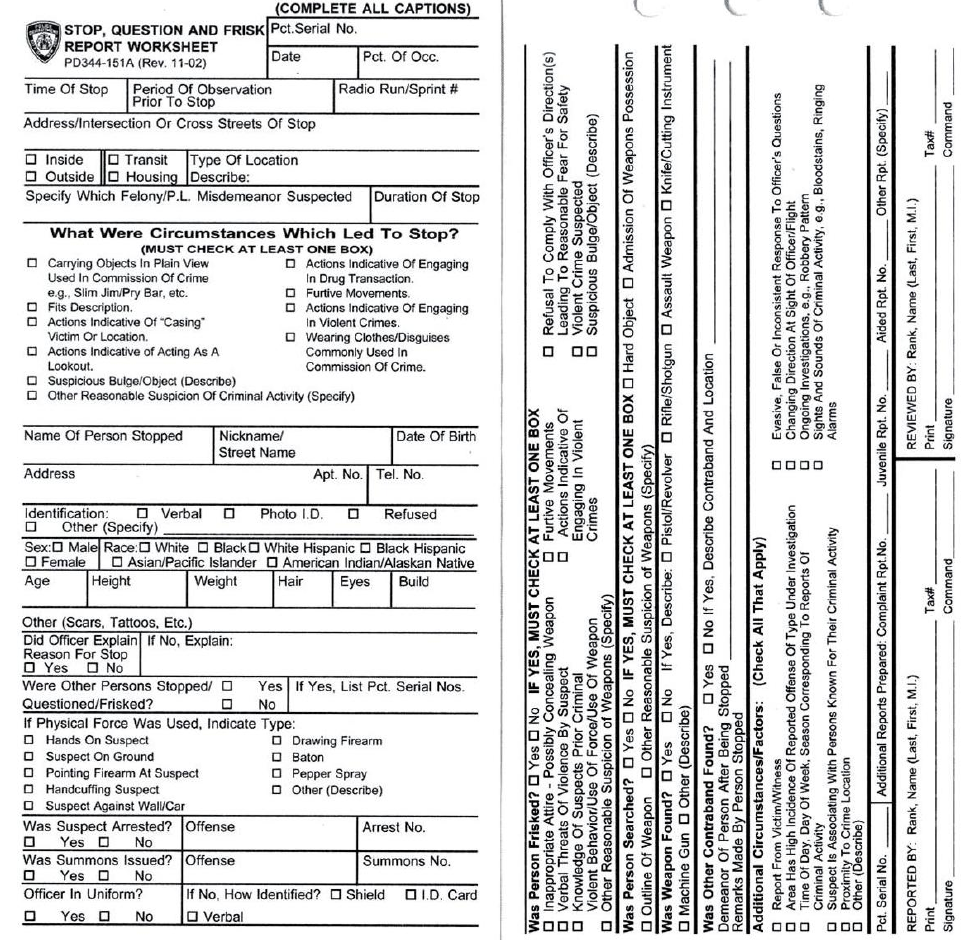}
\caption{The NYPD UF-250 Stop, Question, and Frisk Report Worksheet (form PD344-151A, Rev.\ 11-02), the administrative worksheet officers complete for each stop conducted under the SQF program. Blank copies of the form were entered into evidence as exhibits in \textit{Floyd v.\ City of New York}, No.\ 08 Civ.\ 01034 (S.D.N.Y.), including as part of the October 2010 expert report of Jeffrey Fagan.}
\label{fig: UF250}
\end{figure}

\section{Construction of Impact Zone Geographic Boundaries} \label{supp: impact zone georef}

Beginning in January 2003, the NYPD's Operation Impact concentrated large numbers of officers in small, high-crime areas designated as Impact Zones \citep{macdonaldetal2016}. Impact Zones were located in predominantly minority neighborhoods and had substantially higher stop rates than surrounding areas, so an encounter inside a zone was both more likely to be with a minority civilian and subject to a different policing context than an encounter outside one. To ensure that encounters inside and outside these zones are not compared directly, we include impact zone membership as a post-stratification variable (\Cref{sec: application}).

No official GIS boundaries for these zones have been released publicly. The replication archive that accompanies \citet{macdonaldetal2016}, hosted at \url{https://github.com/macdonaldjohn/Impact-Zone-Data}, contains Stata do-files for the published models together with block-group-month-year panels of crime and arrest counts for 2004--2012; zone membership appears in these panels only as precomputed binary indicators keyed to anonymized block-group identifiers, with no spatial coordinates. The only published spatial representation of the zones is Figure~1 of \citet{macdonaldetal2016}: fifteen small panels, one per zone (Impact Zones III--XVII, activated between January 2004 and January 2012), each showing dark polygons on a schematic map of the five New York City boroughs. We reconstructed zone boundaries from this figure using image processing and georeferencing.

The remainder of this section describes the pipeline and the accuracy of the resulting classification. To reproduce the entire pipeline, run \texttt{make impact-zones} from the replication archive; to classify SQF encounters by zone membership, run \texttt{make classify-zones}; to run the test suite, run \texttt{make impact-zones-test}. The pipeline requires a Python environment (dependencies in \texttt{Impact\_Zones/pyproject.toml}) and R (packages managed via \texttt{renv}). The individual script names mentioned below serve as a guide to the archive for readers who want to inspect specific steps.

\subsection{Image segmentation and panel alignment}

The source figure is a $2{,}250 \times 1{,}586$ pixel raster image arranged as a $5 \times 3$ grid of panels. We split it into fifteen individual panels using ImageMagick (\texttt{split\_fig1.sh} in the replication archive). Because text labels below each panel differ in height across rows, the NYC map sits at different vertical positions across panels---up to 57 pixels of shift between rows, with an additional horizontal drift of roughly 0.4 pixels per column. We corrected these offsets using phase-correlation image registration (OpenCV), registering every panel to the pixel grid of panel~00 (\texttt{align\_panels.py}). After alignment, a single set of ground control points placed on one panel applies to all fifteen.

\subsection{Ground control points and thin-plate spline georeferencing}

We placed 43 ground control points (GCPs) by hand using the QGIS Georeferencer, matching distinctive NYC shoreline landmarks --- Battery Park, Inwood, Throgs Neck, Red Hook, and others --- visible in the raster image to corresponding points on the NYC borough boundary shapefile. The shapefile uses the EPSG:2263 coordinate reference system (NAD83 / NY Long Island, US survey feet). A script (\texttt{candidate\_landmarks.py}) generates a point layer of recommended shoreline targets snapped to the borough boundary, making the GCP selection reproducible.

We then fitted a thin-plate spline (TPS) transform from pixel coordinates to geographic coordinates. The TPS is the standard method for mapping between coordinate systems using landmark data, and handles the local distortions present in schematic maps without imposing a rigid global model such as an affine or polynomial transform \citep[for technical details, see][]{bookstein1989}.

We selected the TPS smoothing parameter via leave-one-out cross-validation (LOOCV). Without smoothing, the LOOCV median error is 471~feet with a maximum of 1{,}824~feet and a systematic eastward displacement in upper Manhattan caused by overfitting to GCP placement noise. At the chosen smoothing parameter value of 5{,}000, the median drops to 351~feet (approximately 1.1 pixels), the maximum drops to approximately 1{,}350~feet (4.4 pixels), and the systematic displacement disappears.

\subsection{Polygon extraction}

We extracted the dark polygons representing impact zones from each aligned panel using a two-layer intensity thresholding strategy (\texttt{georef\_extract.py}). The first layer (threshold at pixel intensity 200) captures the main zone polygons, whose intensities range from roughly 86 to 170. The second layer (threshold at intensity 220) captures thinner features that trace street boundaries at intensities between 200 and 220, while excluding pixels near the coastline (intensity above 245), pixels near NYC borough boundary lines (identified by projecting the borough shapefile into pixel space via the inverse TPS and dilating the resulting mask), and a small noise region in the northeast corner of the image.

The main pipeline script (\texttt{georef\_pipeline.py}) orchestrates the full workflow: loading each aligned panel, calling the extraction and TPS modules, transforming pixel contours to geographic coordinates, clipping all polygons to the NYC land boundary to remove artifacts in water, and dissolving polygons by zone. The output is a GeoJSON file containing 346 polygons in WGS84 (EPSG:4326), each attributed with zone name, activation date, and area.

\subsection{Classification of SQF encounters}

Each SQF encounter record includes $x$/$y$ coordinates in EPSG:2263 for years 2006 onward; coordinates are unavailable for 2003--2005. We classify encounters into three tiers (\texttt{classify\_impact\_zones.R}). First, we spatially join encounters with valid coordinates against the georeferenced zone polygons and classify an encounter as ``inside'' an impact zone only if it falls within a zone that had been activated on or before the date of the encounter. Second, we classify encounters without coordinates as ``outside'' when their precincts have zero spatial overlap with any zone polygon. Third, remaining encounters---those with missing coordinates in precincts that partially overlap a zone---receive a missing value and enter the post-stratification through a missingness indicator.

The same script assigns 2010 Census block, block group, and tract identifiers to all encounters with valid coordinates by spatial join against Census block shapefiles (TIGER/Line). Block group and tract identifiers are derived from the 15-digit block GEOID. The geographic calibration of $\Gamma$ described in \Cref{sec: interpretation} uses these census variables.

\subsection{Sources of positional uncertainty} \label{supp: georef accuracy}

Two sources of error contribute to the positional uncertainty of the reconstructed zone boundaries. Because these errors compound, we assess each separately and then jointly.

\paragraph{Source image resolution (${\sim}307$ feet per pixel).} This is the dominant constraint. The source figure was designed as a schematic illustration, not a GIS product. Each pixel spans approximately 307~feet (94~meters). For context, a typical Manhattan cross-avenue block is roughly 250--300~feet, so one pixel corresponds to approximately one short city block. A Manhattan cross-street block (the longer dimension) is roughly 750--900~feet, or about 2.5--3 pixels. Zone boundaries in the source image are drawn at a width of one to two pixels, so the true boundary could lie anywhere within a band roughly 300--600~feet wide. No downstream processing can recover geographic detail that the source image does not contain.

\paragraph{TPS transform error (median 351 feet, maximum ${\sim}1{,}350$ feet).} The LOOCV measures how accurately the transform maps pixel locations to geographic coordinates. The median error (351~feet, approximately 1.1 pixels) is comparable to the pixel resolution itself, meaning the transform adds little additional uncertainty beyond what the source image already imposes. The maximum error (approximately 1{,}350~feet, or 4.4 pixels) occurs at a small number of GCPs in areas where the schematic map is most distorted relative to true geography. In NYC terms, 1{,}350~feet is roughly five short blocks or 1.5 long blocks.

\paragraph{Combined positional uncertainty.} In the worst case, boundary resolution and transform error compound: a boundary pixel could be mislocated by up to approximately 1{,}650~feet (roughly 500~meters, or about six short Manhattan blocks). In the typical case---median TPS error plus half a pixel of boundary ambiguity---positional uncertainty is roughly 500~feet (about 150~meters, or two short blocks). Encounters well inside a zone, many blocks from any boundary, are classified correctly regardless of this uncertainty. Encounters well outside all zones are likewise unaffected.

\paragraph{Where the error matters and where it does not.} The classification is binary---inside versus outside a zone at the time of the encounter---so error matters only for encounters near a zone boundary. Within the roughly 500-foot band around each boundary, the classification is genuinely ambiguous: an encounter one block inside the drawn boundary might in truth lie one block outside it, or vice versa. This ambiguity is an irreducible consequence of reconstructing boundaries from a schematic figure rather than from an official GIS layer that does not exist.

The analysis handles this uncertainty in two ways. First, encounters with missing coordinates (primarily 2003--2005) are never classified by spatial join; those in precincts that partially overlap a zone receive a missing value and enter the post-stratification through a separate missingness indicator, so they do not contaminate the inside/outside distinction. Second, impact zone membership enters the analysis only as a post-stratification covariate---it defines strata, not the civilian-race indicator or the outcome. A small rate of misclassification near zone boundaries produces slightly coarser strata, merging a few near-boundary encounters into a neighboring stratum, rather than biasing the estimand. The sensitivity analysis then operates on top of these strata.

\subsection{Visual verification and independent quality check}

We verified the pipeline's output by two complementary methods. First, for each of the fifteen panels, we generated side-by-side comparison images (\texttt{compare\_zones.py}) showing the original panel alongside the extracted polygons rendered on a synthetic NYC background (\Cref{fig: compare panel tps}). Second, we produced a diagnostic overlay (\texttt{check\_gcps.py}) that projects the NYC borough boundary shapefile back onto the raster image via the inverse TPS and marks the GCP locations (\Cref{fig: diagnostic tps overlay}). The red boundary lines trace the gray coastline in the image to within one to two pixels across the five boroughs.

\begin{figure}[H]
\centering
\includegraphics[width=0.85\linewidth]{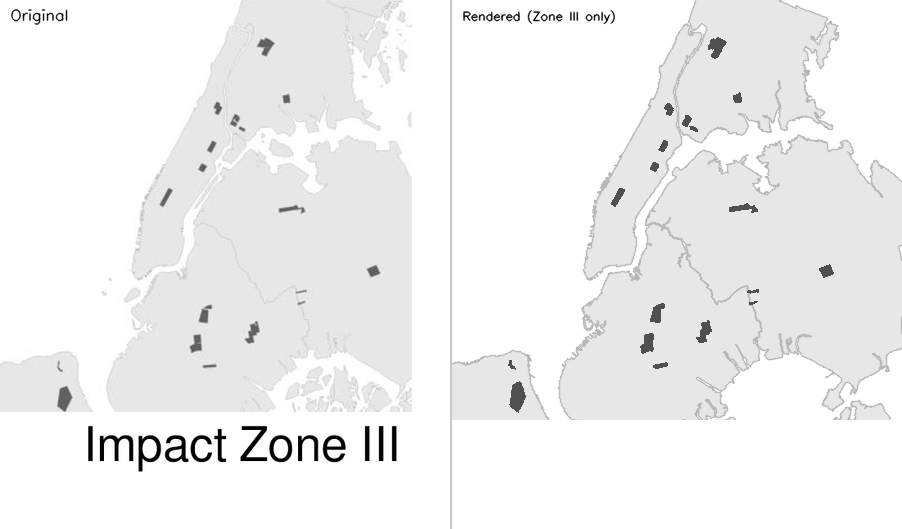}
\caption{Side-by-side comparison for Impact Zone III (panel 00). Left: original panel from \citet{macdonaldetal2016}, Figure~1. Right: extracted polygons rendered on the NYC borough outline using the thin-plate spline transform. Dark regions in the right panel correspond to the georeferenced zone boundaries used in the analysis.}
\label{fig: compare panel tps}
\end{figure}

\begin{figure}[H]
\centering
\includegraphics[width=0.5\linewidth]{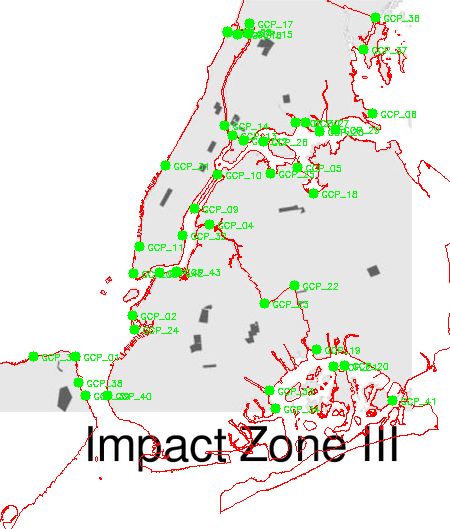}
\caption{Diagnostic overlay for the thin-plate spline transform. Red lines show the NYC borough boundary shapefile projected back onto the raster image via the inverse TPS. Green dots mark the 43 ground control points. The close agreement between the red lines and the gray coastline in the image confirms that the transform aligns geographic and pixel coordinates to within one to two pixels across the five boroughs.}
\label{fig: diagnostic tps overlay}
\end{figure}

As an independent quality check, we also reconstructed impact zone boundaries from 2010 Census block groups using the precinct and neighborhood identifiers in the replication data of \citet{macdonaldetal2016data}. We do not use this block-group reconstruction in the main analysis because the direct georeferencing approach---manual GCP placement with leave-one-out cross-validation---provides transparent, quantified error at each step. The block-group reconstruction serves as a consistency check: the two approaches produce zone boundaries that agree to within the resolution of the source image.

\subsection{Zone activation dates}

\Cref{tab: zone activation} reports the activation date for each of the fifteen impact zones depicted in Figure~1 of \citet{macdonaldetal2016}. Dates are taken from the Stata replication code of \citet{macdonaldetal2016data}; they are not printed on the figure itself. Once a zone is activated, we treat it as an impact zone for all subsequent encounters in the SQF data window, so each encounter is classified as inside a zone whenever it falls within a polygon that had been activated on or before the encounter date.

One caveat applies at the early end of the SQF data window. Impact Zones~I and~II, which were active in 2003, are not depicted in Figure~1 of \citet{macdonaldetal2016} and are not reconstructed here; the MacDonald replication panels begin in 2004 and do not cover 2003. SQF encounters in 2003 therefore receive no zone classification from this pipeline and enter the post-stratification through a missingness indicator.

\begin{table}[H]
\centering
\begin{tabular}{lll}
\toprule
Panel & Zone & Active from \\
\midrule
00 & Impact Zone III   & January 2004          \\
01 & Impact Zone IV    & January 2005          \\
02 & Impact Zone V     & July 2005             \\
03 & Impact Zone VI    & January 2006          \\
04 & Impact Zone VII   & June 2006             \\
05 & Impact Zone VIII  & January 2007          \\
06 & Impact Zone IX    & July 2007             \\
07 & Impact Zone X     & January 2008          \\
08 & Impact Zone XI    & July 2008             \\
09 & Impact Zone XII   & January 2009          \\
10 & Impact Zone XIII  & July 2009             \\
11 & Impact Zone XIV   & January 2010          \\
12 & Impact Zone XV    & August 2010           \\
13 & Impact Zone XVI   & January / August 2011 \\
14 & Impact Zone XVII  & January 2012          \\
\bottomrule
\end{tabular}
\caption{NYPD Impact Zone activation dates. Activation dates are taken from the Stata replication code of \citet{macdonaldetal2016data}. Once a zone is activated, we treat it as an impact zone for all subsequent encounters in the SQF data window.}
\label{tab: zone activation}
\end{table}

\section{Construction of the Geographic Ceiling}\label{sec: supp geo ceiling}

To construct $\Gamma_g^{\mathrm{geo}}$, we follow \citet{zhaoetal2022} in using 2010 Census data to characterize the racial composition of the areas where each stratum's encounters occur. We assign each encounter in the post-stratified data to the Census block group containing its geographic coordinates. For each block group $b$, let $f_b$ denote the minority fraction (Black and Hispanic residents as a share of the total population) and define the minority-to-white population odds $\eta_b \coloneqq f_b / (1 - f_b)$. If patrol were concentrated in block group $b$, the minority encounter odds would roughly reflect $\eta_b$. The residential population of a block group need not match the population an officer actually encounters on patrol --- who is outdoors, at what time, and in what context all matter --- but $\eta_b$ provides a rough benchmark for the demographic structure of the area through which patrols move.

For each stratum $g$, let $\mathcal{B}_g$ denote the set of block groups containing its encounters. To summarize the distribution of $\{\eta_b : b \in \mathcal{B}_g\}$, we compare upper and lower population-weighted quantiles after trimming a fraction $\xi \in [0,0.5)$ from each tail. Let $\eta_g^{(\xi)}$ denote the $\xi$-th population-weighted quantile of $\{\eta_b : b \in \mathcal{B}_g\}$. The geographic ceiling is then defined as $\Gamma_g^{\mathrm{geo}}(\xi) \coloneqq \eta_g^{(1 - \xi)}/\eta_g^{(\xi)}$. The parameter $\xi$ determines how strongly the bound relies on extreme demographic contrasts. When $\xi = 0$, $\Gamma_g^{\mathrm{geo}}(0)$ equals the ratio of the largest and smallest values of $\eta_b$ across block groups in $\mathcal{B}_g$, corresponding to the most extreme scenario in which patrol concentrates entirely in one block group versus another. Without trimming, a single block group with an unusual demographic composition can drive the ceiling even if few encounters occur there. Larger values of $\xi$ exclude these outliers and instead compare more typical locations --- e.g., the $10$th and $90$th percentiles when $\xi = 0.1$. Because these quantiles lie closer together, the resulting ratio is smaller. Thus, larger values of $\xi$ impose tighter ceilings on $\Gamma_g$, reflecting the possibility that patrol spans several locations rather than concentrating exclusively in the most demographically extreme areas.

Encounters within each stratum are already geographically concentrated because the strata condition on precinct and, when available, sector and beat, along with Impact Zone indicators and other spatial covariates. In the data, the median stratum contains encounters in only two block groups. Strata whose encounters all occur in a single block group therefore have $\Gamma_g^{\mathrm{geo}}(\xi) = 1$ by construction, since no within-stratum demographic variation exists. Encounter coordinates are unavailable for 2003--2005, so $\Gamma_g^{\mathrm{geo}}$ cannot be computed directly for strata in those years. Instead, we assign each such stratum the ceiling computed from later-year strata sharing the same geographic identifiers --- precinct, sector, and beat when observed --- so that the bound reflects the demographic variation within the same patrol area.

%

\lstdefinelanguage{Rlang}{
  keywords={library,data,function,if,else,for,in,return,TRUE,FALSE,NULL,
            NA,seq,c,head,min,max,any,sum,filter,group_by,summarise,
            mutate,cat,sprintf,round,sqrt},
  sensitive=true,
  morecomment=[l]{\#},
  morestring=[b]",
  morestring=[b]'
}

\lstset{
  language=Rlang,
  basicstyle=\small\ttfamily\singlespacing,
  keywordstyle=\bfseries,
  commentstyle=\itshape\color{gray},
  stringstyle=\color{darkgray},
  showstringspaces=false,
  breaklines=true,
  breakatwhitespace=true,
  frame=single,
  framesep=3pt,
  xleftmargin=3pt,
  xrightmargin=3pt,
  aboveskip=\medskipamount,
  belowskip=\medskipamount,
  columns=fullflexible,
  keepspaces=true,
  literate={<-}{{$\leftarrow$}}2 {>=}{{$\geq$}}2 {<=}{{$\leq$}}2
}

\section{Replication with the \texttt{jointsens} R Package}
\label{supp: jointsens replication}

The \texttt{jointsens} R package implements the joint sensitivity
analysis developed in this paper. It is available at
\url{https://github.com/XXXX/jointsens} and can be
installed with:

\begin{lstlisting}
# install.packages("remotes")
remotes::install_github("XXXX/jointsens")
\end{lstlisting}

\noindent Below we demonstrate how to replicate the key findings from
Section~6 using the package functions. The analysis requires the
post-stratified SQF dataset (\texttt{poststrat\_sqf\_data.rda}),
which can be reproduced from the replication archive by running
\texttt{make poststrat} in the paper repository. The variable
\texttt{minority} is the civilian-race indicator (1 = Black or Latino,
0 = white), \texttt{force\_any} is the binary outcome (any police
force used), and \texttt{poststratum\_id} identifies the post-strata.

\subsection{Setup and baseline estimate}

\begin{lstlisting}
library(jointsens)
library(dplyr)

load("poststrat_sqf_data.rda")

# Pre-compute per-stratum summary (once, reused for all analyses)
strat_summ <- precompute_strat_summary(
  poststrat_sqf_data,
  treat_var   = minority,
  outcome_var = force_any,
  stratum_var = poststratum_id
)
\end{lstlisting}

\noindent At the baseline ($\underline{\rho} = 0$, $\Gamma = 1$),
the analysis assumes no discrimination in stops and no bias in
encounters. The tilted estimate reduces to the ordinary stratified
difference-in-means:

\begin{lstlisting}
n_oms_baseline <- compute_n_oms_from_rho(
  lb_rho = 0, n0_obs = strat_summ$n0_obs, n1 = strat_summ$n1
)
baseline <- fast_tilted_estimate(
  strat_summ, n_oms = n_oms_baseline,
  Gamma = 1.0, alternative = "greater", tau0 = 0
)

cat("Baseline estimate:", round(baseline$tau_hat * 100, 2), "pp\n")
cat("SE:", round(sqrt(baseline$var_hat), 4), "\n")
# Baseline estimate: 2.30 pp
# SE: 0.0026
\end{lstlisting}

\noindent This matches the baseline result reported in Section~6:
the weighted difference-in-means is approximately
$\BaselineTauPP$ percentage points with a standard error of
$\BaselineSE$.

\subsection{Robustness frontier: the p-value surface}

The central analysis sweeps over a grid of $(\underline{\rho}, \Gamma)$
values, testing $H_0\colon \tau = 0$ against the one-sided alternative
$H_1\colon \tau > 0$ at each point:

\begin{lstlisting}
pval_grid <- fast_grid_sweep(
  strat_summ,
  lb_rho_grid = seq(0, 1, by = 0.05),
  Gamma_grid  = seq(1.0, 1.5, by = 0.001),
  alternative = "greater",
  tau0        = 0
)
head(pval_grid)
\end{lstlisting}

\noindent The output is a data frame with one row per
$(\underline{\rho}, \Gamma)$ combination and columns for the tilted
estimate, variance, test statistic, and upper-tail p-value. This
data frame produces the heatmap in Figure~2 of the main text.

We can verify the key thresholds reported in Section~6. With no
discrimination in stops ($\underline{\rho} = 0$), the finding first
becomes insignificant at $\Gamma = \GammaRhoZeroInsig$:

\begin{lstlisting}
# Smallest Gamma where p > 0.05 at rho = 0
pval_grid |>
  filter(rho_lb == 0, p_upper > 0.05) |>
  summarise(Gamma_star = min(Gamma))
# Gamma_star = 1.06
\end{lstlisting}

\noindent At the other extreme, the smallest $\Gamma$ at which the
test is insignificant for \emph{all} values of $\underline{\rho}$ is
$\Gamma = \GammaAllInsigMin$:

\begin{lstlisting}
# Smallest Gamma where p > 0.05 for ALL rho values
pval_grid |>
  group_by(Gamma) |>
  summarise(all_insig = all(p_upper > 0.05)) |>
  filter(all_insig) |>
  summarise(Gamma_all_insig = min(Gamma))
# Gamma_all_insig = 1.33
\end{lstlisting}

\subsection{Confidence sets at plausible $\underline{\rho}$}

For the plausible range $\underline{\rho} \in \{0.32, 0.34\}$ (derived
from \citealt{knoxetal2020}), we construct 95\% confidence
sets by inverting the two-sided test across a grid of null values
$\tau_0$:

\begin{lstlisting}
cs_result <- fast_conf_set_sweep(
  strat_summ,
  lb_rho_grid = c(0.32, 0.34),
  Gamma_grid  = seq(1.0, 1.5, by = 0.0001),
  tau_grid    = seq(-0.2, 0.4, by = 0.0001),
  alpha       = 0.05,
  validate    = TRUE
)
\end{lstlisting}

\noindent The confidence set at each $(\underline{\rho}, \Gamma)$ is
the set of $\tau_0$ values not rejected at level $\alpha/2$ in either
tail:

\begin{lstlisting}
alpha <- 0.05

conf_sets <- cs_result |>
  group_by(rho_lb, Gamma) |>
  summarise(
    ci_low  = min(tau0[p_upper >= alpha/2 & p_lower >= alpha/2]),
    ci_high = max(tau0[p_upper >= alpha/2 & p_lower >= alpha/2]),
    tau_HL  = median(tau0[p_upper >= alpha/2 & p_lower >= alpha/2]),
    .groups = "drop"
  )
\end{lstlisting}

\noindent The $\Gamma$ at which the confidence interval first includes
zero---the ``changepoint''---can be extracted directly:

\begin{lstlisting}
conf_sets |>
  group_by(rho_lb) |>
  filter(ci_low <= 0, ci_high >= 0) |>
  summarise(changepoint = min(Gamma))
# rho_lb = 0.32: changepoint = 1.32
# rho_lb = 0.34: changepoint = 1.33
\end{lstlisting}

\noindent These match the changepoint values reported in Section~6:
at $\underline{\rho} = 0.32$, the confidence interval first contains
zero at $\Gamma^* = \ChangePointLower$; at $\underline{\rho} = 0.34$,
the changepoint is $\Gamma^* = \ChangePointUpper$. In other words, the
finding of discrimination in force use is robust to encounter-bias odds
ratios up to approximately $1.3$ under plausible levels of
discrimination in stops.

\section*{Disclosure of AI Use}

In accordance with the Taylor \& Francis AI Policy, we disclose the following use of generative AI tools in the preparation of this manuscript. We used Claude (Anthropic; Claude Sonnet 4.5 and Claude Opus 4.7, accessed via the Claude.ai web interface and the Claude Code command-line tool, 2026) and Codex (OpenAI; GPT 5.5 using the Codex command-line tool) for two purposes.

First, we used Claude Code to assist with the reproducible analysis pipeline. The core statistical code --- including the tilted IPW test statistic, the conservative HC2-style variance estimator, and the sequential sensitivity analysis procedure --- was originally written by the authors without AI assistance. Claude Code was subsequently used to improve computational efficiency and code organization including the generation of unit tests. For the reconstruction of NYPD Impact Zone boundaries described in Supplement Section~S.8,  Claude Code assisted with implementation of the image processing, thin-plate spline georeferencing, polygon extraction, and spatial classification of SQF encounters. All code was reviewed, tested, and validated by the authors, and numerical results reported in the manuscript were verified against independent implementations and visual diagnostics.

Second, we used Claude for editing and revision of manuscript prose, including the main text, the abstract, and the presentation of formal results in the supplement. This use included suggestions on sentence-level clarity, word choice, and structural organization such as suggestions about how to condense the writing to fit into the page limits of the journal. The mathematical content of the framework, including the derivation of all formal results, was developed by the authors without AI assistance. The authors reviewed and edited all AI-suggested revisions and remain fully accountable for the content of the manuscript. No AI tool was used to generate or manipulate research data, analytical results, or figures; to derive the formal framework; or to select the estimand or assumptions that organize the paper.

\end{document}